\documentclass[12pt]{article}
\usepackage[utf8]{inputenc}
\usepackage[sectionbib]{natbib}

\title{\vspace{-2cm}Improving time series estimation and prediction via transfer learning}
\author{Yuchang Lin, Qianqian Zhu and Guodong Li\thanks{Address for correspondence: Qianqian Zhu, School of Statistics and Data Science, Shanghai University of Finance and Economics, Shanghai, China. Email: zhu.qianqian@mail.shufe.edu.cn}  \\ \vspace{-0.3cm}\textit{Shanghai University of Finance and Economics and University of Hong Kong}}

\usepackage[margin=1in]{geometry}
\usepackage{graphicx}
\usepackage{mathrsfs}
\usepackage{multirow}
\usepackage{amsfonts}
\usepackage{comment}
\usepackage{amsthm}
\usepackage{xr}
\usepackage{xcolor}
\usepackage{mathtools}
\usepackage{algorithm}
\usepackage{sectsty}
\usepackage[colorlinks]{hyperref}
\usepackage{amsmath}
\usepackage{amssymb}
\usepackage{stmaryrd}
\usepackage[toc,page]{appendix}
\usepackage[mathscr]{euscript}
\usepackage[toc]{appendix}

\usepackage{chngcntr}
\usepackage{apptools}
\usepackage{booktabs}
\usepackage{lscape}
\usepackage{soul}
\usepackage{epstopdf}
\usepackage{array}
\usepackage{diagbox}
\usepackage{mathabx}
\usepackage{algpseudocode}
\usepackage{setspace}
\numberwithin{equation}{section}

\newcolumntype{L}[1]{>{\raggedright\let\newline\\\arraybackslash\hspace{0pt}}m{#1}}
\newcolumntype{C}[1]{>{\centering\let\newline\\\arraybackslash\hspace{0pt}}m{#1}}
\newcolumntype{R}[1]{>{\raggedleft\let\newline\\\arraybackslash\hspace{0pt}}m{#1}}

\makeatletter
\newcommand*{\addFileDependency}[1]{% argument=file name and extension
  \typeout{(#1)}
  \@addtofilelist{#1}
  \IfFileExists{#1}{}{\typeout{No file #1.}}
}
\makeatother

\newtheorem{assum}{Assumption}
\newtheorem{definition}{Definition}

\newtheorem{lemma}{Lemma}
\newtheorem{proposition}{Proposition}
\newtheorem{theorem}{Theorem}

\newtheorem{remark}{Remark}
\newtheorem{claim}{Claim}

\hypersetup{
    colorlinks=true,
    linkcolor=blue,
    citecolor=blue,
    filecolor=magenta,
    urlcolor=cyan,
}

\DeclareMathOperator*{\argmin}{arg\,min}

\newcommand{\bm}{\mathbf}
\newcommand{\bbm}{\boldsymbol}

\newcommand{\rrb}{\rrbracket}
\newcommand{\llb}{\llbracket}
\newcommand{\cm}[1]{\mbox{\boldmath$\mathscr{#1}$}}

\newcommand{\op}{\mathrm{op}}
\newcommand{\F}{\mathrm{F}}
\newcommand{\tr}{\mathrm{tr}}

\newcommand{\prox}{\mathrm{prox}}
\newcommand{\vect}{\mathrm{vec}}

\mathtoolsset{showonlyrefs}

\usepackage{tikz}
\usetikzlibrary{shapes.geometric, arrows}

\tikzstyle{startstop} = [rectangle, rounded corners, minimum width=3cm, minimum height=1cm,text centered, draw=black, fill=purple!30]
\tikzstyle{stop} = [rectangle, rounded corners, minimum width=3.6cm, minimum height=1cm, draw=black, fill=purple!30, text width = 3.6cm]
\tikzstyle{start} = [rectangle, rounded corners, minimum width=2.7cm, minimum height=1cm, draw=black, fill=purple!30, text width = 2.7cm]
\tikzstyle{io} = [trapezium, trapezium left angle=70, trapezium right angle=110, minimum width=3cm, minimum height=1cm, draw=black, fill=blue!30]
\tikzstyle{process} = [rectangle, rounded corners, minimum width=3cm, minimum height=1cm, draw=black, fill=blue!10, text width=10cm]
\tikzstyle{process_new} = [rectangle, rounded corners, minimum width=3cm, minimum height=1cm, draw=black, fill=blue!10, text width=11.5cm]
\tikzstyle{process2} = [rectangle, rounded corners, minimum width=3cm, minimum height=1cm, draw=black, fill=blue!10, text width=10cm]
\tikzstyle{decision} = [diamond, minimum width=3cm, minimum height=1cm, text centered, draw=black, fill=green!30]
\tikzstyle{arrow} = [thick,->,>=stealth]

\DeclareRobustCommand\sampleline[1]{%
    \tikz\draw[#1] (0,0) (0,\the\dimexpr\fontdimen22\textfont2\relax)
    -- (2em,\the\dimexpr\fontdimen22\textfont2\relax);%
}

\begin{document}

\setlength{\parindent}{16pt}

\maketitle

\begin{abstract}
	There are many time series in the literature with high dimension yet limited sample sizes, such as macroeconomic variables, and it is almost impossible to obtain efficient estimation and accurate prediction by using the corresponding datasets themselves.
	This paper fills the gap by introducing a novel representation-based transfer learning framework for vector autoregressive models, and information from related source datasets with rich observations can be leveraged to enhance estimation efficiency through representation learning. 
	A two-stage regularized estimation procedure is proposed with well established non-asymptotic properties, and algorithms with alternating updates are suggested to search for the estimates.
	Our transfer learning framework can handle time series with varying sample sizes and asynchronous starting and/or ending time points, thereby offering remarkable flexibility in integrating information from diverse datasets. 
	Simulation experiments are conducted to evaluate the finite-sample performance of the proposed methodology, and its usefulness is demonstrated by an empirical analysis on 20 macroeconomic variables from Japan and another nine countries.
\end{abstract}


\newpage
\vspace{-1cm}

\linespread{1.55}
\selectfont{}

\section{Introduction}

With the rapid advancement of information technology, high-dimensional data are now ubiquitous, and particularly many of them are time-dependent.
Examples can be found in various fields, including economics and finance \citep{nicholson2020high,wang2024high}, healthcare \citep{DP16,davis2016sparse}, and environmental studies \citep{bahadori2014fast,xu2018muscat}.
These high-dimensional data usually come with large sample sizes, however, the situation is different for macroeconomic variables in the field of economics.
Specifically, the number of variables can be large \citep{stock2009forecasting,Gao_Tsay2022}, while their sample sizes may be limited since they are typically recorded at a lower, say quarterly or even yearly, frequency.
Moreover, it is common for these variables to have asynchronous starting and/or ending time points, while all existing statistical and econometric tools for multivariate or high-dimensional time series require the observations to have uniformly aligned time period.
This leads to an unavoidable data truncation, and hence a much shorter sample size.
It is an important task in the literature to effectively analyze and accurately forecast macroeconomic variables with high dimension yet limited sample sizes.

As an illustrative example, we consider the $20$ quarterly macroeconomic variables from Japan, and the earliest records can trace back to the first quarter of 1955, resulting in 276 observations of the longest variable; see Figure \ref{fig:country_indicator_length} for details.
However, there are only 87 observations for the shortest variable, ending up with the sample size of 87, and it is almost impossible to provide a reliable analysis by using these $20$ macroeconomic variables from Japan themselves.
On the other hand, from Figure \ref{fig:country_indicator_length}, we also have observations for the same $20$ macroeconomic variables but from another nine countries, and they have substantially richer observations; see, for example, the sample size for the dataset of United States is 216.
These macroeconomic variables may have similar interdependence across ten countries although their definitions may even be slightly different.
This raises an interesting question: can we make use of such similarity to leverage information from these data-rich datasets to enhance estimation efficiency and forecasting accuracy for the target dataset of Japan?

\begin{figure}[ht]
	\centering
	\includegraphics[width=0.99\linewidth]{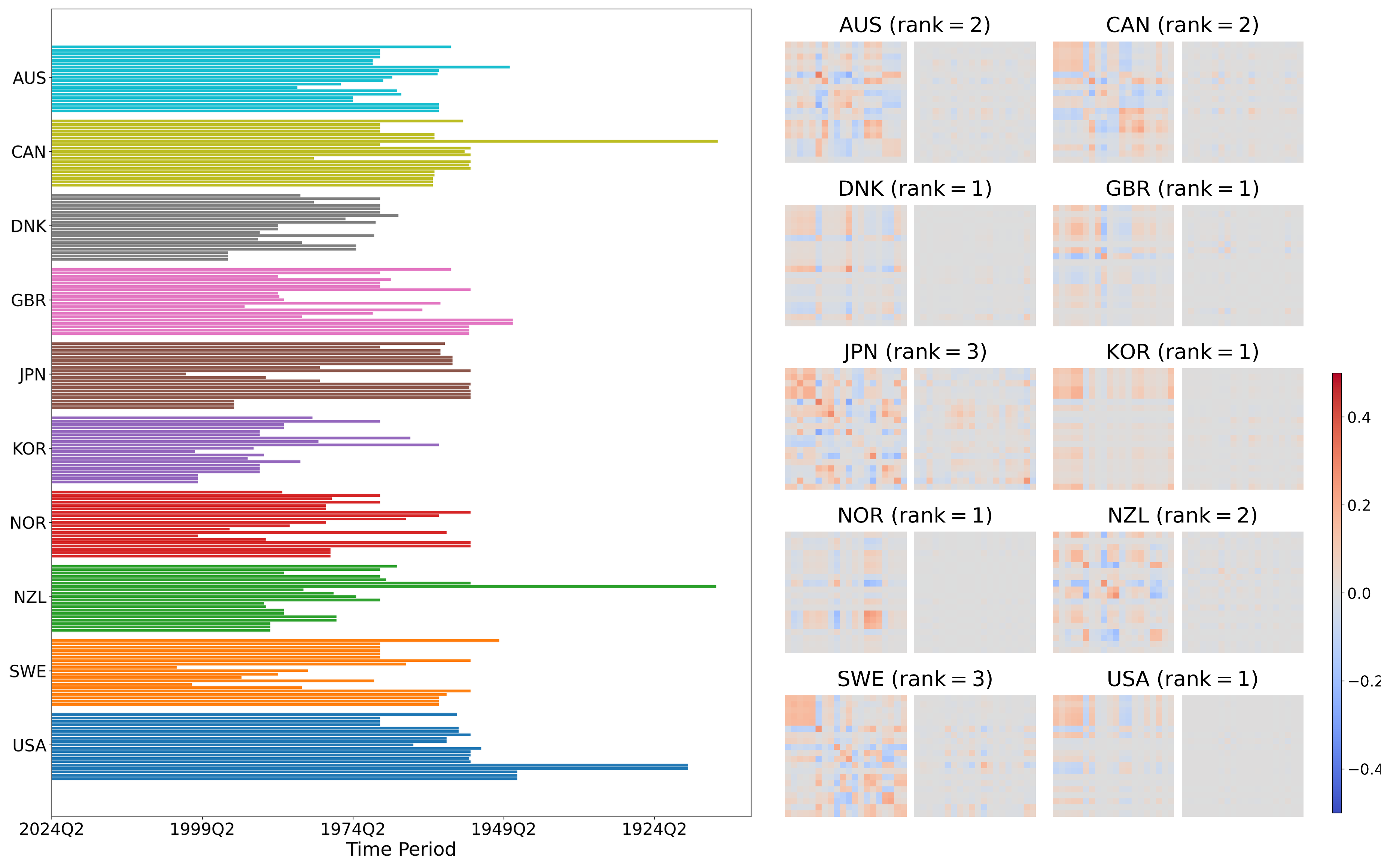}
	\caption{Ending and starting time points for 20 quarterly macroeconomic variables across ten countries (left panel), and heatmaps for projection matrices of response factor spaces before and after removing the common space with rank three (right panel).
	}
	\label{fig:country_indicator_length}
\end{figure} 

%In order to model the large time series data including macroeconomic variables, 

As arguably the most widely used multivariate time series model, the vector autoregression (VAR) has been a primary workhorse for many high-dimensional tasks, including macroeconomic variables; see \cite{basu2015regularized,kock2015oracle,wang2022high}. 
The general VAR model has $O(N^2)$ parameters, where $N$ is the number of variables or sequences, and it hence requires much more sample sizes for a reliable inference.
This limitation becomes more pronounced for the example of macroeconomic variables due to the limited sample size, and it is necessary to first restrict the parameter space such that a reasonable estimation can be achieved.
A direct method in the literature is to assume that the coefficient matrices are sparse and then to apply sparsity-inducing regularized estimation \citep{basu2015regularized, han2015direct, kock2015oracle}. 
This sparse VAR has been widely used in many applications; see, for example, time-course gene expression data, where dependence among genes is believed to be sparse \citep{LZLR09}.
However, for macroeconomic variables, one often observes strong cross-sectional dependence among $N$ scalar series, and it is usually described by assuming that the $N$ variables are driven by a small number of common latent factors \citep{lam2012factor,BW16}.

Another commonly used approach for dimension reduction is to impose low-rank structures to coefficient matrices \citep{negahban2011estimation,basu2019low}, and it will lead to the reduced-rank \citep{velu2013multivariate} and low-multilinear-rank \citep{wang2022high} models when the low-rank assumption is applied to VAR models.
On the one hand, the above low-rank assumption is imposed mainly to reduce the dimension of parameter spaces to $O(N)$. On the other hand, together with the regression setting, it will lead to a nice interpretation of supervised factor modeling \citep{HUANG2025,wang2024high}.
Specifically, we can interpret the low-rank VAR as projecting responses and predictors into a small number of latent factors, termed response and predictor factors, respectively.
The response factors are used to summarize all predictable components of the market, while the predictor factors contain all driving forces.
This well matches the scenario of macroeconomic variables, however, the sample sizes may still not be enough to provide a reliable estimation and accurate prediction, say for the 20 macroeconomic variables from Japan.

Interestingly, when the low-rank VAR model in Section \ref{sec:SupervisedFactor} is applied to the 20 macroeconomic variables from each of ten countries, the response factors are extracted in a similar way across these countries.
Specifically, although the response factor spaces, spanned by the corresponding factor loadings, vary in terms of ranks and directions, they share a common space with a rank as small as three; see Figure \ref{fig:country_indicator_length} for details.
This may be due to the fact that these 20 macroeconomic variables are defined in almost the same way across different countries, and they hence have the similar interdependence.
Moreover, such similarity can also be observed in predictor factors, as well as the temporal ones summarizing the lagged variables; see Section \ref{sec:RealData} for empirical evidences.
As a result, we may consider a new methodology to first estimate the three common spaces by using the data-rich time series from the nine countries, other than Japan, and then to plug them in the VAR modeling on the target time series of Japan, leading to more efficient estimation and more accurate prediction.

The above proposal is related to a huge literature of transfer learning, which seeks to enhance target task performance by strategically transferring knowledge from related source domains.
The concept of transfer learning comes from the machine learning literature \citep{Pan2010ASO}, and it recently has attracted more and more attentions from statisticians.
There are two commonly used ways in the literature to design a transfer learning framework, and they are based on distances and representations, respectively.
The distance-based transfer learning framework measures the similarity between target and source domains by using a distance metric; see, for example, the $\ell_1$-norm or angle-based distance in linear regression \citep{Li2020TLLM,Li2023TLGLM,Gu2022TLangleLM}, and the nuclear-norm in matrix regression \citep{park2025transfer}.
Our problem can be solved by the representation-based transfer learning framework, and it assumes the parameters at target and source domains to share common representations exactly  
\citep{du2020few,Tripuraneni2020diversity,tripuraneni21linear} or loosely with possible task-specific adaptations \citep{chua2021similarRep,duan2023adaptive,chen2025distributed}.
The transfer learning techniques have been widely discussed for independent data with theoretical guarantees, while there are only very few distance-based methods for time series data in the literature \citep{zeng2024TLsptialAR,ma2025transfer}.
To the best of our knowledge, it is still lack of a representation-based transfer learning approach for time series data.

This paper has four main contributions below.  
First, we introduce a transfer learning framework for VAR models at Section \ref{sec:Motivation}, where transition tensors are assumed to have shared representations and task-specific parameters, while similarity measures introduce another layer of flexibility.
This design enables dual advantages: common representations facilitate information transfer from source datasets to enhance target estimation, while data-driven similarity metrics dynamically calibrate source relevance based on approximation accuracy. Moreover, we allow for source and target datasets with asynchronous observation periods.
Secondly, a two-stage estimation procedure is suggested at Section \ref{subsec:Estimation}, and their non-asymptotic properties are established at Section \ref{subsec:Theory}.
Thirdly, due to the nonconvex optimization nature, Section \ref{sec:Algorithms} provides algorithms with alternating updates to search for estimates, and the initialization and hyperparameter selection are also carefully discussed.
Finally, the proposed methodology is applied to the example of macroeconomic variables in Section \ref{sec:RealData}, and prediction improvement can be observed significantly for the 20 variables. 

In addition, the proposed methodology is extended to handle the case with numerous source datasets at Section \ref{sec:LargeK}, and its finite-sample performance is evaluated at Sections \ref{sec:Simulation}. 
A short conclusion and discussion is given at Section \ref{sec:Conclusions}, and all technical proofs are delayed to Appendix.
Throughout this paper, we denote vectors by boldface lower case letters, e.g., $\bm{a}$, matrices by boldface capital letters, e.g., $\bm{A}$, and tensors by Euler script letters, e.g., $\cm{A}$. For a vector $\bm{a}$, denote by $\|\bm{a}\|_2$ its Euclidean norm. 
For a matrix $\bm{A}$, denote by $\bm{A}^\top$, $\|\bm{A}\|_\F$, $\sigma_i(\bm{A})$ and $\mathcal{M}(\bm{A})$ its transpose, Frobenius norm, $i$-th largest singular value, and the space spanned by its columns, respectively. 
When $\bm{A}$ is symmetric, we further denote by $\lambda_{\max}(\bm{A})$ and $\lambda_{\min}(\bm{A})$ its largest and smallest eigenvalue, respectively. 
For two matrices $\bm{A}_1$ and $\bm{A}_2$, we denote their Kronecker product by
$\bm A_1 \otimes \bm A_2$. 
For positive integers $p\geq q$, denote the set of column orthonormal matrices by $\mathbb{O}^{p\times q}:=\{\bm{A}\in\mathbb{R}^{p\times q}:\bm{A}^\top\bm{A}=\bm{I}_q\}$. 
For tensors $\cm{A}$ and $\cm{B}$, denote by $\|\cm{A}\|_\textup{F}$ the Frobenius norm, by $\cm{A}_{(i)}$ the mode-$i$ matricization, and by $\langle\cm{A},\cm{B}\rangle$ the generalized inner product; see the Appendix \ref{sec:prelim} for more tensor introduction and algebra preliminaries. 
%Let $C,C_1,C_2,\ldots$ denote generic positive constants, which may vary across different contexts. 
For two real-valued sequences $\{x_k\}$ and $\{y_k\}$, $x_k\gtrsim y_k$ or $x_k\lesssim y_k$ if there exists a $C>0$ such that $x_k\geq Cy_k$ or $x_k\leq Cy_k$ for all $k$, respectively. 
In addition, we write $x_k\asymp y_k$ if $x_k\gtrsim y_k$ and $y_k\gtrsim x_k$.
The dataset in Section \ref{sec:RealData} and computer programs for the analysis are available at \url{https://github.com/LinyuchangSufe/TLVAR}. 

\section{Model settings}\label{sec:Motivation}

\subsection{Low-rank VAR and supervised factor modeling} \label{sec:SupervisedFactor}
For $N$-dimensional time series $\{\bm y_{t}\}$, we consider a vector autoregressive (VAR) model below,
\begin{equation}\label{eq:VAR}
	\bm y_{t} = \bm A_{1}\bm y_{t-1} +\cdots \bm A_{p} \bm y_{t-p} + \bbm\varepsilon_{t},
\end{equation}
where $\bm A_{j} \in \mathbb{R}^{N\times N}$ with $1\leq j\leq p$ are transition matrices, and $\{\bbm \varepsilon_{t}\}$ are independent and identically distributed ($i.i.d.$) with zero mean and finite covariance matrix $\bm \Sigma_{\bbm \varepsilon}\in \mathbb{R}^{N\times N}$. 
We first rearrange the transition matrices into a tensor $\cm{A}\in \mathbb{R}^{N\times N\times p}$ such that its mode-1 matricization is $\cm{A}_{(1)} = (\bm A_{1}, \ldots, \bm A_{p})\in \mathbb{R}^{N\times Np}$.
Denote $\bm Z_{t} = (\bm y_{t-1}, \ldots, \bm y_{t-p})\in\mathbb{R}^{N\times p}$ and $\bm x_{t} =(\bm y_{t-1}^\top, \ldots, \bm y_{t-p}^\top)^\top\in\mathbb{R}^{Np}$. Note that $\bm x_{t} = \vect{(\bm Z_t)} $, and model \eqref{eq:VAR} can be rewritten into
\begin{equation}\label{eq:VARptensor}
	\bm y_{t} =\langle\cm{A}, \bm Z_{t}\rangle + \bbm \varepsilon_{t}  =\cm{A}_{(1)} \bm x_{t} + \bbm \varepsilon_{t}.
\end{equation}

We impose a low-rank constraint to the transition tensor $\cm{A}$ as in \cite{wang2022high}. Suppose that it has multilinear low ranks $(r_{1}, r_{2}, r_{3})$, i.e., $\text{rank}(\cm{A}_{(j)})=r_{j}$ for $1\leq j \leq 3$, and then there exists a Tucker decomposition \citep{de2000multilinear},  
\begin{equation}\label{eq:HOSVD}
	\cm{A} = \cm{S} \times_1 \bm U \times_2 \bm V \times_3 \bm L=: \llb \cm{S} ;\bm U,\bm V,\bm L\rrb,
\end{equation} 
where $\cm{S} \in \mathbb{R}^{r_{1}\times r_{2} \times r_{3}}$ is the core tensor, and $\bm U \in \mathbb{R}^{N\times r_{1}}, \bm V \in \mathbb{R}^{N\times r_{2}}$ and $\bm L \in \mathbb{R}^{p\times r_{3}}$ are factor matrices. 
The above decomposition is not unique, since 
$\cm{A}=\cm{S} \times_1 \bm{U} \times_2 \bm{V} \times_3 \bm{L} =\left(\cm{S}\times_1 \bm{O}_1 \times_2 \bm{O}_2 \times_3 \bm{O}_3\right) \times_1\left(\bm{U} \bm{O}_1^{-1}\right) \times_2\left(\bm{V} \bm{O}_2^{-1}\right) \times_3\left(\bm{L} \bm{O}_3^{-1}\right)$ for any invertible matrices $\bm{O}_i \in \mathbb{R}^{r_i \times r_i}$ with $1\leq i \leq 3$.
Without loss of generality, we restrict the three factor matrices to be orthonormal, i.e., $\bm U \in \mathbb{O}^{N\times r_{1}}, \bm V \in \mathbb{O}^{N\times r_{2}}$ and $\bm L \in \mathbb{O}^{p\times r_{3}}$; see, for example, the high-order singular value decomposition (HOSVD) in \cite{kolda2009tensor}.
Note that $\bm U, \bm V$ and $\bm L$ are still not unique, while their spanned spaces, $\mathcal{M}(\bm U), \mathcal{M}(\bm V)$ and $\mathcal{M}(\bm L)$, along with the corresponding projection matrices, $\bm U \bm U^\top, \bm V \bm V^\top$ and $\bm L \bm L^\top$, can be uniquely defined.

Denote $\cm{G} = \cm{S} \times_3 \bm L \in \mathbb{R}^{r_1\times r_2 \times p}$, and let $\bm G_j \in \mathbb{R}^{r_1 \times r_2}$ be its $j$-th frontal slice for $1\leq j \leq p$. It holds that $\cm{G}_{(1)} = (\bm G_{1}, \ldots, \bm G_{p})\in \mathbb{R}^{r_1\times r_1p}$, and $\cm{A} = \cm{G} \times_1 \bm U \times_2 \bm V$. As a result, $\bm A_j = \bm U \bm G_j \bm V^\top$ with $1\leq j \leq p$, and model \eqref{eq:VAR} or \eqref{eq:VARptensor} can be rewritten into
\begin{equation}\label{eq:MLRVAR}
	\bm y_{t} = \sum_{j=1}^p \bm U \bm G_j \bm V^\top \bm y_{t-j} + \bbm \varepsilon_{t}
	\hspace{3mm}\text{or}\hspace{3mm}
	\bm U^\top\bm y_{t} = \sum_{j=1}^p  \bm G_j \bm V^\top \bm y_{t-j} + \bm U^\top\bbm \varepsilon_{t}, 
\end{equation}
which can be interpreted from supervised factor modeling perspectives \citep{HUANG2025}. 
Specifically, we first project responses $\bm y_{t}$ onto subspace $\mathcal{M} (\bm U)$ and its orthogonal complement, $\mathcal{M}^\perp(\bm U)$, i.e. $\bm y_{t}= \mathbf{U}\mathbf{U}^{\top}\mathbf{y}_{t}+(\mathbf{I}_{N} - \mathbf{U}\mathbf{U}^{\top})\mathbf{y}_{t}$, and these two parts can be verified to have completely different dynamic structures,
\begin{equation}
	\mathbf{U}\mathbf{U}^{\top}\mathbf{y}_{t} = \sum_{j=1}^p \mathbf{U}\bm G_j \bm V^\top \bm y_{t-j}+ \mathbf{U}\mathbf{U}^{\top}\boldsymbol{\varepsilon}_{t} \quad \text{and} \quad
	(\mathbf{I}_{N} - \mathbf{U}\mathbf{U}^{\top})\mathbf{y}_{t} = (\mathbf{I}_{N} - \mathbf{U}\mathbf{U}^{\top})\boldsymbol{\varepsilon}_{t},
\end{equation}
where all information of $\mathbf{y}_{t}$ related to temporally dependent structures is contained in $\mathcal{M} (\bm U)$, whereas $\mathcal{M}^\perp(\bm U)$ includes only purely idiosyncratic and serially independent components.
Consequently, we call $\mathbf{U}^{\top}\mathbf{y}_{t}\in \mathbb{R}^{r_1}$ the \textit{response factor} since it summarizes all predictable components in responses, and accordingly $\mathcal{M} (\bm U)$ is referred to the \textit{response factor space}.

On the other hand, for the dimension reduction on predictors, we project $\bm y_{t-j}$ onto $\mathcal{M}(\bm V)$, leading to $\bm V\bm V^\top \bm y_{t-j}$, for each $1\leq j \leq p$.
Note that, to measure the relationship between random variables $\bm X$ and $\bm Y$ after removing the effects of a random vector $\bm Z$, we can use the partial covariance function, defined as $\mathrm{pcov}(\bm X,\bm Y| \bm Z)=\mathrm{cov}\{\bm X-\mathbb{E}(\bm X|\bm Z),\,\bm Y-\mathbb{E}(\bm Y|\bm Z)\}$. 
As a result, when $\mathbb{E}\|\bm y_t\|_2^2<\infty$, we can verify that
$$ \mathrm{pcov}(\bm y_t,\bm y_{t-l}\mid \bm V\bm V^\top\bm y_{t-j}, 1\leq j\leq p)=\bm 0 \hspace{3mm}\text{for all $l\geq 1$}, $$
i.e., the space $\mathcal{M}(\bm V)$ captures all information of $\bm y_{t-j}$ that contributes to forecasting $\bm y_t$, or $\bm V^\top \bm y_{t-j}$ contains all driving forces of the market.
Thus, we call $\bm V^\top \bm y_{t}\in \mathbb{R}^{r_2}$ the \textit{predictor factor} for simplicity, and $\mathcal{M}(\bm V)$ can be referred to the \textit{predictor factor space}. 

Similarly, we can conduct the dimension reduction in the temporal direction. Specifically, denote $\cm{H} = \cm{S} \times_1 \bm U \times_2 \bm V \in \mathbb{R}^{N\times N \times r_3}$, and let $\bm H_k\in\mathbb{R}^{N\times N}$ be its $k$-th frontal slice for $1\leq k \leq r_3$. It then holds that $\cm{H}_{(1)} = (\bm H_{1}, \ldots, \bm H_{r_3})\in \mathbb{R}^{N\times Nr_3}$, and $\cm{A}=\cm{H}\times_3 \bm L$.
Correspondingly, model \eqref{eq:VAR} or \eqref{eq:VARptensor} can be rewritten into
\[
\bm y_{t} =\langle \cm{H},\bm Z_{t}\bm L \rangle+ \bbm \varepsilon_{t}= \sum_{k=1}^{r_3} \bm H_k\bm z_{k,t} + \bbm \varepsilon_{t},
\]
where $\bm Z_{t} = (\bm y_{t-1}, \ldots, \bm y_{t-p})\in\mathbb{R}^{N\times p}$, $\bm Z_{t}\bm L=(\bm z_{1,t},\ldots,\bm z_{r_3,t})^{\top}\in\mathbb{R}^{N\times r_3}$, and $N$-dimensional vectors $\bm z_{k,t}$'s are linear combinations of $\bm y_{t-j}$ with $1\leq j\leq p$.
Therefore, we call $\bm L^{\top}\bm Z_{t}^{\top}\in \mathbb{R}^{r_3\times N}$ the \textit{temporal factor}, and $\mathcal{M}(\bm L)$ can be referred to the \textit{temporal factor space}. 

This paper concentrates on the VAR model at  \eqref{eq:VAR} or \eqref{eq:VARptensor}, with the low-rank constraint at \eqref{eq:HOSVD}.
Note that matrices $\bm U, \bm V$ and $\bm L$ serve as factor loadings to extract the three types of factors, respectively, and they can be referred to the representations of  coefficient tensor $\cm{A}$.
When the model is applied to the example of macroeconomic variables, it ends up with similar  response, predictor and temporal representations across ten countries; see Figure \ref{fig:country_indicator_length} and Section \ref{sec:RealData} for details.
This may be due to the fact that these macroeconomic variables have similar interdependence across different countries. % although their definitions are slightly different. 
As a result, this paper will introduce a transfer learning framework for VAR models by making use of such similarity.

\subsection{Transfer learning framework} \label{sec:ProblemSetup}

Consider $K+1$ $N$-dimensional time series, one target time series $\{\bm y_{t,0}\}$ and $K$ source time series $\{\bm y_{t,k}\}$ with $k\in [K]=\{1,\ldots, K\}$, and each of them is assumed to follow the VAR model with low-rank constraints in Section \ref{sec:SupervisedFactor}.
Specifically, for each $k\in \{0\} \cup [K]$, 
\begin{equation}\label{eq:VARp}
	\bm y_{t,k} =\cm{A}_{k(1)} \bm x_{t,k} + \bbm \varepsilon_{t,k} = \bm A_{1,k}\bm y_{t-1,k} +\cdots +\bm A_{p,k} \bm y_{t-p,k} + \bbm \varepsilon_{t,k}, 
\end{equation}
where $\cm{A}_k\in\mathbb{R}^{N\times N\times p}$ is the transition tensor with $\bm A_{j,k} \in \mathbb{R}^{N\times N}$ for $1\leq j\leq p$ and its mode-1 matricization $\cm{A}_{k(1)}=(\bm A_{1,k},\ldots,\bm A_{p,k})\in\mathbb{R}^{N\times Np}$, 
$\bm x_{t,k} =(\bm y_{t-1,k}^\top, \ldots, \bm y_{t-p,k}^\top)^\top\in\mathbb{R}^{Np}$,
and $\{\bbm \varepsilon_{t,k}\}$ are $i.i.d.$ with zero mean and finite covariance matrix $\bm \Sigma_{\bbm \varepsilon,k}\in \mathbb{R}^{N\times N}$.
Moreover, the transition tensor $\cm{A}_k$ has multilinear low ranks $(r_{1,k}, r_{2,k}, r_{3,k})$, i.e. $\text{rank}(\cm{A}_{k(j)})=r_{j,k}$ for $1\leq j\leq 3$, and we then have Tucker decomposition,
$\cm{A}_k =\llb \cm{S}_k ;\bm U_k,\bm V_k,\bm L_k\rrb$,
where $\cm{S}_k \in \mathbb{R}^{r_{1,k}\times r_{2,k} \times r_{3,k}}$ is the core tensor, and $\bm U_k \in \mathbb{O}^{N\times r_{1,k}}, \bm V_k \in \mathbb{O}^{N\times r_{2,k}}$ and $\bm L_k \in \mathbb{O}^{p\times r_{3,k}}$ are orthonormal factor matrices.
Finally, we assume that, across all target and source time series, the error terms are  independent, and the order $p$ is the same without loss of generality.

For a transfer learning problem, it is crucial to formalize the similarity between target and source domains or time series, and we can then leverage auxiliary information from the source time series to enhance estimation efficiency at the target one.
This paper assumes the similarity of response, predictor and temporal representations, respectively.
Specifically, there exist three common representations, $\bm U \in \mathbb{O}^{N\times s_1}, \bm V \in \mathbb{O}^{N\times s_2}$ and $\bm L \in \mathbb{O}^{p\times s_3}$, such that $\mathcal{M}(\bm U)=\mathcal{M}(\bm U_k,0\leq k\leq K)$, $\mathcal{M}(\bm V)=\mathcal{M}(\bm V_k,0\leq k\leq K)$, and $\mathcal{M}(\bm L)=\mathcal{M}(\bm L_k,0\leq k\leq K)$, where $s_j\geq \max\{r_{j,k},0\leq k\leq K\}$ with $1\leq j\leq 3$, and they are ranks of the common response, predictor and temporal representations, respectively.
As a result, for each $k\in \{0\} \cup [K]$, we have $\bm U_k = \bm U \bm O_{1,k}, \bm V_k = \bm V \bm O_{2,k}$ and $\bm L_k = \bm L \bm O_{3,k}$, where $\bm O_{i,k} \in \mathbb{O}^{s_{i} \times r_{i,k}}$ with $1\leq i\leq 3$ are rotation matrices, and then the transition tensor has an unified decomposition,
\begin{equation}\label{eq:add1}
	\cm{A}_k = \llb \cm{S}_k ;\bm U \bm O_{1,k},\bm V \bm O_{2,k},\bm L \bm O_{3,k}\rrb = \llb \cm{D}_k ;\bm U ,\bm V ,\bm L\rrb,
\end{equation}
where $\cm{D}_k = \llb \cm{S}_k ;\bm O_{1,k} , \bm O_{2,k} , \bm O_{3,k}\rrb \in \mathbb{R}^{s_1 \times s_2 \times s_3}$ varies among different VAR modeling tasks. It is noteworthy to point out that the above decomposition is also not unique.

In the meanwhile, the constraint at \eqref{eq:add1} may be restrictive in real applications, and some flexibility is needed to accommodate possible deviations.
To this end, we first introduce two similarity measures, $h_{\mathrm{S}}\geq 0$ and $h_{\mathrm{T}}\ge 0$, for the departures of transition tensors from the low-rank structure at \eqref{eq:add1} in the source and target domains, respectively.
It is then assumed that there exist three matrices $\bm U \in \mathbb{O}^{N\times s_1}, \bm V \in \mathbb{O}^{N\times s_2}$ and $\bm L \in \mathbb{O}^{p\times s_3}$ such that  
\begin{equation}\label{eq:similaritymeasure}
	\max_{k\in [K]} \min_{\cm{D}_k } \| \cm{A}_k - \llb   \cm{D}_k; \bm U, \bm V, \bm L \rrb \|_\F \le h_{\mathrm{S}} \quad\text{and}\quad \min_{\cm{D}_0 } \| \cm{A}_0 - \llb   \cm{D}_0; \bm U, \bm V, \bm L \rrb \|_\F \le h_{\mathrm{T}}.
\end{equation}
Smaller values of $h_{\mathrm{S}}$ and $h_{\mathrm{T}}$ imply stronger representational similarity, and the case with $h_{\mathrm{S}}=h_{\mathrm{T}}=0$ corresponds to an exact multilinear low-rank structure with common representations of $\bm U$, $\bm V$ and $\bm L$ at \eqref{eq:add1}.
When $h_{\mathrm{S}}>0$ and $h_{\mathrm{T}}>0$, it allows the transition tensors $\cm{A}_k$'s to deviate from this restrict structure, and hence more flexibility; see \cite{duan2023adaptive}.

\begin{remark}\label{remark:similarity}
	From \eqref{eq:similaritymeasure}, we have $\cm{A}_k = \llb \cm{D}_k; \bm U, \bm V, \bm L \rrb + \cm{R}_k$ with $\cm{R}_k\in \mathbb{R}^{N\times N\times p}$, and it holds that $\|\cm{R}_k\|_{\F}\leq h_{\mathrm{S}}$ for $k\in [K]$, and $\|\cm{R}_0\|_{\F}\leq h_{\mathrm{T}}$.
	The deviation tensor $\cm{R}_k$ may be low-rank to accommodate task-specific  representations, and it may also be sparse for heterogeneous signals.
	The framework in this paper is more general than the representation-based transfer learning that assumes fully shared representations across all tasks \citep{Tripuraneni2020diversity}, making it much more flexible for high-dimensional time series. 
\end{remark}

This paper employs different similarity measures for the sources and target domains and, specifically, we may consider the case with $h_{\mathrm{S}}<h_{\mathrm{T}}$.
As a result, source information can be aggressively leveraged under tighter similarity constraints, while target dynamics can be better estimated through a relaxed constraint. 
In addition, the framework at \eqref{eq:similaritymeasure} separates shared latent representations, i.e. $\bm U$, $\bm V$ and $\bm L$, from task-specific adaptations, $\cm{D}_k$, and this enables two advantages: source domain representation learning and efficient target adaptation.
Specifically, it can improve estimation efficiency of $\bm U$, $\bm V$ and $\bm L$ by aggregating information from source domains and hence leveraging larger combined sample sizes, and we then transfer these estimated representations to learn the target-specific tensor $\cm{D}_0$, enhancing estimation efficiency of transition tensor $\cm{A}_0$.
This adaptive mechanism, termed the transfer learning for VAR models, formalizes a representation-based transfer learning approach for high-dimensional time series modeling and forecasting. 

\section{High-dimensional estimation}\label{sec:Method}

\subsection{Transfer learning via ordinary least squares method}\label{subsec:Estimation}

Suppose that there are $K+1$ observed $N$-dimensional time series, $\{\bm y_{t,k}, -p+1\leq t\leq T_k\}$ with $k\in \{0\} \cup [K]$, generated by model \eqref{eq:VARp} with the constraint at \eqref{eq:similaritymeasure}.
This section conducts a transfer learning to improve the estimation efficiency and prediction accuracy of the target time series, $\{\bm y_{t,0}, -p+1\leq t\leq T_0\}$, by leveraging the auxiliary information from source time series,  $\{\bm y_{t,k}, -p+1\leq t\leq T_k\}$ with $k\in [K]$. 
These datasets can be temporally unaligned, with heterogeneous starting and/or ending points across domains. This setup enables leveraging cross-domain temporal dependencies while accommodating asynchronous data availability, which is a common situation in macroeconomic and financial applications.

Denote $\bm x_{t,k}=(\bm y_{t-1,k}^\top , \ldots, \bm y_{t-p,k}^\top)^\top\in\mathbb{R}^{Np}$, $\bm X_k = (\bm x_{1,k}, \ldots, \bm x_{T_k,k})\in\mathbb{R}^{Np\times T_k}$, 
$\bm Y_k = (\bm y_{1,k}, \ldots, \bm y_{T_k,k})\in\mathbb{R}^{N\times T_k}$, and
$\bm E_k = (\bbm{\varepsilon}_{1,k}, \ldots, \bbm{\varepsilon}_{T_k,k})\in\mathbb{R}^{N\times T_k}$.  
It holds that, from model \eqref{eq:VARp}, 
$\bm Y_k=\cm{A}_{k(1)}\bm X_k+\bm E_k$. We consider the ordinary least squares (OLS) estimation, and the corresponding loss functions can be defined below, for each $k\in \{0\} \cup [K]$
\[
\mathcal{L}_k(\cm{A}_k) = \frac{1}{2T_k}\sum_{t=1}^{T_k}\|\bm y_{t,k} - \sum_{j=1}^p\bm A_{j,k}\bm y_{t-j,k} \|_\F^2 = \frac{1}{2T_k}\left\|\bm Y_k - \cm{A}_{k(1)}\bm X_k\right\|_\F^2.
\]
We next introduce a two-stage regularized estimation procedure for the transfer learning.

\begin{itemize}
\item \textit{Stage I: representation learning}. This stage is to extract common representations from source datasets.
Specifically, given the order $p$ and  ranks $s_1$, $s_2$ and $s_3$, we leverage source datasets to jointly estimate these shared representations $\bm U, \bm V$ and $\bm L$ and task-specific tensors $\{\cm{D}_k\}$ with $k\in [K]$. 
\begin{equation}\label{eq:MTLloss}
	(\widehat{\bm U}, \widehat{\bm V}, \widehat{\bm L},\{\widehat{\cm{A}}_k,\widehat{\cm{D}}_k, k\in [K]\}) \in \argmin
	\sum_{k\in [K]} w_k \{\mathcal{L}_k(\cm{A}_k) + \lambda_k \|\cm{A}_k - \llb   \cm{D}_k; \bm U, \bm V, \bm L \rrb \|_\F\}.
\end{equation}
The dataset-specific weights $\{w_k\}$ are related to the source selection, while the regularization parameter $\lambda_k$ for each $k\in [K]$ forces the transition tensor $\cm{A}_k$ not far from the multilinear low-rank structure at \eqref{eq:add1}.
We delay their selection to Section \ref{subsec:tuning}.

\item \textit{Stage II: transfer learning.}
Based on the estimated representations, $\widehat{\bm U}, \widehat{\bm V}$ and $\widehat{\bm L}$, from \eqref{eq:MTLloss}, this stage implements transfer learning for the target time series by solving the following optimization problem,
\begin{equation}\label{eq:TLloss}
		(\widehat{\cm{A}}_0,\widehat{\cm{D}}_0)
		\in 
		\argmin \{ \mathcal{L}_0(\cm{A}_0) + \lambda_0 \|\cm{A}_0 - \llb \cm{D}_0; \widehat{\bm U}, \widehat{\bm V}, \widehat{\bm L} \rrb \|_\F \},
\end{equation}
where the regularization parameter $\lambda_0$ controls the strength of transfer, and its selection is delayed to Section \ref{subsec:tuning}.
The source datasets are usually data-rich, and hence $\widehat{\bm U}, \widehat{\bm V}$ and $\widehat{\bm L}$ will be much more efficient.
As a result, the estimation efficiency of $\widehat{\cm{A}}_0$ will be improved, and hence the prediction accuracy of the target time series.
\end{itemize}

\begin{remark}
In the proposed transfer learning framework at Section \ref{sec:ProblemSetup}, the similarity is assumed to each of response, predictor and temporal representations, however, it may not be the case at some scenarios. For example, the source and target time series may share similar representations only for response and predictor factors. At this situation, we may keep $\bm U$ and $\bm V$ to be shared, while the common representation $\bm L$ is switched to be task-specific, i.e. $\bm L_k$ with $k\in\{0\}\cup[K]$. The proposed two-stage estimation procedure can be adjusted accordingly. 
\end{remark}

\subsection{Non-asymptotic properties}\label{subsec:Theory}

This subsection establishes theoretical guarantees of the two-stage estimaton procedure in Section \ref{subsec:Estimation}. Consider the VAR model at \eqref{eq:VARp}, and denote its 
matrix polynomial by
$\bm{\Xi}_k(z) = \bm{I}_N - \sum_{i=1}^{p} \bm A_{i,k} z^i$ for each $k \in \{0\} \cup [K]$, where $z \in \mathbb{C}$ with $\mathbb{C}$ being the complex space.

\begin{assum}\label{asmp:1}
	(Stationarity) 
	For each $k\in \{0\} \cup [K]$, the determinant of $\bm{\Xi}_k(z)$ is nonzero for all $|z|<1$ with $z \in \mathbb{C}$. %\vspace{-0.1cm}
\end{assum}

\begin{assum}\label{asmp:2}
	(Sub-Gaussian errors) 
	For each $k\in \{0\} \cup [K]$, 
	let $\bbm{\varepsilon}_{t,k}=\bm{\Sigma}_{\bbm{\varepsilon},k}^{1/2}\bbm{\zeta}_{t,k}$ with $\bm\Sigma_{\bbm \varepsilon,k} = \textup{var}(\bbm{\varepsilon}_{t,k} )$, $\mathbb{E}(\bbm{\zeta}_{t,k})=\bm{0}$ and $\textup{var}(\bbm{\zeta}_{t,k})=\bm{I}_N$.
	Assume that $\bm\Sigma_{\bbm \varepsilon,k}$ are positive definite for all $k\in \{0\} \cup [K]$, $\{\bbm{\zeta}_{t,k}\}$ are $i.i.d.$ across all time points $t$ and $k\in \{0\} \cup [K]$, and the entries of $\bbm{\zeta}_{t,k}$ are mutually independent and $\sigma^2$-sub-Gaussian.
\end{assum}

Assumption \ref{asmp:1} gives a necessary and sufficient condition for the existence of a unique strictly stationary solution to the VAR model at \eqref{eq:VARp}. 
The sub-Gaussian condition in Assumption \ref{asmp:2} is standard in the literature of high-dimensional time series \citep{zheng20,wang2022high,wang2024common}, while the independence assumption across datasets is commonly used in the literature of transfer learning \citep{Li2020TLLM,tripuraneni21linear}. 

Denote by $\bm{A}_{j,k}^*$ and $\cm{A}_k^*$ with $1\leq j\leq p$ and $k \in \{0\} \cup [K]$ the true transition matrices and tensors of model \eqref{eq:VARp}, respectively, and they should not be far from the multilinear low-rank structure at \eqref{eq:add1} such that a successful transfer learning can be achieved. 

\begin{assum}\label{asmp:3}
	(Task-relatedness)
	There exist $\bm U^*, \bm V^*, \bm L^*$ and $\{\cm{D}_{k}^*, k\in \{0\} \cup [K]\}$ such that $\max_{k \in [K]}\| \cm{A}_k^* - \llb   \cm{D}_k^*; \bm U^*, \bm V^*, \bm L^* \rrb \|_\F \le h_{\mathrm{S}}$, $\| \cm{A}_0^* - \llb   \cm{D}_0^*; \bm U^*, \bm V^*, \bm L^* \rrb \|_\F \le h_{\mathrm{T}}$, $\|\cm{D}_k^*\|_\F \le \alpha_1$ for all $k\in \{0\} \cup [K]$, and $\min_{1\leq j\leq 3}\lambda_{\min}(s_j\sum_{k\in[K]}w_k\cm{D}_{k(j)}^* \cm{D}_{k(j)}^{*\top}) \geq \alpha_2 $, where $h_{\mathrm{S}}$, $h_{\mathrm{T}}$, $\alpha_1$ and $\alpha_2$ are four positive numbers.
	
%	\begin{itemize}
%		\item[(i)] 
%		$\ \max_{k\in \{0\} \cup [K]} \| \cm{A}_k^* - \llb   \cm{D}_k^*; \bm U^*, \bm V^*, \bm L^* \rrb \|_\F \le h$, i.e., $\max\{h_{\mathrm{S}},h_{\mathrm{T}}\} \le h$;
%		\item[(ii)] $\max_{k\in \{0\}\cup [K]} \|\cm{D}_k^*\|_\F \le \alpha_1$, and $s_i\sum_{k\in[K]}w_k\cm{D}_{k(i)}^* \cm{D}_{k(i)}^{*\top} \succeq \alpha_2 \bm I_{s_i}$ for $i \in [3]$.
%	\end{itemize}
%$h>0, \alpha_1>0, \alpha_2>0$
\end{assum}

The conditions in Assumption \ref{asmp:3} have been widely used in the literature of linear representation learning \citep{Tripuraneni2020diversity,duan2023adaptive}.  
Particularly, the values of $h_{\mathrm{S}}$ and $h_{\mathrm{T}}$ control the task similarity, i.e., the deviations of transition tensors from the low-rank structure, and a smaller value will lead to greater task similarity across source and target datasets. 
In addition, the task-specific tensors, $\{\cm{D}_k^*,k\in \{0\} \cup [K]\}$, are uniformly bounded, and this actually is mild under Assumption \ref{asmp:1}. 
Moreover, the value of $\alpha_2$ is related to the task diversity, which facilitates the full coverage of representation spaces, $\mathcal{M}(\bm U)$, $\mathcal{M}(\bm V)$, and $\mathcal{M}(\bm L)$, and hence leads to a faster convergence rate than the case without such condition \citep{Tripuraneni2020diversity,duan2023adaptive,tian2024RepLLM}.  

We next establish non-asymptotic properties of the two-stage estimation in Section \ref{subsec:Estimation}, and they rely on the temporal and cross-sectional dependence of both source and target time series \citep{basu2015regularized}.
To this end, for a matrix polynomial $\bm \Xi(z)$, we define
\begin{equation*}
	\mu_{\min}(\bm \Xi) = \min_{|z| = 1} \lambda_{\min}\big(\bm \Xi^\dagger(z)\bm \Xi(z)\big)
	\hspace{3mm}\text{and}\hspace{3mm}
	\mu_{\max}(\bm \Xi) = \max_{|z| = 1} \lambda_{\max}\big(\bm \Xi^\dagger(z)\bm \Xi(z)\big),
\end{equation*}
where  $\bm \Xi^\dagger(z)$ is the conjugate transpose of $\bm \Xi(z)$, and hence the quantities of $\mu_{\min}(\bm \Xi_k) $ and $\mu_{\max}(\bm \Xi_k) $ with $k \in \{0\} \cup [K]$.
Moreover, the VAR$(p)$ model at \eqref{eq:VARp} can be rewritten into an equivalent VAR$(1)$ form with a companion matrix $\bm{\Phi}_k\in\mathbb{R}^{Np\times Np}$ for each $k \in \{0\} \cup [K]$; see the Appendix \ref{sec:VMAinf} for details. As a result, the corresponding matrix polynomial is given by
$\widetilde{\bm{\Xi}}_k(z) = \bm{I}_{Np} - \bm{\Phi}_k z$, and we can similarly define $\mu_{\min}(\widetilde{\bm{\Xi}}_k)$'s and $\mu_{\max}(\widetilde{\bm{\Xi}}_k)$'s.

Let $s_{\max} = \max\{ s_1,s_2,s_3\}$ and $\mu = \alpha_1/ \alpha_2$, and we first present non-asymptotic properties of the representation learning estimation. 
Denote $\rho = \min_{k \in [K]} \{{\lambda_{\min}(\bm \Sigma_{\bbm \varepsilon,k})}/{\mu_{\max}(\bm \Xi_k)}\}$, $L = \max_{k \in [K]} \{ {3\lambda_{\max}(\bm \Sigma_{\bbm \varepsilon,k})}/{\mu_{\min}(\bm \Xi_k)}\}$, $\widetilde{L} = \max_{k \in [K]} \{{3\lambda_{\max}(\bm \Sigma_{\bbm \varepsilon,k})}/{\mu_{\min}(\widetilde{\bm \Xi}_k)}\}$, $\kappa = L / \rho$ and $\widetilde{\kappa} = \widetilde{L} / \rho$.
Since $\bm U$, $\bm V$ and ${\bm L}$ are not unique, we define the corresponding estimation errors below,
\[
\mathrm{dist}\{(\widehat{\bm U}, \widehat{\bm V}, \widehat{\bm L}), (\bm U^*, \bm V^*, \bm L^*)\}=
\max\{\|\sin \Theta (\widehat{\bm U}, \bm U^* )\|_2, \|\sin \Theta (\widehat{\bm V}, \bm V^* )\|_2, \|\sin \Theta (\widehat{\bm L}, \bm L^* )\|_2\},
\]
where  $\Theta(\bm A, \bm B)\in\mathbb{R}^s$ are the $s$ principal angles between the column spaces of $\bm A\in\mathbb{O}^{N\times s}$ and $\bm B\in\mathbb{O}^{N\times s}$, and $\sin (\cdot)$ takes sines to each entry of the vector \citep{stewart1998perturbation}. 
Note that the VAR models at \eqref{eq:VARp} for $K$ source time series, together with the multilinear low-rank structure at \eqref{eq:add1}, have the complexity of $d_{\mathrm{M}} = N(s_1+s_2) + p s_3 + K s_1 s_2 s_3$.

\begin{theorem}[Representation learning]\label{thm:rep}
	Suppose that $ h_{\mathrm{S}} \lesssim L^{-1} \sigma  \sqrt{[N^2p + N\log (N K)]/T_k}$ for all $k\in [K]$, and $T_{k} \gtrsim \max \{\kappa^{2}\sigma^{4}, \kappa \sigma^{2} \}(Np + \log K)$.
	If Assumptions \ref{asmp:1}--\ref{asmp:3} hold and $
	\lambda_k \asymp \kappa^{3/2} \mu^2 s_{\max}\sigma \sqrt{[N^2p + N\log (N K)]/T_k}$, then
	\begin{equation}
		\sum_{k\in [K]} w_k \| \llb   \widehat{\cm{D}}_k; \widehat{\bm U}, \widehat{\bm V}, \widehat{\bm L} \rrb  -\llb   \cm{D}_k^*; \bm U^*, \bm V^*, \bm L^* \rrb   \|_\F^2  \lesssim \kappa^2\alpha_1^2 \left( \frac{\sigma^2}{ L} \cdot \frac{d_{\mathrm{M}}}{\widetilde{T}} + h_{\mathrm{S}}^2\right)
	\end{equation}
	and hence
	\begin{equation}
		\mathrm{dist}\{(\widehat{\bm U}, \widehat{\bm V}, \widehat{\bm L}), (\bm U^*, \bm V^*, \bm L^*)\}
		\lesssim 
		s_{\max}^{1/2}\kappa\mu \left( \frac{ \sigma}{ L^{1/2}} \cdot \sqrt{\frac{d_{\mathrm{M}}}{\widetilde{T}}} + h_{\mathrm{S}}\right)
	\end{equation}
with a probability at least
$1 - \exp\left(- C[Np+\log(NK)]\right) - \exp(-C \min_{k \in [K]} \{T_k\}) $, where $\widetilde{T} = \min_{k\in [K]} \{ w_k^{-1} \lambda_{\max}^{-1}(\bm \Sigma_{\bbm \varepsilon,k})T_k \}$, and $C>0$ is a generic constant.
\end{theorem}

The above theorem requires a not too large deviation of $h_{\mathrm{S}}$ such that the proposed estimation procedure can learn the common representations successfully.
In the meanwhile, we need a penalty of $\lambda_k\gtrsim h_{\mathrm{S}}$ such that the estimates with exact multilinear low-rank structures can be achieved, i.e., $\widehat{\cm{A}}_k=\llb   \widehat{\cm{D}}_k; \widehat{\bm U}, \widehat{\bm V}, \widehat{\bm L} \rrb $ for all $k\in [K]$.
In addition, $\widetilde{T}$ denotes the variance-adjusted effective sample size, and it holds that $\widetilde{T}\asymp \sum_{k \in [K]} T_k$ when we choose $w_j=T_j/(\sum_{k \in [K]} T_k)$ with $1\leq j\leq K$ as in Section \ref{subsec:tuning}.
When the parameters of $\kappa$, $\mu$, $L$, $\sigma$ and $s_{\max}$ are bounded away from zero and infinite, the learned common representations have a convergence rate of $\sqrt{d_{\mathrm{M}}/\widetilde{T}}+h_{\mathrm{S}}$, where $h_{\mathrm{S}}$ is a bias term, and $\sqrt{d_{\mathrm{M}}/\widetilde{T}}$ admits a variance reduction due to the pooled sample size across all $K$ source time series. The deviation of $h_{\mathrm{S}}$ should be smaller than $\sqrt{d_{\mathrm{M}}/\widetilde{T}}$ such that the variance reduction from a large $\widetilde{T}$ is not offset.
We next state non-asymptotic properties of the transfer learning estimation.

\begin{theorem}[Transfer learning]\label{thm:TL}
	Suppose that $ h_{\mathrm{T}} \lesssim  L^{-1} \sigma  \sqrt{[N^2p + N\log (N) ]/T_0}$, and $T_{0} \gtrsim \max\{ \widetilde{\kappa}^{2}\sigma^{4}, \widetilde{\kappa} \sigma^{2} \}Np$. If $
	\lambda_0 \asymp \kappa^{3/2} \mu^2 s_{\max}\sigma\sqrt{[N^2p + N\log (N)]/T_0}$, and the conditions of Theorem \ref{thm:rep} hold, then 
	\begin{equation}
		\|\widehat{\cm{A}}_0-\cm{A}_0^*\|_\F
		\lesssim 
		\underbrace{\frac{s_{\max}^{1/2}\kappa\mu \sigma} {L^{1/2}} \sqrt{\frac{d_{\mathrm{M}}}{\widetilde{T}}}}_{Representation}
		+ \underbrace{\frac{\sigma M_{0}}{\rho}\sqrt{\frac{s_1s_2s_3}{T_{0}}}}_{Transfer} 
		+  \underbrace{s_{\max}\kappa^{5/2}\mu^2  \max\{h_{\mathrm{S}},h_{\mathrm{T}}\} }_{Approximation},
	\end{equation}
    with a probability at least $1 - \exp\left(-C [Np+\log(NK)]\right)-  \exp(-C \min_{k \in [K]} \{T_k\}) - \exp\left(-C T_0\right)$,
	where
	$M_{0} =\lambda_{\max}(\bm \Sigma_{\bbm \varepsilon,0})/\mu_{\min}^{1/2}(\bm \Xi_0)$, and $C>0$ is a generic constant. 
\end{theorem}

When the quantities of $\kappa$, $\mu$, $\sigma$, $M_{0}$ and $s_{\max}$ are bounded away from zero and infinite, the estimation error bound in Theorem \ref{thm:TL} reduces to $\|\widehat{\cm{A}}_0-\cm{A}_0^*\|_\F
\lesssim \sqrt{{d_{\mathrm{M}}}/{\widetilde{T}}} + \sqrt{{s_1s_2s_3}/{T_{0}}} +\max\{h_{\mathrm{S}}, h_{\mathrm{T}}\}$, where the first part comes from the errors in estimating three common representations, $\widehat{\bm U}$, $\widehat{\bm V}$ and $\widehat{\bm L}$, at Stage I, and the second one is generated when we estimate the task-specific parameter tensor $\cm{D}_0\in\mathbb{R}^{s_1\times s_2\times s_3}$ in Stage II.
The third part captures model misspecification due to the low-rank approximation at \eqref{eq:similaritymeasure}.
When $h_{\mathrm{S}}$ and $h_{\mathrm{T}}$ are smaller than $\sqrt{{d_{\mathrm{M}}}/{\widetilde{T}}}$ and $\sqrt{{s_1s_2s_3}/{T_{0}}}$, respectively, we can obtain a pronounced improvement on estimation error bound, i.e., $\|\widehat{\cm{A}}_0-\cm{A}_0^*\|_\F
\lesssim \sqrt{{d_{\mathrm{M}}}/{\widetilde{T}}} + \sqrt{{s_1s_2s_3}/{T_{0}}}$.
Moreover, the proposed transfer leaning framework can achieve further improvements when the source datasets are sufficiently abundant such that we can have a larger value of $\widetilde{T}$.

On the other hand, it is of interest to consider two extreme cases: the worst and exact transfer learning.
For the worst transfer learning, the conditions of Theorem \ref{thm:rep} may not be satisfied; see, for example, the value of $h_{\mathrm{S}}$ is too large. Accordingly, the learned common representations are no longer informative.
The exact transfer learning refers to the case with $h_{\mathrm{S}}=h_{\mathrm{T}}=0$ and $\lambda_k = \infty$ for all $k\in \{0\} \cup [K]$ in the two-stage estimation procedure.

\begin{proposition}[Worst transfer learning]\label{prop:worse}
	Suppose that Assumptions~\ref{asmp:1}--\ref{asmp:2} hold. If $T_0 \gtrsim \max \{ \widetilde{\kappa}^2 \sigma^4, \widetilde{\kappa}\sigma^2 \} Np$ and 
	$ \lambda_0 \asymp \kappa^{3/2} \mu^2 s_{\max}\sigma\sqrt{[N^2p + N\log (N)]/T_0}$, then \[
	\|\widehat{\cm{A}}_0-\cm{A}_0^*\|_\F
	\lesssim
	s_{\max}\kappa^{3/2}\mu^2 M_0\sqrt{\frac{N^2p}{T_{0}}}.
	\]
	with a probability at least $1 - \exp(-C [Np+\log(N)]) - \exp(-C T_0)$.
	%, where $C>0$ is a generic constant.
\end{proposition}

\begin{proposition}[Exact transfer learning]\label{prop:PoolVAR}
	Suppose that Assumptions \ref{asmp:1}--\ref{asmp:3} hold with $h_{\mathrm{S}}=h_{\mathrm{T}}=0$, and $\lambda_k = \infty$ for all $k\in \{0\} \cup [K]$ in the two-stage estimation procedure. If
	$T_{k} \gtrsim  \max\{s_1 ,s_2\} \max \{\widetilde{\kappa}^{2} \sigma^{4}, \widetilde{\kappa} \sigma^{2}\} N$ for all $k\in [K]$, and 
	$T_{0} \gtrsim s_1s_2s_3\max \{ \widetilde{\kappa}^{2}\sigma^{4}, \widetilde{\kappa} \sigma^{2} \}$, then
	\begin{equation*}
		\|\widehat{\cm{A}}_0-\cm{A}_0^*\|_\F
		\lesssim 
		\frac{s_{\max}^{1/2}\kappa\mu \sigma} {L^{1/2}} \sqrt{\frac{d_{\mathrm{M}}}{\widetilde{T}}}
		+ \frac{\sigma M_{0}}{\rho}\sqrt{\frac{s_1s_2s_3}{T_{0}}} 
	\end{equation*}
	with a probability at least $1 - \exp\left(-C [Np+\log(NK)]\right)-  \exp(-C \min_{k \in [K]} \{T_k\}) - \exp\left(-C T_0\right)$.
	%,	where $C>0$ is a generic constant.
\end{proposition}

From Proposition \ref{prop:worse}, when the source time series fail to provide useful auxiliary information, the proposed transfer learning framework can still attain an error rate no greater than that of conventional single-task learning on the target time series, hence avoiding negative transfer. 
Note that, in fact, our framework imposes no structural constraint on coefficient tensors, $\cm{A}_k$'s, and this generality inevitably requests larger sample sizes although the bound at Proposition \ref{prop:PoolVAR} is the same as that at Theorem \ref{thm:TL} with sufficiently small values of $h_{\mathrm{S}}$ and $h_{\mathrm{T}}$.
As in the representation-based transfer learning \citep{Tripuraneni2020diversity,tripuraneni21linear,du2020few}, the exact transfer learning assumes fully shared representations across source and target time series, and hence less sample sizes are needed. 

\section{Implementation Issues}\label{sec:Algorithms}

\subsection{Algorithms}\label{subsec:Algorithm1} 

This subsection introduces algorithms for the two optimization problems in the proposed two-stage estimation procedure in Section \ref{subsec:Estimation}.

We first consider the optimization problem at \eqref{eq:MTLloss} from the representation learning at Stage I. 
Let $\cm{R}_k = \cm{A}_k -\llb \cm{D}_k; \bm U, \bm V, \bm L \rrb$, and it is then equivalent to
\begin{equation}\label{eq:MTLlossform2}
	(\widehat{\bm U}, \widehat{\bm V}, \widehat{\bm L},\{\widehat{\cm{R}}_k,\widehat{\cm{D}}_k, k\in [K]\}) 
	\in 
	\argmin  \sum_{k\in [K]} w_k \{\mathcal{L}_k(\llb   \cm{D}_k; \bm U, \bm V, \bm L \rrb+\cm{R}_k ) + \lambda_k \|\cm{R}_k\|_\F\} .
\end{equation}
The above optimization can be conducted by alternating updates between two blocks, $\{\bm U,\bm V, \bm L, \cm{D}_k, k\in [K] \}$ and $\{\cm{R}_k, k\in [K]\}$, while this block-coordinate structure allows for an efficient algorithm.
Specifically, by fixing the block of $\{\bm U,\bm V, \bm L, \cm{D}_k, k\in [K] \}$, we can update $\{\cm{R}_k, k\in [K]\}$ by solving $K$ separated sub-optimization problems,
\begin{equation}\label{eq:single task}
	\min \mathcal{L}_k(\llb \cm{D}_k^{(i)}; \bm U^{(i)}, \bm V^{(i)}, \bm L^{(i)} \rrb+\cm{R}_k ) + \lambda_k \|\cm{R}_k\|_\F
\end{equation}
for each $k \in [K]$, where $\{\bm U^{(i)},\bm V^{(i)}, \bm L^{(i)}, \cm{D}_k^{(i)}, k\in [K] \}$ are outputs from the $i$th iteration with $i\geq 0$.
The loss function $\mathcal{L}_k(\llb \cm{D}_k^{(i)}; \bm U^{(i)}, \bm V^{(i)}, \bm L^{(i)} \rrb +\cdot)$ at \eqref{eq:single task} is smooth and convex, while Frobenius norm $\|\cdot\|_\F$ is convex but non-differentiable at the origin. 
As a result, the overall objective function is a combination of one smooth loss and one non-smooth penalty term, making the proximal gradient descent method \citep{Parikh2014prox} well-suited here. 
Denote by $\eta_k$ the step size, and by $\prox_{c}\{\cm{A}\}= (1-c/\|\cm{A}\|_\F)_{+} \cm{A}$ the proximal operator, where $\cm{A}$ is a tensor, and $x_+=x$ as $x\geq 0$, and $0$ as $x<0$.
We then can update each $\cm{R}_k$ by
\begin{equation}\label{eq:update1}
	\cm{R}_k^{(i+1)} = \prox_{\eta_k \lambda_k} \{\cm{R}_k^{(i)} - \eta_k \nabla \mathcal{L}_k(\llb   \cm{D}_k^{(i)}; \bm U^{(i)}, \bm V^{(i)}, \bm L^{(i)} \rrb + \cm{R}_k^{(i)}  )\},
\end{equation}
where $\cm{R}_k^{(i)}$ contains outputs from the $i$th iteration.
Step size $\eta_k$ is usually chosen to be the inverse of the Lipschitz constant of $\nabla \mathcal{L}_k(\cdot)$ to ensure convergence \citep{beck2009fast}. As a result, we set it to $\eta_k = T_k \lambda_{\max}^{-1}(\bm X_k \bm X_k^\top)$ in this paper.

On the other hand, by fixing the block of $\{\cm{R}_k, k\in [K]\}$, we can update $\{\bm U,\bm V, \bm L, \cm{D}_k, k\in [K] \}$ by minimizing $\sum_{k\in [K]} w_k \mathcal{L}_k(\llb \cm{D}_k; \bm U, \bm V, \bm L \rrb+\cm{R}_k^{(i+1)} )$, while it is challenging to directly optimize this objective function with orthonormality constraints.
Note that our aim is to search for a multilinear low-rank tensor, $\llb \cm{D}_k; \bm U, \bm V, \bm L \rrb$, rather than its specific decomposition, and hence there is no harm to give up the orthonormality of $\bm U$, $\bm V$ and $\bm L$.
Specifically, we further adjust the loss function by adding three regularization terms,
\begin{equation}\label{eq:poolconstrained}
	\begin{split}
		\mathcal{L}^{\mathrm{RL}}(\bm U, \bm V, \bm L,\cm{D}_k, k\in [K])=&  \frac{a}{4}\left\| \bm U^\top \bm U-b^2\bm I_{s_1}\right\|_\F^2+ \frac{a}{4}\left\| \bm V^\top \bm V-b^2\bm I_{s_2}\right\|_\F^2 + \frac{a}{4}\left\| \bm L^\top \bm L-b^2\bm I_{s_3}\right\|_\F^2\\
		&+\sum_{k\in [K]} w_k \mathcal{L}_k(\llb   \cm{D}_k; \bm U, \bm V, \bm L \rrb+\cm{R}_k^{(i+1)} ),
	\end{split}
\end{equation}
where $a$, $b>0$ are two tuning parameters; see also \cite{han2021optimal}. 
The regularization terms, $\| \bm U^\top \bm U-b^2\bm I_{s_1}\|_\F^2$, $\| \bm V^\top \bm V-b^2\bm I_{s_2}\|_\F^2$ and $\| \bm L^\top \bm L-b^2\bm I_{s_3}\|_\F^2$,
are employed to prevent $\bm{U}$, 
$\bm{V}$, and $\bm{L}$ from being singular, and they can also balance the scaling of $\bm{U}$, $\bm{V}$, $\bm{L}$ and $\cm{D}_k$'s. 
A gradient descent algorithm can be employed to minimize the objective function at \eqref{eq:poolconstrained}.

Denote by $\widetilde{\bm{U}}$, $\widetilde{\bm{V}}$, $\widetilde{\bm{L}}$ and $\cm{\widetilde{D}}_k$'s the minimizers of $\mathcal{L}^{\mathrm{RL}}(\bm U, \bm V, \bm L,\cm{D}_k, k\in [K])$, and it then holds that $\widetilde{\bm{U}}^\top\widetilde{\bm{U}} = b^2\bm{I}_{s_1}$, $\widetilde{\bm{V}}^\top\widetilde{\bm{V}} = b^2\bm{I}_{s_2}$ and $\widetilde{\bm{L}}^\top\widetilde{\bm{L}} = b^2\bm{I}_{s_3}$, i.e., they are also the minimizers of $\sum_{k\in [K]} w_k \mathcal{L}_k(\llb   \cm{D}_k; \bm U, \bm V, \bm L \rrb+\cm{R}_k^{(i+1)} )$.
In fact, if this is not true, then there exist invertible matrices $\bm{O}_i\in\mathbb{R}^{s_i\times s_i}$ such that $\widetilde{\bm{U}} =\bar{\bm{U}}\bm{O}_1$, $\widetilde{\bm{V}} =\bar{\bm{V}}\bm{O}_2$ and $\widetilde{\bm{L}} =\bar{\bm{L}}\bm{O}_3$, where $\bar{\bm{U}}^\top\bar{\bm{U}} = b^2\bm{I}_{s_1}$, $\bar{\bm{V}}^\top\bar{\bm{V}} = b^2\bm{I}_{s_2}$ and $\bar{\bm{L}}^\top\bar{\bm{L}} = b^2\bm{I}_{s_3}$.
Let $\cm{\widebar{D}}_k = \cm{\widetilde{D}}_k\times_1\bm{O}_1\times_2\bm{O}_2\times_3\bm{O}_3$ with $k\in[K]$. Then $\llb   \cm{\widebar{D}}_k; \bar{\bm{U}}, \bar{\bm{V}}, \bar{\bm{L}} \rrb =\llb   \cm{\widetilde{D}}_k; \widetilde{\bm{U}}, \widetilde{\bm{V}}, \widetilde{\bm{L}} \rrb$ for all $k\in[K]$, while the three regularization terms are reduced to zero. This leads to a contradiction with the definition of minimizers.
Moreover, the proposed algorithm is not sensitive to the choice of regularization parameters, $a$ and $b$, in $\mathcal{L}^{\mathrm{RL}}(\bm U, \bm V, \bm L,\cm{D}_k, k\in [K])$, and they are set to one in all our simulation and empirical studies.
Similar regularization methods have been widely applied to non-convex low-rank matrix estimation problems; see \cite{tu2016low}, \cite{wang2017unified} and references therein.

\begin{algorithm}[htp]
	\caption{Alternating Algorithm for Representation Learning at Stage I}\label{algorithm}
	\begin{spacing}{1.0}
		\begin{algorithmic}[1]
			\State \textbf{Input:} $K$ source datasets $\{\bm y_{t,k},-p+1\leq t\leq T_k\}$, weights $\{w_k\}$, regularization parameters $\{\lambda_k\}$, number of iteration $I$.
			\State \textbf{Initialization:} 
			\State \quad Initialize matrices $\bm{U}^{(0)} \in \mathbb{R}^{N \times s_1}$, $\bm{V}^{(0)} \in \mathbb{R}^{N \times s_2}$ and $\bm{L}^{(0)} \in \mathbb{R}^{p \times s_3}$, and tensors $\cm{D}_k^{(0)} \in \mathbb{R}^{s_1\times s_2 \times s_3}$ and $\cm{R}_k^{(0)} \in \mathbb{R}^{N \times N \times p}$ for all $k\in[K]$.
			\For{iteration $i = 1,2\ldots I$}
			\State Step 1: Given $\bm U^{(i)}, \bm V^{(i)}, \bm L^{(i)}$ and $\cm{D}_k^{(i)}$, calculate $\cm{R}_k^{(i+1)}$ by using the update at \eqref{eq:update1} for each $k\in[K]$;
			\State Step 2: Given $\{\cm{R}_k^{(i+1)},k\in[K]\}$, calculate $\bm U^{(i+1)}, \bm V^{(i+1)}, \bm L^{(i+1)}, \{\cm{D}_k^{(i+1)}, k\in[K]\}$  by minimizing $\mathcal{L}^{\mathrm{RL}}(\bm U, \bm V, \bm L,\cm{D}_k, k\in [K])$ at \eqref{eq:poolconstrained}. 
			\EndFor
			\State \textbf{Output:} Obtain $\widehat{\bm U}, \widehat{\bm V}$ and $\widehat{\bm L}$ by orthonormalizing $\bm{U}^{(I)}$, $\bm{V}^{(I)}$ and $\bm{L}^{(I)}$, respectively.
		\end{algorithmic}
	\end{spacing}
\end{algorithm}

Algorithm \ref{algorithm} gives the details in searching for the common representations, yielding $\widehat{\bm U}, \widehat{\bm V}$ and $\widehat{\bm L}$. We can further orthonormalize the three representations if needed.

We next consider the optimization problem at \eqref{eq:TLloss} from the transfer learning at Stage II.
Let $\cm{R}_0=\cm{A}_0-\llb \cm{D}_0; \widehat{\bm U}, \widehat{\bm V}, \widehat{\bm L} \rrb$, and it is equivalent to the following optimization problem,
\begin{equation}\label{eq:target task}
	(\widehat{\cm{R}}_0,\widehat{\cm{D}}_0)
	\in\argmin \{\mathcal{L}_0(\llb \cm{D}_0; \widehat{\bm U}, \widehat{\bm V}, \widehat{\bm L} \rrb+\cm{R}_0 ) + \lambda_0 \|\cm{R}_0\|_\F \}, 
\end{equation}
which can be conducted by alternating updates between two blocks, $\cm{D}_0$ and $\cm{R}_0$.
Specifically, by fixing $\cm{D}_0$ at $\cm{D}_0^{(i)}$, we use the proximal gradient descent method at \eqref{eq:update1} to update the block of $\cm{R}_0$, ending up with $\cm{R}_0^{(i+1)}$, where $\cm{D}_0^{(j)}$ and $\cm{R}_0^{(j)}$ are outputs from the $j$th iteration.
In the meanwhile, by fixing $\cm{R}_0$ at $\cm{R}_0^{(i+1)}$, we can even have an explicit form to update $\cm{D}_0$,
\begin{equation}\label{eq:closedformD0}
	\cm{D}_{0(1)}^{(i+1)} = [\widehat{\bm U}^\top (\bm Y_0 -  \cm{R}_{0(1)}^{(i+1)} \bm X_0 ) \bm X_0^\top(\widehat{\bm L} \otimes \widehat{\bm V}) ] [(\widehat{\bm L} \otimes \widehat{\bm V})^\top \bm X_0 \bm X_0^\top (\widehat{\bm L} \otimes \widehat{\bm V})]^{-1}.
\end{equation}

Algorithm \ref{algorithmTL} presents the details of transfer learning, and it leads to the estimates of $\widehat{\cm{R}}_0$ and $\widehat{\cm{D}}_0$.
As a result, we have the estimated transition tensor, $\widehat{\cm{A}}_0=\llb \widehat{\cm{D}}_0; \widehat{\bm U}, \widehat{\bm V}, \widehat{\bm L} \rrb+\widehat{\cm{R}}_0$, and hence the corresponding transition matrices $\widehat{\bm A}_{j,0}$ with $1\leq j\leq p$.

\begin{algorithm}[H]
	\caption{Alternating Algorithm for Transfer Learning at Stage II}\label{algorithmTL}
	\begin{spacing}{0.9}
		\begin{algorithmic}[1]
			\State \textbf{Input:} Target dataset $\{\bm y_{t,0},-p+1 \le t \le T_0\}$, $\{\widehat{\bm U}, \widehat{\bm V}, \widehat{\bm L}\}$ estimated from Stage I, regularization parameters $\{\lambda_0\}$, number of iteration $I$.
			% , initial value: $\bm U^{(0)}$, $\bm V^{(0)}$, $\{\cm{D}_k^{(0)}\}_{k\in [K]}$, $\{\cm{R}_k^{(0)}\}_{k\in [K]}$.
			\State \textbf{Initialize:} 
			\State \quad Initialize $\cm{D}_0^{(0)} \in \mathbb{R}^{s_1\times s_2 \times s_3}$ and $\cm{R}_0^{(0)} \in \mathbb{R}^{N \times N \times p}$.
			\For{iteration $i = 1,2\ldots I$}
			\State Step 1: Given $\cm{D}_0^{(i)}$, update $\cm{R}_0$ by using the proximal gradient descent method,
			\[
			\cm{R}_0^{(i+1)} = \prox_{\eta_0 \lambda_0} \{\cm{R}_0^{(i)} - \eta_0 \nabla \mathcal{L}_k(\llb   \cm{D}_0^{(i)}; \widehat{\bm U}, \widehat{\bm V}, \widehat{\bm L} \rrb + \cm{R}_0^{(i)}  )\} \hspace{3mm}\text{with}\hspace{3mm} \eta_0 = T_0 \lambda_{\max}^{-1}(\bm X_0 \bm X_0^\top);
			\]
			\State Step 2: Given $\cm{R}_0^{(i+1)}$, calculate $ \cm{D}_0^{(i+1)}$ by using the update at \eqref{eq:closedformD0}.
			\EndFor
			\State \textbf{Output:} Obtain $\widehat{\cm{D}}_0=\cm{D}_0^{(I)}$ and $\widehat{\cm{R}}_0=\cm{R}_0^{(I)}$.
		\end{algorithmic}
	\end{spacing}
\end{algorithm}

\subsection{Parameter initialization}\label{subsec:Initialization}
The two algorithms in Section \ref{subsec:Algorithm1} both involves non-convex optimization, and hence a good initialization is necessary.
We first consider $\bm{U}^{(0)}$, $\bm{V}^{(0)}$ and $\bm{L}^{(0)}$, which encode the dominant low-dimensional representations across all source datasets, and the procedure is given below. 
\begin{itemize}
	\item \textit{Source tensor approximation}: 
	Following \cite{wang2022high}, we fit a rank-constrained VAR model to each source task, $\widetilde{\cm{A}}_k(r_{1,k}, r_{2,k}, r_{3,k}) = \argmin \mathcal{L}_k(\cm{A}_k)$, where Tucker ranks $(r_{1,k}, r_{2,k}, r_{3,k})$ are selected by the ridge-type ratio method \citep{xia2015consistently}.
%	\[\hat{r}_{i,k} = \argmin_{1\le j\le \bar{r}_i - 1} \frac{\sigma_{j+1}(\widetilde{\cm{A}}_k(\bar{r}_{1,k}, \bar{r}_{2,k}, \bar{r}_{3,k})_{(i)})+ s(N,p,T_k)}{\sigma_{j}(\widetilde{\cm{A}}_k(\bar{r}_{1,k}, \bar{r}_{2,k}, \bar{r}_{3,k})_{(i)})+ s(N,p,T_k)},\]
%	with penalty $s(N,p,T_k) = \sqrt{Np\log(T_k)/(10T_k)}$ used in Section \ref{sec:RealData} for real data analysis.
%	Here $(\bar{r}_{1,k}, \bar{r}_{2,k}, \bar{r}_{3,k})$ are conservatively large initial ranks. 
	\item \textit{Subspace extraction}: For each source task $k\in [K]$, we extract the source-specific representations $\widetilde{\bm U}_k$, $\widetilde{\bm V}_k$ and $\widetilde{\bm L}_k$ by conducting the HOSVD on $\widetilde{\cm{A}}_k({r}_{1,k}, {r}_{2,k}, {r}_{3,k})$.
	\item \textit{Subspace aggregation via principal component analysis}: 
	We aggregate the estimated representation space across $K$ sources by the following weighted projection matrices:
	\[\widetilde{\bm \Sigma}_u = \sum_{k\in [K]} w_k \widetilde{\bm U}_k \widetilde{\bm U}_k^\top, \;\;
	\widetilde{\bm \Sigma}_v = \sum_{k\in [K]} w_k \widetilde{\bm V}_k \widetilde{\bm V}_k^\top \;\;\text{and}\;\; 
	\widetilde{\bm \Sigma}_l = \sum_{k\in [K]} w_k \widetilde{\bm L}_k \widetilde{\bm L}_k^\top.\]
	We conduct eigen-decomposition on $\widetilde{\bm \Sigma}_u$ and suppress these insignificant eigenvalues. Accordingly, rank $s_1$ is the number of significant eigenvalues, and $\bm U^{(0)}$ contains the corresponding eigenvectors. Similarly, we can obtain $\bm V^{(0)}$ and $\bm L^{(0)}$, as well as ranks $s_2$ and $s_3$, by conducting eigen-decomposition on $\widetilde{\bm \Sigma}_v$ and $\widetilde{\bm \Sigma}_l$, respectively.
\end{itemize}

Note that the above procedure also suggests a method to select the ranks of common representations.
Moreover, let $\lambda$ be an eigenvalue of matrix $\widetilde{\bm \Sigma}_u$, and $\bm u$ be the corresponding eigenvector.
It holds that $\lambda = \bm u^\top \widetilde{\bm \Sigma}_u \bm u = \sum_{k\in [K]} w_k \|\widetilde{\bm U}_k^\top \bm u \|_2^2 = \sum_{k\in [K]} w_k \cos^2(\theta_{k})$,
where $\theta_k$ is the principal angle between $\bm u$ and $\mathcal{M}(\widetilde{\bm U}_k)$.
As a result, we have $0\le \lambda \le 1$, and a larger value of $\lambda$ implies that its eigenvector is more consistently aligned with source-specific response representations.
Therefore, $\bm U^{(0)}$ can capture all directions within $\{\widetilde{\bm U}_k, k\in[K]\}$ if all positive eigenvalues are taken into account, while it will choose the most well-aligned directions if we only consider the eigenvalues exceeding a threshold, say between $70\%$ and $80\%$. The situations for predictor and temporal representations, $\widetilde{\bm \Sigma}_v$ and $\widetilde{\bm \Sigma}_l$, are similar.

We finally consider the initialization for task-specific tensors $\cm{D}_k$ and $\cm{R}_k$ with $k\in\{0\}\cup [K]$, and the deviation tensor can be simply set to $\cm{R}_{k}^{(0)}=\bm{0}$. 
In the meanwhile, we can set $\cm{D}_k^{(0)} = \llb \widetilde{\cm{A}}_k({r}_{1,k}, {r}_{2,k}, {r}_{3,k}); \bm U^{(0)\top}, \bm V^{(0)\top}, \bm L^{(0)\top} \rrb$ with $k\in[K]$ for all source tasks, where $\bm{U}^{(0)}$, $\bm{V}^{(0)}$ and $\bm{L}^{(0)}$ are the initialized common representations.
Finally, from \eqref{eq:closedformD0}, we can initialize $\cm{D}_0^{(0)}$ by 
$\cm{D}_{0(1)}^{(0)} = [\widehat{\bm U}^\top \bm Y_0  \bm X_0^\top(\widehat{\bm L} \otimes \widehat{\bm V}) ] [(\widehat{\bm L} \otimes \widehat{\bm V})^\top \bm X_0 \bm X_0^\top (\widehat{\bm L} \otimes \widehat{\bm V})]^{-1}$.

\subsection{Hyperparameter selection}\label{subsec:tuning}
We first consider the selection of weights $\{w_k, k\in[K]\}$, which determine the relative contribution of each source dataset in the representation learning at Stage I, and our aim is to obtain an efficient estimation of common representations.
From Theorem \ref{thm:rep}, the variance-adjusted effective sample size has the form of, $\widetilde{T} = \min_{k\in [K]} \{ w_k^{-1} \lambda_{\max}^{-1}(\bm \Sigma_{\bbm \varepsilon,k})T_k \}$, and a larger $\widetilde{T}$ will lead to a sharper estimation error bound for common representations.
As a result, we maximize $\widetilde{T}$ with the constraint of $\sum_{k=1}^{K} w_k = 1$, leading to the optimal weights of 
$$
w_k^{\mathrm{opt}} = \frac{T_k \lambda_{\max}^{-1}(\bm \Sigma_{\bbm \varepsilon,k})}{\sum_{k\in[K]}T_k \lambda_{\max}^{-1}(\bm \Sigma_{\bbm \varepsilon,k})}.
$$
The error variances $\bm \Sigma_{\bbm \varepsilon,k}$'s are unknown in real application, and we may estimate them based on residual sequences $\{\widehat{\bbm \varepsilon}_{t,k}\}$ from the fitted VAR models. 
However, this will introduces additional variation, especially when the dimension $N$ is large.
On the other hand, we may factor out the quantities of $\lambda_{\max}^{-1}(\bm \Sigma_{\bbm \varepsilon,k})$ with $k\in[K]$ by assuming their equivalency, leading to a simplified choice of $w_k = T_k / (\sum_{k\in [K]}T_k)$ with $k\in[K]$.

We next consider the selection of regularization parameters $\lambda_k$ for $k\in\{0\}\cup[K]$, which control how far we allow the transition tensors to depart from the multilinear low-rank structure.
In the representation learning at Stage I, a larger value of $\lambda_k$ will enhance estimation efficiency of the shared representations $\bm U, \bm V$ and $\bm L$, when $\cm{A}_k$ aligns well with the low-rank structure $\llb \cm{D}_k; \bm U, \bm V, \bm L \rrb$, while a smaller value can mitigate bias when $\cm{A}_k$ is far from the low-rank structure.
On the other hand, in the transfer learning at Stage II, a larger value of $h_{\mathrm{T}}$ can incorporate more task-specific variations, i.e., $\lambda_0$ can be different from $\{\lambda_k,k\in[K]\}$ such that more flexibility can be accommodated.
Specifically, from Theorems \ref{thm:rep} and \ref{thm:TL}, we suggest $\lambda_k =c_{\mathrm{S}} \sqrt{[N^2p + N\log (N K)]/T_k}$ for $k\in [K]$ and $\lambda_0 =c_{\mathrm{T}} \sqrt{[N^2p + N\log (N) ]/T_0}$, where two constants $c_{\mathrm{S}}$ and $c_{\mathrm{T}}$ can be chosen by validation methods. 

\section{Extension to the case with many source datasets}\label{sec:LargeK}

For the $K$ source tasks in the transfer learning framework in Section \ref{sec:ProblemSetup}, each is assumed to have one task-specific parameter tensor $\cm{D}_k$ and three common representations, $\bm{U}$,$\bm{V}$ and $\bm{L}$, however, it may also have task-specific representations, $\bm{U}_k$,$\bm{V}_k$ and $\bm{L}_k$.
When the proposed methodology in Sections \ref{sec:ProblemSetup} and \ref{subsec:Estimation} is applied, these significant task-specific representations will be identified as a part of common representations, while the insignificant ones contribute to the deviation $h_{\mathrm{S}}$, leading to very large values of $s_j$ with $1\leq j\leq 3$ and $h_{\mathrm{S}}$ when there are many source datasets.
This makes it difficult to learn common representations efficiently such that they can be used to improve estimation efficiency of the target dataset.

To solve this problem, we further introduce task-specific representations to the multilinear low-rank structure at \eqref{eq:add1}, i.e. $\cm{A}_k = \llb   \cm{D}_{k}; [\bm U~\bm U_k], [\bm V~\bm V_k],[\bm L~\bm L_k] \rrb$, and accordingly the similarity measure $h_{\mathrm{S}}$ at \eqref{eq:similaritymeasure} can also be adjusted to
\begin{equation*}%\label{eq:similarityLargeK}
	\max_{k\in [K]} \min_{\substack{ \cm{D}_{k} , \bm U_k, \bm V_k,\bm L_k}} \| \cm{A}_k - \llb   \cm{D}_{k}; [\bm U~\bm U_k], [\bm V~\bm V_k],[\bm L~\bm L_k] \rrb \|_\F \le h_{\mathrm{S}},
\end{equation*}
where $\cm{D}_k \in \mathbb{R}^{(s_1+r_{1,k})\times (s_2+r_{2,k}) \times (s_3+r_{3,k})}$ is task-specific parameter tensor, $\bm U_k \in \mathbb{O}^{N\times r_{1,k}}, \bm V_k \in \mathbb{O}^{N\times r_{2,k}}$ and $\bm L_k \in \mathbb{O}^{p\times r_{3,k}}$ are task-specific representations, and $\bm U \in \mathbb{O}^{N\times s_1}, \bm V \in \mathbb{O}^{N\times s_2}$ and $\bm L \in \mathbb{O}^{p\times s_3}$ are common representations.
Moreover, the optimization problem at \eqref{eq:MTLloss} for the representation learning in Stage I can also be modified into   
\begin{equation*}%\label{eq:MTLlossLargeK}
	\begin{split}
		(\widehat{\bm U}, \widehat{\bm V}, \widehat{\bm L},& \{\widehat{\cm{A}}_k,\widehat{\cm{D}}_{k},\widehat{\bm U}_{k},\widehat{\bm V}_{k},\widehat{\bm L}_{k}, k\in [K]\}) \\
		&\in 
		\argmin
		 \sum_{k\in [K]} w_k \{\mathcal{L}_k(\cm{A}_k) 
		+ \lambda_k \|\cm{A}_k - \llb   \cm{D}_{k}; [\bm U~\bm U_k], [\bm V~\bm V_k],[\bm L~\bm L_k] \rrb \|_\F\}. 
	\end{split}
\end{equation*}
By a method similar to Sections \ref{subsec:Theory} and \ref{sec:Algorithms}, we can establish theoretical properties and algorithms for the corresponding two-stage estimation procedure.
Finally, we may also introduce the task-specific representations, $\bm U_0, \bm V_0$ and $\bm L_0$, to the target dataset if needed.

\section{Simulation studies}\label{sec:Simulation}

This section conducts three simulation experiments to evaluate how the proposed methodology is affected by three key quantities, the similarity measure $h=h_{\mathrm{S}}=h_{\mathrm{T}}$, the sample size of target dataset $T_0$, and the number of source datasets $K$, respectively.

The data are generated by model \eqref{eq:VARp} with $p=1$ and $4$, where the innovations $\bbm{\varepsilon}_{t,k}$'s are $N$-dimensional standard normal random vectors, and they are independent across all $t$ and $k$.
The sample sizes of source time series are fixed at $T_k = 300$ with $k \in [K]$ for simplicity, and the number of replications is 500.
For the VAR model with $p=4$, we construct the transition tensors $\cm A_{k}= \llbracket  \cm D_k; \bm U_k,\bm V_k, \bm L_k\rrbracket + \cm R_k$ with $k\in \{0\} \cup [K]$ by using a three-step procedure below. (i.) We first randomly generate three shared representations, $\bm U \in \mathbb{O}^{N\times s_1}$, $\bm V \in \mathbb{O}^{N\times s_2}$ and $\bm L \in \mathbb{O}^{p\times s_3}$ with $s_3=3$, and then randomly pick up two columns from each of them to form $\bm U_k\in \mathbb{O}^{N\times 2}$, $\bm V_k\in \mathbb{O}^{N\times 2}$ and $\bm L_k\in \mathbb{O}^{p\times 2}$, respectively.
(ii.) For the task-specific tensors $\cm D_k= \llbracket\cm S_k;\bm O_{1,k}, \bm O_{2,k}, \bm O_{3,k}\rrbracket\in \mathbb{R}^{2\times 2 \times 2}$, we set $\cm S_k \in \mathbb{R}^{2\times 2 \times 2}$ to a super-diagonal tensor with entries being sampled from $[0.5,0.8]$ uniformly, and then being normalized to have unit Frobenius norm subsequently, while $\bm O_{j,k}\in \mathbb{O}^{2\times 2}$ with $1\leq j\leq 3$ are randomly generated orthonormal matrices.
(iii.) For the deviation tensor $\cm R_k$, we first generate a random tensor with the same size, then project it such that $\cm R_{k(1)} \in \mathcal{M}^\perp(\bm U)$, $\cm R_{k(2)} \in \mathcal{M}^\perp(\bm V)$ and $\cm R_{k(3)}  \in \mathcal{M}^\perp(\bm L)$, and finally standardize it such that $\|\cm R_k\|_\F = h$. 
On the other hand, for the VAR model with $p=1$, there is only one transition matrix. We then suppress the temporal factor loadings, $\bm L$ and $\bm L_k$, and all other settings are the same.

In addition, we compare the two-stage estimation procedure (TL-VAR or TL) in Section \ref{subsec:Estimation} with three related modeling procedures: the pooling VAR (Pool-VAR or Pool) model, i.e., the exact transfer learning with $\lambda_k =\infty$ for all $k\in\{0\}\cup[K]$, the multilinear low-rank VAR (MLR-VAR or MLR) modeling only for the target time series $\{\bm y_{t,0}\}$ with ranks $(2,2,2)$ and $2$ for the cases with $p=4$ and 1, respectively, and the general VAR modeling only for the target time series $\{\bm y_{t,0}\}$ with order $p=4$ and 1.
The Pool-VAR is a variant of our method with $h=0$, and it hence enforces an identical low-rank structure across source and target transition tensors or matrices with shared representations. In the meanwhile, both MLR-VAR and VAR are applied to target time series only, where the VAR has no structural constraint, while the MLR-VAR imposes a low-rank structure.

The first simulation experiment is to investigate the impact of similarity measure $h$, and the sample size of target time series is fixed at $T_0 =100$.
We consider four different settings of $(K,N,s_1,s_2)=(5, 10, 3, 3)$, $(10, 10, 3, 3)$, $(10, 20, 3, 3)$ or $(10, 20, 5, 5)$, and the value of $h$ varies from $0$ to $2$ with an increment of $0.25$.
Figure \ref{fig:RMSE_wrt_h} gives root mean squared errors (RMSEs) of the estimated transition matrix or tensor, $\widehat{\bm A}_0$ or $\widehat{\cm{A}}_0$, from the target time series, and the RMSEs from general VAR models with $(N,p) = (20,4)$ are not presented in the last two subplots since they all exceed $10$ due to poor performance.
It can be seen that, by leveraging the low-rank assumption, the TL-VAR, Pool-VAR and MLR-VAR significantly outperform the general VAR when the value of $h$ is small.
However, their performance become worse when $h$ increases, since the true transition matrices or tensors gradually deviate from low-rank structures and become more heterogeneous, 
In fact, the Pool-VAR and MLR-VAR even have a worse performance than the general VAR for a very large value of $h$, while our TL-VAR can still maintain a better performance.
This confirms the flexibility of our transfer learning framework in balancing structural assumptions with dataset dissimilarity. 
%As a comparison, the Pool-VAR and MLR-VAR exhibit sensitivity to departure from low-rank structures. 
\begin{figure}
	\centering
	\includegraphics[width=0.95\linewidth]{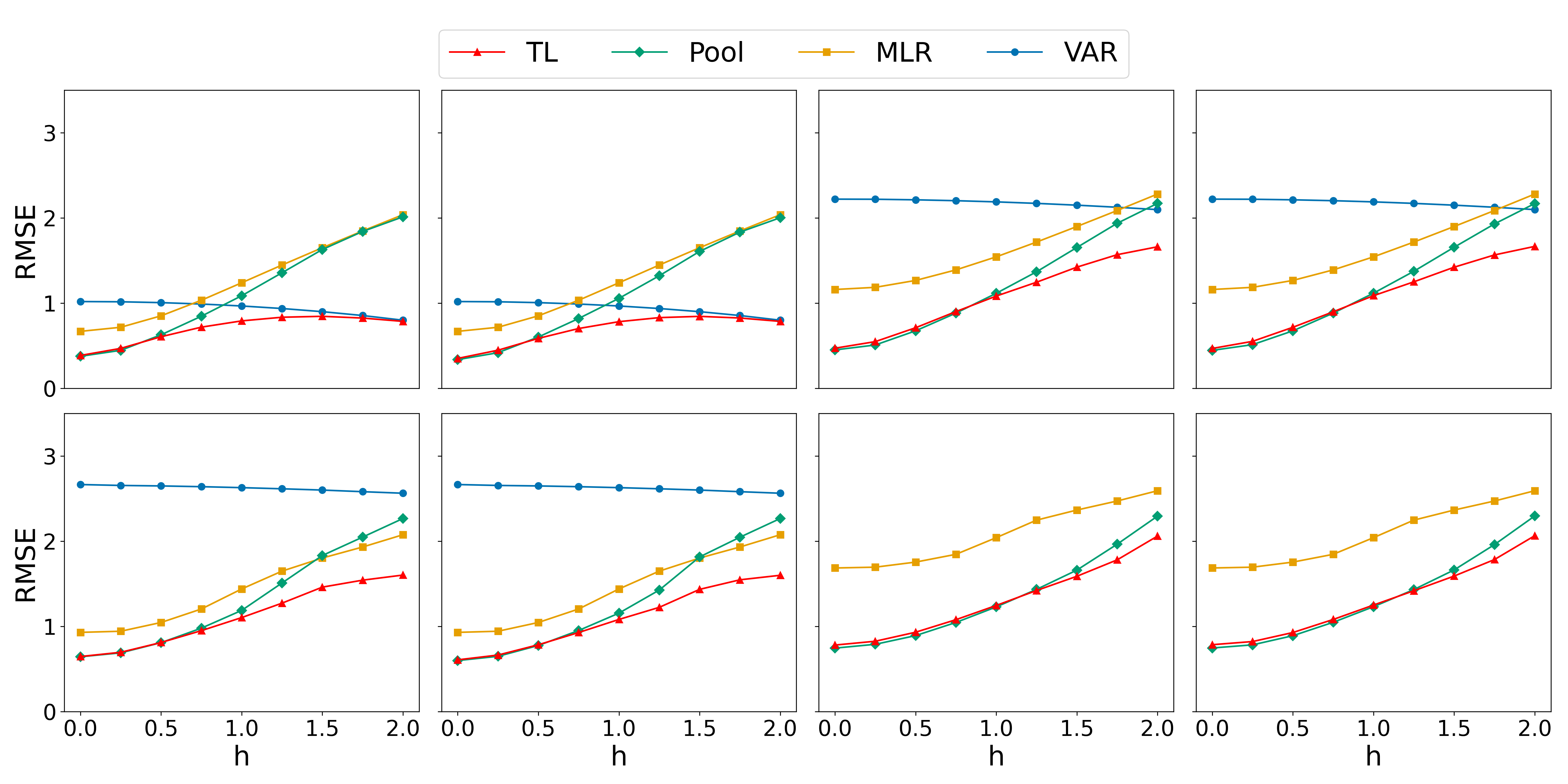}
	\caption{RMSEs  of the estimated transition matrix or tensor, $\widehat{\bm A}_0$ or $\widehat{\cm{A}}_0$, from the target time series with varying similarity measures $h$. There are four competing estimation procedures (TL, Pool, MLR and VAR), two VAR models with $p=1$ (upper panel) and $p=4$ (lower panel), and four settings of $(K,N,s_1,s_2)$ from the left to right panels.}		
	\label{fig:RMSE_wrt_h}
\end{figure}  

The second experiment is to evaluate finite-sample performance of the proposed methodology, as well as the three competing estimation procedures, with varying sample sizes of target datasets, and the similarity measure is fixed at $h=0.5$.
We use the same four settings of $(K,N,s_1,s_2)$ in the first experiment, and the sample size $T_0$ varies from 50 to 300 with an increment of 50.  
Figure \ref{fig:RMSE_wrt_T0} presents the RMSEs of $\widehat{\bm A}_0$ or $\widehat{\cm{A}}_0$ from the target time series. It can be seen that the TL-VAR has a substantial efficiency gain in data-scarce regimes (i.e. smaller $T_0$), making necessary to borrow structural information from source datasets to compensate the limited number of observations in target datasets. 
On the other hand, in data-rich regimes (i.e. larger $T_0$), the efficiency gain looks small, and the single-task method, i.e. MLR-VAR, also has a good performance.
Of course, it becomes significant again when we zoom in these plots. As a result, we can conclude the necessity of our transfer learning framework, especially when the sample sizes of target datasets are small.

\begin{figure}
	\centering
	\includegraphics[width=0.95\linewidth]{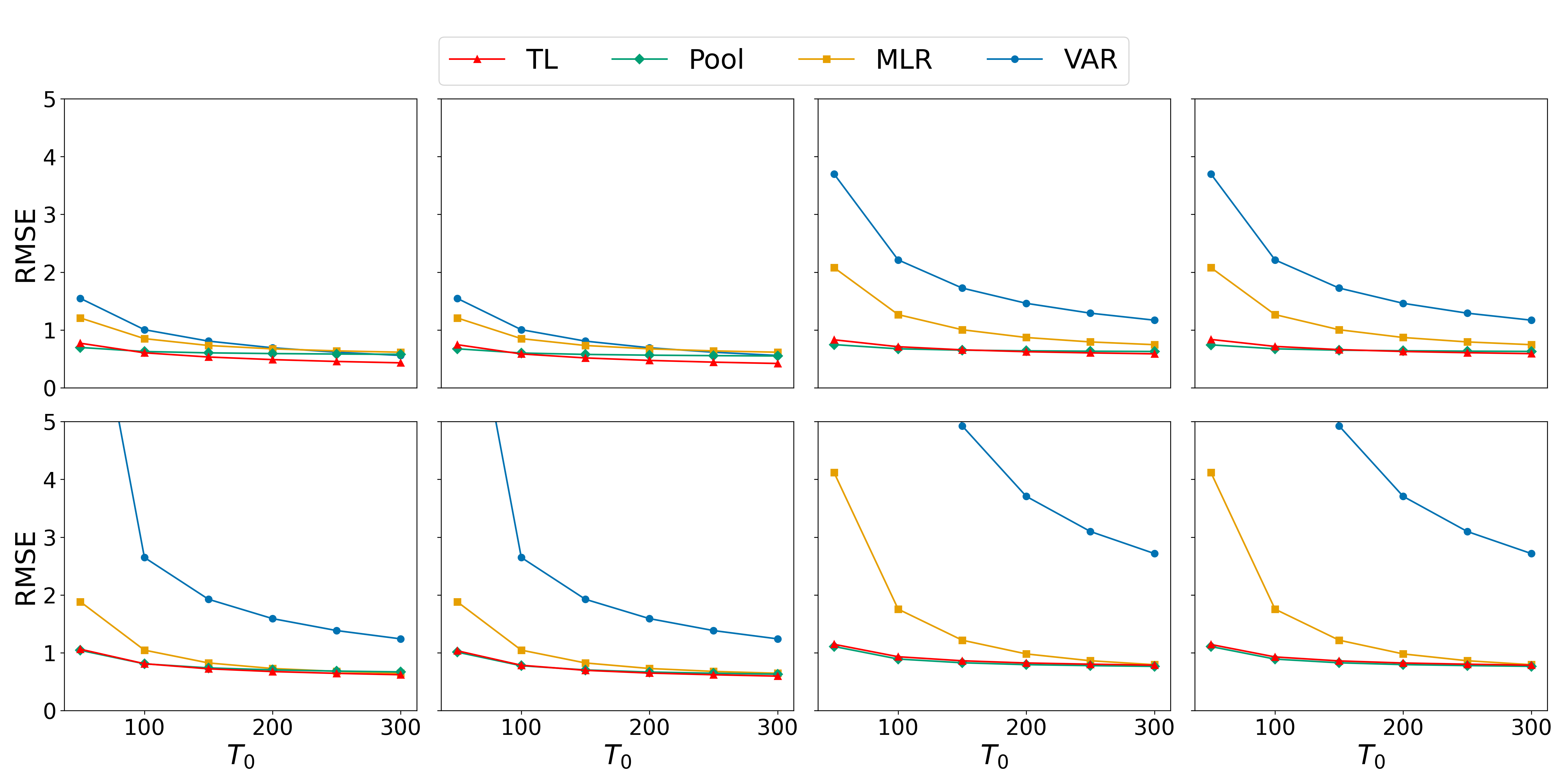}
	\caption{RMSEs  of the estimated transition matrix or tensor, $\widehat{\bm A}_0$ or $\widehat{\cm{A}}_0$, from the target time series with varying sample sizes $T_0$. There are four competing estimation procedures (TL, Pool, MLR and VAR), two VAR models with $p=1$ (upper panel) and $p=4$ (lower panel), and four settings of $(K,N,s_1,s_2)$ from the left to right panels. }
	\label{fig:RMSE_wrt_T0}
\end{figure}

The third experiment is to evaluate how the number of source datasets $K$ affects the transfer learning procedure.
We set three different values of $K$, i.e. $5$, $10$ and $50$, and the similarity measure $h$ varies from $0$ to $1$ with an increment of $0.25$.
The three competing estimation procedures are not considered in this experiment, and the other settings are fixed at $(T_0,N,s_1,s_2)=(100,20,5,5)$.
Figure \ref{fig:RMSE_wrt_K} provides the RMSEs of $\widehat{\bm A}_0$ or $\widehat{\cm{A}}_0$ from the target time series.
It can be seen that increasing $K$ can uniformly improve estimation efficiency of target time series, with RMSEs decreasing when more source datasets are incorporated.
The improvement comes from the better estimation of common representations, $\bm U$, $\bm V$ and $\bm L$, due to a larger pool of related sources.
As a result, our transfer learning procedure has benefits proportional to the availability of  informative source datasets, suggesting that practitioners should leverage related sources, as many as possible, to enhance estimation accuracy.

%, particularly when the shared representations dominate cross-task variability.

\begin{figure}[htbp]
	\centering
	\includegraphics[width=0.95\linewidth]{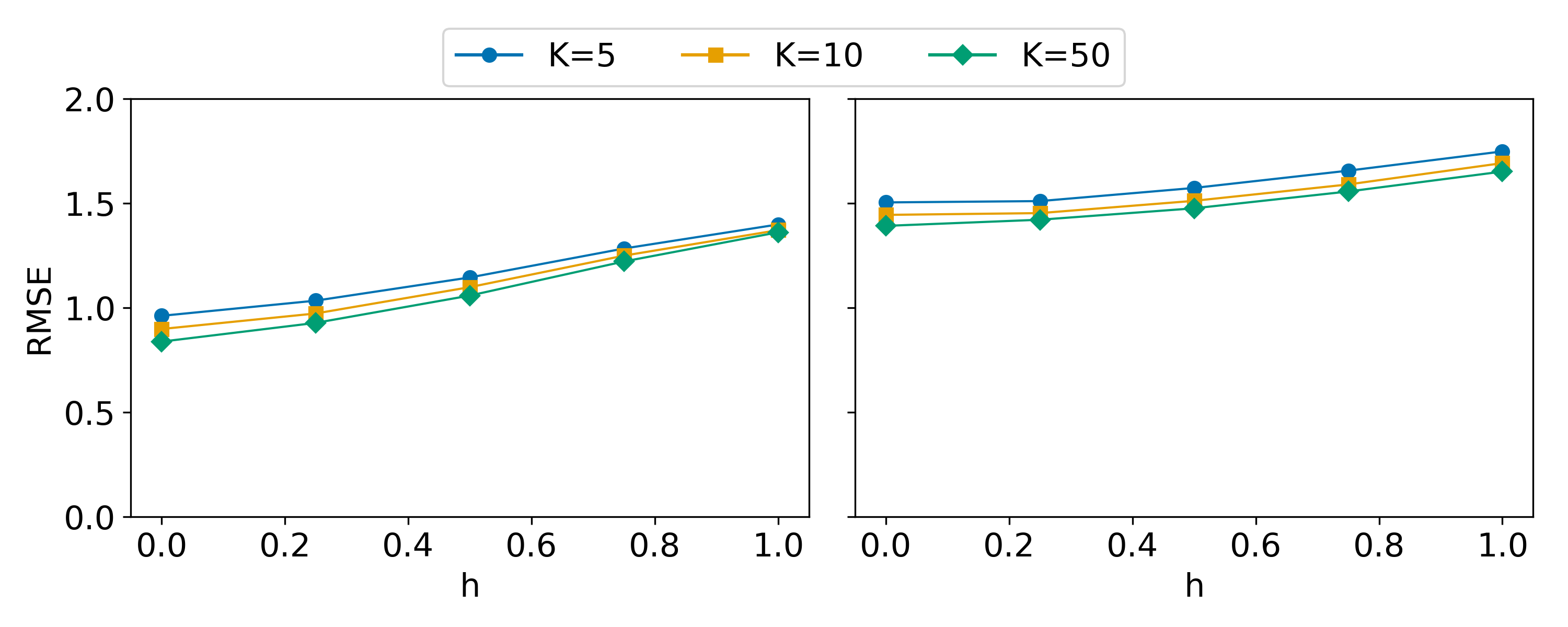}
	\caption{RMSEs of the estimated transition matrix or tensor, $\widehat{\bm A}_0$ or $\widehat{\cm{A}}_0$, from the target time series with varying similarity measure $h$. There are three numbers of source datasets, and two VAR models with $p=1$ (left panel) and $p=4$ (right panel). }
	\label{fig:RMSE_wrt_K}
\end{figure}

%\newpage
\section{Empirical Analysis}\label{sec:RealData}

This section analyzes $20$ quarterly macroeconomic variables of Japan, and they can be classified into six categories: GDP \& production, labor market, interest rates \& money supply, prices \& inflation, housing \& real estate, and others. 
The same 20 macroeconomic variables are also available at another $K=9$ countries: Australia (AUS), Canada (CAN), Denmark (DNK), Great Britain (GBR), South Korea (KOR), Norway (NOR), New Zealand (NZL), Sweden (SWE), and the United States (USA). 
The data are downloaded from the OECD (\url{www.oecd.org}), and Figure \ref{fig:country_indicator_length} gives the ending and starting time points of each variable across all ten countries.
We first truncate the data such that the 20 variables have uniformly aligned time period for each country, and a preprocessing procedure is then conducted: a possible seasonal adjustment, a possible transformation to remove non-stationarity, and finally standardizing each resulting sequence with mean zero and variance one; see Appendix \ref{sec:additional-empirical} for details.
As a result, the target time series of Japan contains $87$ time points, whereas the source time series have sample sizes from $91$ (Sweden) to $216$ (United States).

We first explore structural similarity of $N=20$ macroeconomic variables across ten countries.
To this end, by a method similar to the initialization in Section \ref{subsec:Initialization}, a multilinear low-rank VAR model \citep{wang2022high} is fitted to $N$-dimensional time series from each country, where the order is set to $p=4$ as in \cite{koop2013forecasting}, and then the HOSVD is conducted to the estimated transition tensor $\widetilde{\cm{A}}_k$, leading to task-specific representations of $\widetilde{\bm U}_k$, $\widetilde{\bm V}_k$ and $\widetilde{\bm L}_k$ with $0\leq k\leq K$.
As a result, the eigen-decomposition can be performed to the weighted projection matrices of ten countries, $\sum_{k=0}^K w_k \widetilde{\bm U}_k \widetilde{\bm U}_k^\top$,
$\sum_{k=0}^K w_k \widetilde{\bm V}_k \widetilde{\bm V}_k^\top$ and
$\sum_{k=0}^K w_k \widetilde{\bm L}_k \widetilde{\bm L}_k^\top$, and we can choose the common representations of $\widetilde{\bm{U}}$, $\widetilde{\bm{V}}$ and $\widetilde{\bm{L}}$ by using the elbow method with ranks $(s_1, s_2, s_3) = (3, 5, 2)$, capturing $79.08\%$, $77.60\%$ and $70.86\%$ of the total variance, respectively.  
Moreover, to measure whether $\widetilde{\bm{U}}$ can well represent these $\widetilde{\bm U}_k$'s, we project each $\widetilde{\bm U}_k$ to the orthogonal complement of $\mathcal{M}(\widetilde{\bm{U}})$, and it results in the residual $(\bm{I}_N - \widetilde{\bm{U}} \widetilde{\bm{U}}^{\top})\widetilde{\bm U}_k$, whose projection matrix is presented in Figure \ref{fig:country_indicator_length} in terms of heatmaps.  
The residual projection matrices of $\widetilde{\bm V}_k$'s and $\widetilde{\bm L}_k$'s are given in Appendix \ref{sec:additional-empirical}.
It can be seen that these residual projection matrices are all nearly zero, suggesting that these task-specific representations share three common spaces with small ranks, respectively.

We next apply the proposed transfer learning framework to model the target time series of Japan by leveraging auxiliary information from source datasets of the $K=9$ countries.
The ranks of common representations are selected to $(s_1, s_2, s_3) = (3, 5, 2)$ by the elbow method during the parameter initialization with nine source countries in Section \ref{subsec:Initialization}, and they capture $81.29\%, 76.46\%$, and $71.40\%$ of the total variance, respectively. 
Moreover, for regularization parameters $\lambda_k =c_{\mathrm{S}} \sqrt{[N^2p + N\log (N K)]/T_k}$ with $k\in[K]$, the last 20 time points for all nine countries are reserved for the validation method to search for $c_{\mathrm{S}}$, among $(0,2]$ with an increment of 0.25, and a rolling forecasting method with fixed starting points is used to calculate prediction errors, leading to the selection of $c_{\mathrm{S}}=0.5$. 
We set $c_{\mathrm{T}}=c_{\mathrm{S}}$, and the algorithms in Section \ref{subsec:Algorithm1} are applied to the whole datasets, leading to the estimated transition tensor, $\widehat{\cm{A}}_0=\llb \widehat{\cm{D}}_0; \widehat{\bm U}, \widehat{\bm V}, \widehat{\bm L} \rrb+\widehat{\cm{R}}_0$, for the target time series of Japan.

Figure \ref{fig:projection_matrixUVL} presents the projection matrices of three common representations, i.e. $\widehat{\bm U}\widehat{\bm U}^\top$, $\widehat{\bm V}\widehat{\bm V}^\top$ and $\widehat{\bm L}\widehat{\bm L}^\top$.
It can be seen that the response representations reveal interpretable groupings, which are well aligned with macroeconomic categories, while the predictor ones capture fine-grained temporal dependencies rather than static groupings, suggesting that predictive signals are not confined within macroeconomic categories but span across them.
Moreover, the common temporal representations concentrate to the first three lags,  and it is consistent with a macroeconomic modeling principle that recent data hold greater predictive relevance.
We now focus on the estimated transition tensor $\widehat{\cm{A}}_{0(1)}=(\widehat{\bm A}_{1,0},\ldots,\widehat{\bm A}_{4,0})$, and Figure \ref{fig:realdataJPNcoe} gives heatmaps of the four transition matrices.
Larger absolute values of coefficients can be observed when the responses lie in the category of prices \& inflation. 
Specifically, lagged labor market variables, i.e. hourly earnings and unit labor costs, exert positive pressure on price and inflation levels, indicative of cost-push inflation dynamics where rising wages propagate through supply chains.
In the meanwhile, lagged price and inflation variables exhibit negative effects, reflecting mean reversion mechanisms that dampen inflationary persistence, i.e. a stabilizing feedback likely tied to monetary policy adjustments or market corrections. 
Finally, the estimated transition matrices diminish as lags increase, consistent with macroeconomic inertia where recent shocks dominate short-term fluctuations.

\begin{figure}[htp]
	\centering
	\includegraphics[width=0.95\linewidth]{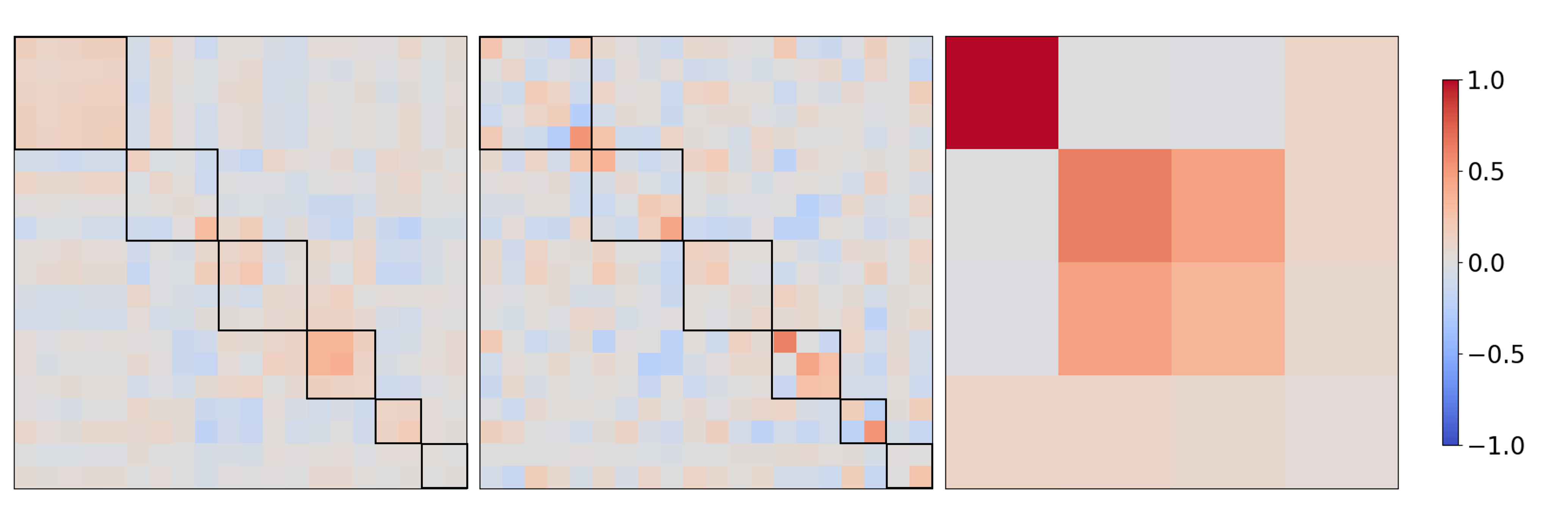}
	\caption{Heatmaps for projection matrices of the common response, predictor and temporal representations with ranks of 3, 5 and 2 from the left to right panels, respectively. The marked squares in the left and middle panels refer to the six macroeconomic categories.}
	\label{fig:projection_matrixUVL}
\end{figure}

\begin{figure}[htp]
	\centering
	\includegraphics[width=1\linewidth]{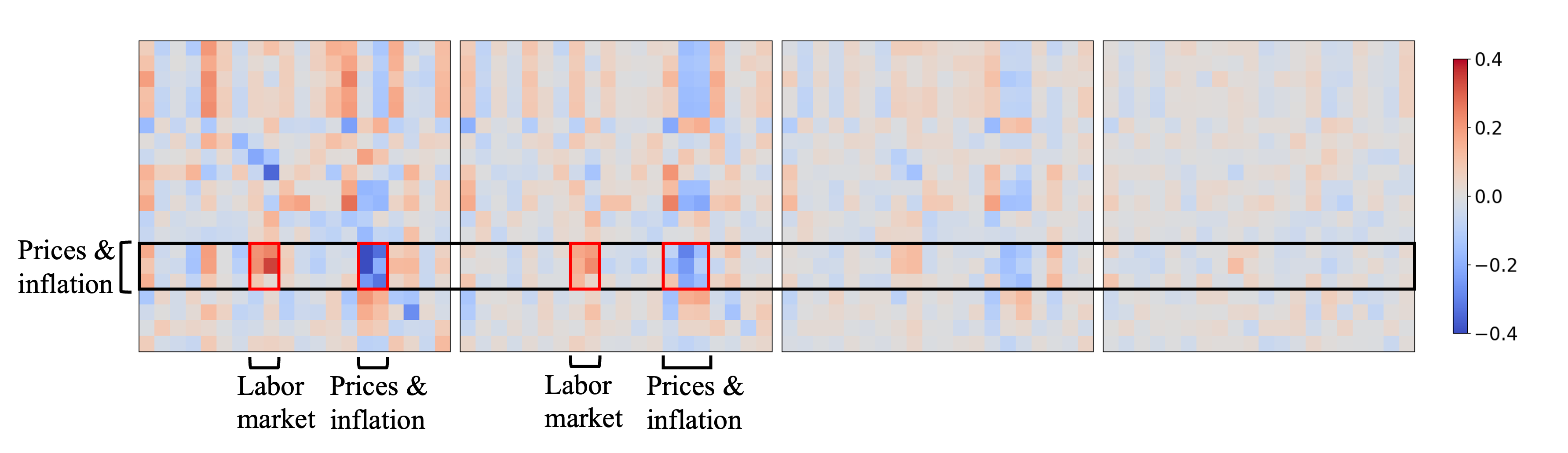}
	\caption{Heatmaps for transition matrices $\widehat{\bm A}_{j,0}$ at lags from $j=1$ to 4 (from the left to right panels) with $\widehat{\cm{A}}_{0(1)}=(\widehat{\bm A}_{1,0},\ldots,\widehat{\bm A}_{4,0})$.}
	\label{fig:realdataJPNcoe}
\end{figure}

Finally, we compare the proposed transfer learning framework (TL-VAR or TL) with five competitors: two transfer learning and three single-task learning methods. The two transfer learning approaches include the pooling VAR (Pool-VAR or Pool), i.e. the exact transfer learning, which enforces an identical shared representation structure across all countries, and the initialization VAR (Initial-VAR or Initial), where the initialized parameters of $\bm U^{(0)}$, $\bm V^{(0)}$ and $\bm L^{(0)}$ are used as the fitted common representations.
The three single-task learning methods use the target time series of Japan only, and they include the general VAR model without structural constraint (VAR), the multilinear low-rank VAR (MLR-VAR or MLR) with ranks being selected by the ridge-type ratio method, and the sparse VAR (Sparse-VAR or Sparse) with lasso penalty, where the hyperparameter is selected by grid search among $\{0.1, 0.3, 0.5, 0.7, 0.9\}$. The order of all the above methods is fixed at $p=4$.
A rolling forecasting procedure with fixed starting points is conducted to the dataset of Japan, and the last 20 time points are reserved for the testing.
For our TL-VAR model, the common representations $\widehat{\bm U}, \widehat{\bm V}$ and $\widehat{\bm L}$ are estimated from the full datasets, and we only re-estimate models for the target time series of Japan at each prediction.
On the other hand, to evaluate effects of the number of source datasets, we consider two groups of sources: all the nine countries, i.e. $K=9$, and the case without South Korea (KOR), Norway (NOR) and Sweden (SWE), i.e. $K=6$, since they have sample sizes comparable to that of Japan.

\begin{table}
	\caption{\label{table_prediction} Rolling forecasting errors with nine $K=9$ or six $K=6$ source datasets. There are three transfer learning methods, TL-VAR, Pool-VAR and Initial-VAR, and three single-task learning methods, VAR, MLR-VAR and Sparse-VAR.}
	\begin{center}
		\resizebox{0.95\textwidth}{!}{
			\begin{tabular}{cccccccccc}
				\toprule
				& \multicolumn{3}{c}{$K=9$} & \multicolumn{3}{c}{$K=6$} & \multicolumn{3}{c}{Single-task Learning} \\
				\cmidrule(r){2-4} \cmidrule(r){5-7} \cmidrule(r){8-10}
				& \textbf{TL} & \textbf{Pool} & \textbf{Initial} & \textbf{TL} & \textbf{Pool} & \textbf{Initial} & \textbf{VAR} & \textbf{MLR} & \textbf{Sparse} \\
				\hline
				RMSFE & \textbf{4.836} & 5.185 & 5.451 & 4.992 &  5.502  & 5.465  & 135.444 & 6.444 & 5.094 \\
				MAFE & \textbf{1.948} & 2.202 & 2.312 & 2.082 & 2.256& 2.319  & 43.482 & 2.258 & 2.283 \\
				\bottomrule
		\end{tabular}}
	\end{center}
\end{table}

Table \ref{table_prediction} gives the root mean squared forecast error (RMSFE) and mean absolute forecast error (MAFE) to evaluate forecasting performance, and it can be seen that
our TL-VAR method, whether using nine or six source datasets, outperforms all competitors, demonstrating its ability to achieve positive transfer by leveraging information from related sources. 
Moreover, the superior performance of TL-VAR against Pool-VAR implies that overly restrictive assumptions, such as imposing an exact low-rank condition on the transition matrices, may limit the potential for transfer learning improvements. In addition, the TL-VAR method with $K=9$ exhibits slightly better prediction performance than that with $K=6$, aligning with simulation findings that increasing the number of related informative source datasets will enhance the prediction accuracy of our transfer learning framework.

\section{Conclusion and discussion}\label{sec:Conclusions}

This paper introduces a representation-based transfer learning framework for vector autoregressive (VAR) models, and it can leverage auxiliary information from source time series to enhance estimation efficiency and prediction accuracy of target time series. 
A two-stage regularized estimation procedure is then suggested, and its non-asymptotic error bounds are established. 
The proposed transfer learning framework has no temporal alignment constraint on both source and target time series, and it hence is flexible to integrate datasets with asynchronous time periods. 
Moreover, from both theoretical analysis and simulation studies, our method can achieve positive transfer when transition matrices of VAR models can be sufficiently approximated by low-rank structures with common representations.
Its advantages against single-task learning methods become more pronounced for the cases with smaller target sample sizes or a larger number of available related source datasets.
Finally, the applications to 20 macroeconomic variables from Japan and another nine countries further confirm the necessity of the proposed transfer learning framework.

This paper can be extended along three directions below.
First, to capture persistent dynamics or volatility clustering in high-dimensional time series, the proposed transfer learning framework can be generalized to more advanced models, such as autoregressive moving average \citep{Wilms2023} or generalized autoregressive conditional heteroscedastic \citep{Zhu_Li_Zhang_Li2025} models.
In addition, given that economic and financial data typically exhibit heavy-tailed distributions \citep{wang2023robust}, it is of interest to incorporate robust estimation methods into the transfer learning framework.   
Finally, complex dependent data become increasingly common in the era of big data \citep{wang2024high}, and it is of importance to apply the transfer learning techniques to tensor-valued time series. 
It will be valuable contributions to the literature to investigate transfer learning techniques for these advanced frameworks, and then to establish their theoretical foundations to guide practitioners. 
%We leave them for future researches.

%\AtAppendix{\counterwithin{lemma}{section}}
%\AtAppendix{\counterwithin{definition}{section}}
%
% \setcounter{tocdepth}{0}
\section*{Appendix}\label{appendix}
This Appendix provides notation, technical details for all theorems and propositions, as well as additional empirical results.
Specifically, Section~\ref{sec:notation} introduces the notation, problem setup, and equivalent formulations of the VAR($p$) model.
Section~\ref{sec:mainproof} presents the detailed proofs of Theorems~\ref{thm:rep}--\ref{thm:TL} and Propositions~\ref{prop:worse}--\ref{prop:PoolVAR} in the manuscript, which need Propositions~\ref{prop:determrep}--\ref{prop:determTLpool} and Lemmas~\ref{lemma:deterministic1}--\ref{lemma:uniondeviation}.
Section~\ref{sec:techlemmas} collects auxiliary Lemmas~\ref{lem-inf-conv}--\ref{lem-frac} for Section~\ref{sec:mainproof}.
Finally, Section~\ref{sec:additional-empirical} reports additional empirical results that complement those presented in the manuscript.
% Throughout this Appendix, we denote vectors by boldface small letters, e.g. $\bm{x}$, $\bm{y}$, matrices by boldface capital letters, e.g. $\bm{X}$, $\bm{Y}$, and tensors by calligraphic boldface letters, e.g. $\cm{X}$, $\cm{Y}$. For any two real-valued sequences $x_k$ and $y_k$, we write $x_k\gtrsim y_k$ if there exists a  constant $c>0$ such that $x_k\geq cy_k$ for all $k$, and write $x_k\asymp y_k$ if $x_k\gtrsim y_k$ and $y_k\gtrsim x_k$. Denote $C$ as a generic positive constant, which is independent of the dimensions and the sample size.
% For a matrix $\bm{X}$, denote by $\bm{X}^\top$, $\|\bm{X}\|_\F$, $\sigma_i(\bm{X})$, $\mathcal{M}(\bm{X})$, and $\mathcal{M}^\perp(\bm{X})$ its transpose, Frobenius norm, $i$-th largest singular value, column space, and orthogonal complement of column space, respectively. For a square matrix $\bm{X}$, denote by $\lambda_{\max}(\bm{X})$ and $\lambda_{\min}(\bm{X})$ its largest and smallest eigenvalue, respectively. For two matrices $\bm{X}_1$ and $\bm{X}_2$, denote by $(\bm{X}_1,\bm{X}_2 )$, $\bm X_1 \otimes \bm X_2$ and $\langle\bm{X}_1,\bm{X}_2\rangle$ their column-wise matrix concatenation, Kronecker and inner product, respectively.

% \tableofcontents

% 重置计数器
\setcounter{section}{0} % 重置章节计数器
\setcounter{equation}{0} % 重置公式计数器
\setcounter{figure}{0} % 重置公式计数器
\setcounter{table}{0} % 重置公式计数器
% 设置附录的章节编号格式：A.1, A.2 ...
\def\theequation{A\arabic{section}.\arabic{equation}}
\def\thesection{A\arabic{section}}

% 设置附录中定理、引理等的编号
\renewcommand{\thetheorem}{A\arabic{theorem}}  
\renewcommand{\thelemma}{A\arabic{lemma}}  
\renewcommand{\thedefinition}{A\arabic{definition}}  
\renewcommand{\theproposition}{A\arabic{proposition}}  
\renewcommand{\thefigure}{A\arabic{figure}}  
\renewcommand{\thetable}{A\arabic{table}}  

\section{Notation and Problem Setup}\label{sec:notation}

This section collects the technical preliminaries for the theoretical analysis.
Section \ref{sec:prelim} introduces the core notation and basic tensor–matrix algebra used throughout the proofs. 
Section \ref{sec:vecform} restates the problem setup by rewriting the VAR($p$) model into an equivalent vectorized form. %which yields a unified tensor representation for multi-task estimation and transfer learning across related time series. 
Section \ref{sec:VMAinf} rewrites the VAR($p$) model into a stacked VAR$(1)$ form and a moving-average representation to characterize temporal and cross-sectional dependence, leading to the task-specific regularity parameters used in our non-asymptotic analysis.
Section \ref{sec:definitions} presents key definitions employed in establishing the main theoretical results.

\subsection{Tensor algebra \label{sec:prelim}}

Tensors are multi-dimensional generalizations of matrices and serve as higher-order extensions of familiar linear-algebraic objects. An $m$-th order tensor $\bm{\mathcal{X}} \in \mathbb{R}^{p_1 \times \cdots \times p_m}$ is an $m$-way array. We summarize here the basic tensor notation and operations used throughout this Appendix, and refer to \citet{kolda2009tensor} for a detailed review of tensor algebra.

Matricization, also known as unfolding, is the process of reordering the elements of a third- or higher-order tensor into a matrix. The most commonly used matricization is the one-mode matricization defined as follows.  For any $m$-th-order tensor $\cm{X}\in\mathbb{R}^{p_1\times\cdots\times p_m}$, its mode-$s$  matricization $\cm{X}_{(s)}\in\mathbb{R}^{p_s\times p_{-s}}$, with $p_{-s}=\prod_{i=1, i\neq s}^{m}p_i$, is the matrix obtained by setting the $s$-th tensor mode as its rows and collapsing all the others into  its columns, for $s=1,\dots, m$. Specifically,  the $(i_1,\dots,i_d)$-th element of $\cm{X}$  is mapped to the $(i_s,j)$-th element of $\cm{X}_{(s)}$, where
\begin{equation}
	j=1+\sum_{\substack{k=1\\k\neq s}}^m(i_k-1)J_k~~\text{with}~~J_k=\prod_{\substack{\ell=1\\\ell\neq s}}^{k-1}p_\ell.
\end{equation}

We next review the concepts of tensor-matrix multiplication, tensor generalized inner product and norm. For any $m$-th-order tensor $\cm{X}\in\mathbb{R}^{p_1\times\cdots\times p_m}$ and matrix $\bm{Y}\in\mathbb{R}^{q_k\times p_k}$ with $1\leq k\leq m$,
the mode-$k$ multiplication $\cm{X}\times_k\bm{Y}$ produces an $m$-th-order tensor in $\mathbb{R}^{p_1\times\cdots\times p_{k-1}\times q_k\times p_{k+1}\times\cdots\times p_m}$  defined by
\begin{equation}
	\left(\cm{X}\times_k\bm{Y}\right)_{i_1\cdots i_{k-1}ji_{k+1}\dots i_d}=\sum_{i_k=1}^{p_k}\cm{X}_{i_1\cdots i_m}\bm{Y}_{ji_k}.
\end{equation}
For any two tensors $\cm{X}\in\mathbb{R}^{p_1\times p_2\times\cdots\times p_m}$ and $\cm{Y}\in\mathbb{R}^{p_1\times p_2\times\cdots p_n}$ with $m\geq n$, their generalized inner product $\langle\cm{X},\cm{Y}\rangle$ is the $(m-n)$-th-order tensor  in $\mathbb{R}^{p_{n+1}\times\dots\times p_m}$ defined by
\begin{equation}\label{eq:tensorinner}
	\langle\cm{X},\cm{Y}\rangle_{i_{n+1}\dots i_{m}}=\sum_{i_1=1}^{p_1}\sum_{i_2=1}^{p_2}\dots\sum_{i_n=1}^{p_n}\cm{X}_{i_1i_2\dots i_ni_{n+1}\dots i_m}\cm{Y}_{i_1i_2\dots i_n},
\end{equation}
where  $1\leq i_{n+1}\leq p_{n+1},\dots,1\leq i_m\leq p_m$. In particular, when $m=n$, it reduces to the conventional real-valued inner product.  In addition,  the Frobenius norm of any tensor $\cm{X}$ is defined as $\|\cm{X}\|_{\text{F}}=\sqrt{\langle\cm{X},\cm{X}\rangle}$.

\subsection{Problem setup}\label{sec:vecform}
For each task $k \in \{0\} \cup [K]$, $\bm Y_k = \cm A_{k(1)} \bm X_k + \bm E_k$, where $ \cm{A}_{k} \in \mathbb{R}^{N\times N \times p}$, $\bm Y_k = (\bm y_{1,k}, \ldots, \bm y_{T_k,k}) \in \mathbb{R}^{N \times T_k}$, $\bm X_k = (\bm x_{1,k}, \ldots, \bm x_{T_k,k})\in \mathbb{R}^{Np \times T_k}$ with $\bm x_{t,k}=(\bm y_{t-1,k}^\top , \ldots, \bm y_{t-p,k}^\top)^\top \in \mathbb{R}^{Np}$ and $\bm E_k = (\bbm{\varepsilon}_{1,k}, \ldots, \bbm{\varepsilon}_{T_k,k}) \in \mathbb{R}^{Np \times T_k}$. Define the square loss for each task as
$$
\mathcal{L}_k(\bbm \alpha_k) = \frac{1}{2 T_k} \| \bm y_k - \bm Z_k \bbm \alpha_k \|_2^2,
$$
where $ \bm y_k = \text{vec}(\bm Y_k^\top), \bm Z_k =  \bm I_N \otimes \bm X_k^\top $, and $ \bbm \alpha_k = \text{vec}(\cm A_{k(1)}^\top) $.
Let $\bbm \alpha_k^*$ denotes the ground truth of $\bbm \alpha_k$. Suppose there exist $\bm B^* = \bm U^* \otimes \bm L^* \otimes \bm V^*$ and $\bm d_k = \text{vec}(\cm D_{k(1)}^\top)$ for all $k \in \{0\} \cup [K]$, where $\bm U^* \in \mathbb{O}^{N \times s_1}$, $\bm V^* \in \mathbb{O}^{N \times s_2}$, and $\bm L^* \in \mathbb{O}^{p \times s_3}$. 
Then, by the similarity measure as defined in the main text, equation \eqref{eq:similaritymeasure} and Assumption \ref{asmp:3}($ii$), it follows that
$$ \max_{k\in [K]}  \min_{\bm d_k \in \mathbb{R}^{s_1 s_2 s_3}}\| \bbm \alpha_k^* - \bm B^* \bm d_k \|_\F \le h_{\mathrm{S}} 
\quad\text{and}\quad 
 \min_{\bm d_0 \in \mathbb{R}^{s_1 s_2 s_3}}\| \bbm \alpha_0^* - \bm B^* \bm d_0\|_\F \le h_{\mathrm{T}}$$
The overall estimation consists of two stages: representation learning and transfer learning.
\begin{itemize}
\item \textbf{Stage I (Representation learning):} Estimate the shared representation $\bm B$ by solving a weighted minimization over source tasks:
\begin{equation}\label{eq:vecrepresentation learning}
    \begin{split}
         \widehat{\bm B} \in& \argmin_{\substack{
        \bm B = \bm U \otimes \bm L \otimes \bm V, \\\bm U,\bm V,\bm L \in \mathbb{O}}} \sum_{k\in [K]} w_k\min_{\substack{\bbm \alpha_k , \bm d_k}} [\mathcal{L}_k(\bbm \alpha_k) + \lambda_k\|\bbm \alpha_k - \bm B \bm d_k\|_2].
    \end{split}
\end{equation}
\item \textbf{Stage II (Transfer learning):} Adapt the learned representation to the target task
$$
(\widehat{\bbm \alpha}_0, \widehat{\bm d}_0) \in \argmin_{\bbm \alpha, \bm d} \left\{ \mathcal{L}_0(\bbm \alpha) + \lambda_0 \| \widehat{\bm B} \bm d - \bbm \alpha \|_2 \right\}.
$$
\end{itemize}
Denote 
\begin{equation}\label{eq:infcov}
    \widetilde{\mathcal{L}}_k(\bbm \beta) = \min_{\bbm\alpha} \{\mathcal{L}_k(\bbm\alpha) + \lambda_k \|\bbm\beta - \bbm\alpha\|_2\}. 
\end{equation}
Then, for each source task $ k \in [K] $, the optimal solutions to the inner minimizations satisfy the following equivalent forms:
\begin{align}
    &\widehat{\bm d}_k \in \argmin_{\bm d} \widetilde{\mathcal{L}}_k (  \widehat{\bm B}\bm d ),
    \label{eq:representation-1} \\
    &\widehat{\bbm \alpha}_k \in \argmin_{\bbm \alpha} \{ \mathcal{L}_k (\bbm \alpha) + \lambda_k \| \widehat{\bm B} \widehat{\bm d}_k - \bbm \alpha \|_2 \}, \label{eqn-lowrank-mtl-theta} \\
    & \sum_{k=1}^{K} w_k \widetilde{\mathcal{L}}_k (  \widehat{\bm B} \widehat{\bm d}  ) \leq \sum_{k=1}^{K} w_k \widetilde{\mathcal{L}}_k (  \bm B^{*} \bm d^{*}  ) .\label{eq:lowrank-mtl-2}
\end{align}

\subsection{Equivalent forms for VAR($p$) model}\label{sec:VMAinf}
Following the stability analysis framework of \cite{basu2015regularized}, we quantify both temporal and cross-sectional dependencies in the model.
For any $k\in\{0\}\cup [K]$, let $\bm \Xi_k(z) = \bm I_N - \sum_{i=1}^p \bm A_{i,k} z^i$ be the matrix polynomials with $z \in \mathbb{C}$.
Define $\mu_{\min}(\bm \Xi_k) = \min_{|z| = 1}\lambda_{\min}(\bm \Xi_k^\dagger(z)\bm \Xi_k(z))$, where $\bm \Xi_k^\dagger(z)$ is the conjugate transpose of $\bm \Xi_k(z)$.
We further introduce the following task-specific quantities:
	$$
	\rho_k = \frac{\lambda_{\min}(\bm \Sigma_{\bbm \varepsilon,k})}{\mu_{\max}(\bm \Xi_k)}, \;\;
	L_k = \frac{3\lambda_{\max}(\bm \Sigma_{\bbm \varepsilon,k})}{\mu_{\min}(\bm \Xi_k)},\;\;
	\kappa_k = \frac{L_k}{\rho_k}, \;\;\text{and}\;\;
	M_{k} = \frac{\lambda_{\max}(\bm \Sigma_{\bbm \varepsilon,k})}{\mu_{\min}^{1/2}(\bm \Xi_k)}.
	$$

The VAR($p$) process for the sequence ${\bm y_t \in \mathbb{R}^N}$ can be equivalently expressed as a VAR($1$) process in the stacked lag form ${\bm x_t \in \mathbb{R}^{Np}}$, given by $\bm x_{t,k} = \bm \Phi_k^* \bm x_{t-1,k} + \bm e_{t,k}$, and hence the VMA($\infty$) representation of $\bm x_{t,k} = \sum_{i=1}^\infty \bm \Phi_k^{*i} \bm e_{t-i,k}$ or $\widetilde{\bm x}_k = \widetilde{\bm \Phi}_k^* \widetilde{\bm e}_k$ by Assumption \ref{asmp:1}--\ref{asmp:2}, where
        \begin{equation}\label{eq:bigVAR1coe}
        	\bm{\Phi}_k^* =
			\begin{bmatrix}
			\bm{A}_{1,k}^* & \bm{A}_{2,k}^* & \cdots & \bm{A}_{p-1,k}^* & \bm{A}_{p,k}^* \\
			\bm{I}_N & \bm{0} & \cdots & \bm{0} & \bm{0} \\
			\bm{0} & \bm{I}_N & \cdots & \bm{0} & \bm{0} \\
			\vdots & \vdots & \ddots & \vdots & \vdots \\
			\bm{0} & \bm{0} & \cdots & \bm{I}_N & \bm{0}
			\end{bmatrix}, \quad
			\widetilde{\bm \Phi}_k^* =
			\begin{bmatrix}
			\bm{I}_{Np} & \bm{\Phi}_k^* & \bm{\Phi}_k^{*2} & \cdots & \bm{\Phi}_k^{*T_k - p-1} & \cdots \\
			\bm{0} & \bm{I}_{Np} & \bm{\Phi}_k^* & \cdots & \bm{\Phi}_k^{*T_k-p - 2} & \cdots \\
			\bm{0} & \bm{0} & \bm{I}_{Np} & \cdots & \bm{\Phi}_k^{*T_k-p - 3} & \cdots \\
			\vdots & \vdots & \vdots & \ddots & \vdots & \ddots \\
			\bm{0} & \bm{0} & \bm{0} & \cdots & \bm{I}_{Np} & \cdots
			\end{bmatrix},
        \end{equation}
    $ \bm e_{t,k} =(\bbm \varepsilon_{t,k}^\top, \bm 0^\top, \ldots, \bm 0^\top)^\top \in \mathbb{R}^{Np}$, $\widetilde{\bm e}_k = (\bm e_{T_k,k}^{\top}, \bm e_{T_k-1,k}^{\top},\ldots, \bm e_{1,k}^{\top})\in \mathbb{R}^{Np T_k}$ and $\widetilde{\bm x}_k = (\bm x_{T_k}^\top,  \bm x_{T_k-1,k}^\top,\\ \ldots, \bm x_{1,k}^\top  )^\top \in \mathbb{R}^{Np T_k}$. Thus, by Assumption \ref{asmp:2}, it follows that
    \begin{equation}\label{eq:VMAinf}
    	\widetilde{\bm e}_k = \widetilde{\bm \Sigma}_{\bbm \varepsilon,k}^{1/2}\widetilde{\bbm \zeta}_{k},
    \end{equation} 
        where $\widetilde{\bm \Sigma}_{\bbm \varepsilon,k} = \bm I_{pT_k} \otimes \bbm \Sigma_{\bbm \varepsilon,k}$ and
        $\widetilde{\bbm \zeta}_k = (\overline{\bbm \zeta}_{T_k-1,k}, \ldots, \overline{\bbm \zeta}_{p,k}^{\top})$ with $\overline{\bbm \zeta}_{t,k} = (\bbm \zeta_{t,k}^\top,\bm 0^\top,\ldots, \bm 0^\top)^\top \in \mathbb{R}^{Np}$.

	We now characterize the temporal and cross-sectional dependence implied by the VMA($\infty$) representation.
For each $k \in \{0\} \cup [K]$, define the characteristic matrix polynomial 
$$\widetilde{\bm \Xi}_k(z) = \bm{I}_{Np} - \bm{\Phi}_k^* z,$$ 
with $z \in \mathbb{C}.$ Recall that
	$
	\mu_{\min}(\widetilde{\bm \Xi}_k) = \min_{|z| = 1} \lambda_{\min}\big(\widetilde{\bm \Xi}_k^\dagger(z)\widetilde{\bm \Xi}_k(z)\big)$ and 
	$\mu_{\max}(\widetilde{\bm \Xi}_k) =\max_{|z| = 1} \\\lambda_{\max}\big(\widetilde{\bm \Xi}_k^\dagger(z)\widetilde{\bm \Xi}_k(z)\big),
	$
	which differ from $\mu_{\min}(\bm \Xi_k)$ and $\mu_{\max}(\bm \Xi_k)$ defined earlier in \cite{basu2015regularized}.
	For each $k \in \{0\} \cup [K]$, define the task-specific quantities:
	$$
    \widetilde{L}_k = \frac{3\lambda_{\max}(\bm \Sigma_{\bbm \varepsilon,k})}{\mu_{\min}(\widetilde{\bm \Xi}_k)}
    \quad\text{and}\quad
	\widetilde{\kappa}_k = \frac{\widetilde{L}_k}{\rho_k}.
	$$
	We also define the global quantities: $\rho = \min_{k \in [K]}   \{\rho_k\},  
    L = \max_{k \in [K]} \{L_k\},
    \widetilde{L} = \max_{k \in [K]} \{\widetilde{L}_k\}, $
	$$
    \widetilde{T} = \min_{k\in [K]}\left\{\frac{T_k}{w_k\lambda_{\max}(\bm \Sigma_{\bbm \varepsilon,k})}\right\},
    \quad \kappa = \frac{L}{\rho}
    \quad\text{and}\quad
    \widetilde{\kappa} = \frac{\widetilde{L}}{\rho}.
    $$

\subsection{Definitions}\label{sec:definitions}

We next introduce the regularity conditions and deviation quantities used in the non-asymptotic error bounds for the target estimator $ \widehat{\bbm \alpha}_0 $ under a shared low-rank representation.
We assume that each loss function $ \mathcal{L}_k: \mathbb{R}^{N^2p} \to \mathbb{R} $.

\begin{definition}[Regularity and $h$-relatedness]
A function $ \mathcal{L} $ is said to be $ (\bbm \beta, \rho, L, \eta) $-regular if 
($i$) $ \mathcal{L} $ is convex and twice differentiable;
($ii$) $ \rho \bm I \preceq \nabla^2 \mathcal{L}(\bbm \beta) \preceq L \bm I $; 
($iii$) $ \| \nabla \mathcal{L}(\bbm \beta) \|_2 \le \eta $.
Moreover, $ \{\mathcal{L}_k,k\in\{0\} \cup [K]\} $ is said to be $h$-related with regularity parameters $ (\bm B^*,\{\bbm \alpha_k^*,\bm d_k^* , \eta_k,k\in [K] \},\rho, L) $, if for each $ k \in \{0\} \cup [K] $, $ \mathcal{L}_k $ is $ (\bm B^* \bm d_k^*, \rho, L, \eta_k) $-regular and $ \| \bbm \alpha_k^* - \bm B^* \bm d_k^* \|_2 \le h $.
\end{definition}
To control the statistical error, we next introduce the following deviation quantities:
\begin{align*}
 	&    \left\| \bm P^* \nabla \mathcal{L}_k(\bm B^* \bm d_k^*) \right\|_2 \le \tau_k, \quad
	\text{and} \quad \sup_{\bm W \in \bm \Omega(2s_1, 2s_2, 2s_3; N, K, p)} 
        \sum_{k\in [K]} \left\langle w_k^{1/2} \nabla \mathcal{L}_k (\bm B^* \bm d_k^*), \bm w_k \right\rangle   \le \gamma,
\end{align*}
where $ \bm P^* = \bm B^* \bm B^{*\top}$ and
\begin{align*}
    \bm \Omega(s_1, s_2, s_3; N, K, p) = \Big\{\bm W = &[\bm w_1 \; \cdots \; \bm w_K]: ~ \bm U \in \mathbb{O}^{N \times s_1}, \bm V \in \mathbb{O}^{N \times s_2}, \bm L \in \mathbb{O}^{p \times s_3},\\
    &~\bm d_k \in \mathbb{R}^{s_1 s_2 s_3},~\bm w_k = (\bm U\otimes \bm L \otimes \bm V) \bm d_k, ~ \|\bm W_k\|_\F = 1 \Big\}.
\end{align*}

To facilitate the technical analysis, we define a matrix set that characterizes structured low-dimensional directions, along with its covering and decomposition properties.
\begin{definition}[Matrix Set $\bm \Theta(\epsilon, d, q)$] \label{definition:matrix_set}

Define the matrix set
$$
\bm \Theta(\epsilon, d, q) = \left\{ \bm{W} \in \bm \Theta(\epsilon, d, q) \mid \|\bm{W}\|_\mathrm{F} = 1 \right\}.
$$
Let $\overline{\bm \Theta}$ be its $\epsilon$-covering. For any $\bm{W} \in \bm \Theta(\epsilon, d, q)$, there exists $\overline{\bm{W}} \in \overline{\bm \Theta}$ such that
\begin{itemize}
\item[($i$)] $\|\bm{W} - \overline{\bm{W}}\|_\F \leq \epsilon$, and the covering number of $\overline{\bm \Theta}$ is at most $(C/\epsilon)^d$;
\item[($ii$)] $\bm{W} - \overline{\bm{W}} = \sum_{i=1}^q \bm{N}_i,$
where $\bm{N}_i/\|\bm{N}_i\|_\F \in \bm \Theta(\epsilon, d, q)$ and
$\|\bm{W} - \overline{\bm{W}}\|_\F^2 = \sum_{i=1}^q \|\bm{N}_i\|_\F^2.$
\end{itemize}

\end{definition}

\section{Proofs of Main Results}\label{sec:mainproof}

This section is organized as follows. Section \ref{proofs_theorems} includes proofs for Theorems~\ref{thm:rep}--\ref{thm:TL} and Propositions~\ref{prop:worse}--\ref{prop:PoolVAR}, together with auxiliary Propositions \ref{prop:determrep}--\ref{prop:determTLpool}.  Section~\ref{sec:proofofprop} contains the proofs of Propositions~\ref{prop:determrep}--\ref{prop:determTLpool}.  Section~\ref{sec:proofoflemma} states the supporting Lemmas \ref{lemma:deterministic1}--\ref{lemma:uniondeviation} needed for Propositions~\ref{prop:determrep}--\ref{prop:determTLpool}, and provides their proofs.

Specifically, Propositions~\ref{prop:determrep}--\ref{prop:probeventsrep} establish the key deterministic and probabilistic results used to prove Theorem~\ref{thm:rep}. Propositions~\ref{prop:determTL}--\ref{prop:probeventsTL} provide the corresponding transfer-learning guarantees for the target task, which are used to prove Theorem~\ref{thm:TL}. %When the similarity requirements on $h_{\mathrm{S}}$ or $h_{\mathrm{T}}$ are not met, Proposition~\ref{prop:worse} provides a fallback bound for the target estimator, following Lemma~\ref{lemma:personalization} \citep{duan2023adaptive}.
Proposition~\ref{prop:determTLpool} is stated as a special case of Proposition~\ref{prop:determTL} and underlies Proposition~\ref{prop:PoolVAR}.

\subsection{Proofs of Theorems \ref{thm:rep}--\ref{thm:TL} and Propositions \ref{prop:worse}--\ref{prop:PoolVAR}}\label{proofs_theorems} 

\begin{proof}[\bf Proof of Theorem \ref{thm:rep}]
	
	We first introduce two critical random events that form the basis of our subsequent deterministic analysis:
\begin{align*}
    \mathcal{E}_1 &= \left\{ k \in [K]: \rho \bm I_{N^2p} \preceq  \nabla^2 \mathcal{L}_k(\bbm \alpha_k^*)   \preceq L\bm I_{N^2p},~
    \left\| \nabla \mathcal{L}_k(\bbm \alpha_k^*) \right\|_2 \lesssim\eta_k^*\right\},
\end{align*}
and event $\mathcal{E}_2$ follows that 
\begin{equation}
        \sup_{\bm W \in \bm \Omega(2s_1, 2s_2, 2s_3; N, K, p)} \sum_{k \in [K]} \left\langle w_k^{1/2} \nabla \mathcal{L}_k (\bbm \alpha_k^*), \bm w_k \right\rangle \lesssim \gamma^*,
\end{equation}
where $\bm{W} = [ \bm{w}_1: \cdots : \bm{w}_K]$ 
and $\sum_{k \in [K]}  \| \bm{w}_k \|_{2}^2 = 1$.

\begin{proposition}[Representation Deterministic Analysis]\label{prop:determrep}
    Let events $\mathcal{E}_1$ and $\mathcal{E}_2$ hold.
    Suppose there exist constants $ C > 0 $ such that for each $ k \in [K] $, $2C L h_{\mathrm{S}} \leq \eta_k^* + \frac{\lambda_k}{2 C s_{\max}\kappa^{3/2} \mu^2 }$.
    If Assumption \ref{asmp:3} holds and 
    $
    2C \kappa^{3/2} \mu^2 s_{\max} \eta_k^* < \lambda_k \le \rho \alpha_1
    $, then 
    $$\Big\| \left\{w_k^{1/2} (\widehat{\bm B} \widehat{\bm d}_k  - \bm B^{*} \bm d_k^{*}) ,k \in [K]\right\} \Big\|_\F \lesssim \rho^{-1} (\gamma^* + Lh_{\mathrm{S}}) $$ 
    and let $s_{\max} = \max\{s_1,s_2,s_3\} $,
    $$\max\{\|\sin \Theta (\widehat{\bm U}, \bm U^* )\|_\F, \|\sin \Theta (\widehat{\bm V}, \bm V^* )\|_\F, \|\sin \Theta (\widehat{\bm L}, \bm L^* )\|_\F\}
		\lesssim  \frac{s_{\max}^{1/2}}{\rho \alpha_2}(\gamma^*  + L h_{\mathrm{S}}).$$
\end{proposition}

Since the above bound critically depends on events $\mathcal{E}_1$ and $\mathcal{E}_2$, the subsequent probabilistic analysis (Proposition \ref{prop:probeventsrep}) establishes that these events occur with high probability. 

\begin{proposition}[Probability of Random Events $\mathcal{E}_1$ and $\mathcal{E}_2$]\label{prop:probeventsrep}
    Let 
    % $\delta > s_1s_2s_3$,
    $$\eta_k^* =\sigma M_{k}\sqrt{\frac{N^2p+N\log (NK) }{T_k}}
    \quad 
    % ~\tau_k^* = \sigma M_{k}\sqrt{\frac{s_1s_2s_3+\log K}{T_k}}
\quad \text{and} \quad \gamma^* =\sigma L^{1/2}\sqrt{\frac{d_M}{\widetilde{T}}}.$$
    Under Assumptions \ref{asmp:1}--\ref{asmp:3}, if $ T_{\min}  \gtrsim[ Np+\log (K) ]\max(\widetilde{\kappa}^{2} \sigma^{4}, \widetilde{\kappa} \sigma^{2})$,
    then events $\mathcal{E}_1$ and $\mathcal{E}_2$ hold with probability at least $
	1 -\exp\left(-C_1 [Np+\log (NK) ]\right) - 2\exp(-C_2 T_{\min} \min\left( \widetilde{\kappa}^{-2} \sigma^{-4}, \widetilde{\kappa}^{-1} \sigma^{-2} \right))
	$. 
    \end{proposition}

With Propositions \ref{prop:determrep}--\ref{prop:probeventsrep} established, the proof of Theorem~\ref{thm:rep} is complete.  

\end{proof}

\begin{proof}[\bf Proof of Theorem \ref{thm:TL}]
The proof of Theorem \ref{thm:TL} builds upon the results established in Theorem \ref{thm:rep}. We first introduce another two random events related to the target dataset that form the basis of our subsequent deterministic analysis:
\begin{align*}
    \mathcal{E}_3 = \left\{ \rho \bm I_{N^2p} \preceq  \nabla^2 \mathcal{L}_0(\bbm \alpha_0^*)   \preceq L\bm I_{N^2p},
    \left\| \nabla \mathcal{L}_0(\bbm \alpha_0^*) \right\|_2 \lesssim\eta_0^* \right\}\quad \text{and}\quad
    \mathcal{E}_4 = \left\{\left\| \bm P^* \nabla \mathcal{L}_0(\bbm \alpha_0^*) \right\|_2 \lesssim\tau_0^*\right\}.
\end{align*}

\begin{proposition}[Transfer Learning Deterministic Analysis]\label{prop:determTL}
    Suppose there exist constants $ C > 0 $ such that $2C L h_{\mathrm{T}} \leq \eta_0^* + \frac{\lambda_0}{2 C s_{\max}\kappa^{3/2} \mu^2 }$. 
    If
    $
    2C \kappa^{3/2} \mu^2 s_{\max} \eta_0^* < \lambda_0 \le \rho \alpha_1
    $ and the conditions of Proposition \ref{prop:determrep}, events $\mathcal{E}_3$ and $\mathcal{E}_4$ hold, then
    $$
    \| \widehat{\bbm \alpha}_0  - \bbm \alpha^{*}_0 \|_2 
     \lesssim 
    s_{\max}^{1/2} \kappa^{3/2} \mu  \gamma^* 
    + \frac{ \tau_0^*  }{\rho}
    +   s_{\max} \kappa^{5/2} \mu^2  \max\{h_{\mathrm{S}},h_{\mathrm{T}}\}.
    $$
\end{proposition}

\begin{proposition}[Probability of Random Events $\mathcal{E}_3$ and $\mathcal{E}_4$]\label{prop:probeventsTL}
    Let 
    % $\delta > s_1s_2s_3$, 
    $$\eta_0^* =\sigma M_{0}\sqrt{\frac{N^2p+N\log (N) }{T_0}}
    \quad \text{and} \quad
    \tau_0^* = \sigma M_{0}\sqrt{\frac{s_1s_2s_3}{T_0}}.$$
    Under Assumptions \ref{asmp:1}--\ref{asmp:3}, if $T_0 \gtrsim [Np+\log(N)]\max ( \widetilde{\kappa}^{2}\sigma^{4}, \widetilde{\kappa} \sigma^{2} )$, 
    then events $\mathcal{E}_3$ and $\mathcal{E}_4$ hold with probability at least $
	1 - \exp\left(-C_1 [Np+\log N ]\right) - 2\exp(-C_2 T_0 \min\left( \widetilde{\kappa}^{-2} \sigma^{-4}, \widetilde{\kappa}^{-1} \sigma^{-2} \right))
	$.
    \end{proposition}
With Propositions \ref{prop:probeventsrep}--\ref{prop:probeventsTL} established, the proof of Theorem~\ref{thm:TL} is now complete.

\end{proof}

\begin{proof}[\bf Proof of Proposition \ref{prop:worse}]
    In the case of worst transfer learning, $h_{\mathrm{S}}$ or $h_{\mathrm{T}}$ is large that can not satisfy the requirements in Theorems \ref{thm:rep}--\ref{thm:TL}. 
    By Lemma \ref{lemma:personalization}, we obtain that
    \begin{align*}
    \| \widehat{\bbm \alpha}_0  - \bbm \alpha^{*}_0 \|_2 
     \leq \frac{\eta_0^* + \lambda_0}{\rho} \leq \frac{3\lambda_0}{2\rho} 
     \leq \frac{3C s_{\max} \kappa^{3/2} \mu^2 }{\rho}.
    \end{align*}
    This completes the proof.
\end{proof}

\begin{proof}[\bf Proof of Proposition \ref{prop:PoolVAR}]
    Proposition \ref{prop:PoolVAR} is a special case of Theorems \ref{thm:rep} and \ref{thm:TL} when $h_{\mathrm{S}} = h_{\mathrm{T}} = 0$ and $\lambda_k = \infty$ hold for all $k$. In this case, we have $\bbm \alpha_k^* = \bm B^* \bm d_k^*$ for all $k$ with a common representation space. 
    To prove Proposition \ref{prop:PoolVAR}, it is sufficient to show Propositions \ref{prop:determrep}--\ref{prop:probeventsTL} under these conditions. It is easy to show that Propositions \ref{prop:determrep}, \ref{prop:probeventsrep} and \ref{prop:probeventsTL} hold for this case.   
    Moreover, Proposition \ref{prop:determTL} can be simplified as in Proposition \ref{prop:determTLpool} below, which can be derived directly by Claim \ref{claim-lowrank-subspace-sharp} and Lemma \ref{lemma:deterministic2}.
    In addition, by the same arguments as in Claim \ref{claim:pooldeviation}, we can show event $\mathcal{E}_4$ in Proposition \ref{prop:determTLpool} holds with high probability. Thus, the proof of Proposition \ref{prop:PoolVAR} is complete. 
     
    \begin{proposition}\label{prop:determTLpool}
    Suppose the conditions of Theorem \ref{prop:PoolVAR} and event $\mathcal{E}_4$ hold,
    then we have 
    $$
    \| \widehat{\bm B} \widehat{\bm d}_0  - \bm B^* \bm d_0^* \|_2 
     \lesssim 
    s_{\max}^{1/2} \kappa^{3/2} \mu  \gamma^*  + \frac{ \tau_0^*  }{\rho}.
    $$
    \end{proposition}
   
\end{proof}

%\subsection{Lemmas \ref{lemma:deterministic1}--\ref{lemma:uniondeviation}}\label{sec:lemmas}

\subsection{Proofs of Propositions \ref{prop:determrep}--\ref{prop:probeventsTL}}\label{sec:proofofprop}

\begin{proof}[\bf Proof of Proposition 
    \ref{prop:determrep}]

    Note that $\nabla \mathcal{L}_k(\bbm \alpha_k) = T_k^{-1} \bm Z_k^{\top} (\bm y_k - \bm Z_k \bbm \alpha_k)$ and $\nabla^2 \mathcal{L}_k(\bbm \alpha_k) = T_k^{-1} \bm Z_k^{\top} \bm Z_k \preceq L \bm I_{N^2p} $. 
    For any $ k \in [K] $, we have
    $$
     \| \nabla \mathcal{L}_k (\bbm \alpha_k^*) - \nabla \mathcal{L}_k (\bm B^* \bm d_k^*) \|_2 =  \| T_k^{-1} \bm Z_k^\top \bm Z_k (\bbm \alpha_k^* - \bm B^* \bm d_k^*) \|_2 
    \leq \| T_k^{-1} \bm Z_k^\top \bm Z_k \|_2  \| \bbm \alpha_k^* - \bm B^* \bm d_k^* \|_2 \leq L h_{\mathrm{S}}.
    $$
    By triangle's inequality, 
    \begin{equation}
        \begin{split}
             \|\nabla \mathcal{L}_k (\bm B^* \bm d_k^*) \|_2 
    \leq& \|\nabla \mathcal{L}_k (\bbm \alpha_k^*) \|_2  + 
    \|\nabla \mathcal{L}_k (\bbm \alpha_k^*) -\nabla \mathcal{L}_k (\bm B^* \bm d_k^*)\|_2 \\
    \leq& \|\nabla \mathcal{L}_k (\bbm \alpha_k^*) \|_2 + L h_{\mathrm{S}}.
        \end{split}
    \end{equation}
   Moreover, under the conditions of Proposition \ref{prop:determrep}, we have
    \begin{align}
    2 C  s_{\max} \kappa^{3/2} \mu^2\eta_k^* < \lambda_k 
    \quad \text{and} \quad 
    2 C L h_{\mathrm{S}} \leq \eta_k^* + \frac{\lambda_k}{2 C s_{\max}\kappa^{3/2} \mu^2 },
    \end{align}
    Consequently, it holds that
   \begin{equation}
    \begin{split}
        C s_{\max}  \kappa^{3/2} \mu^2 \eta_k
        \le &  C s_{\max}  \kappa^{3/2} \mu^2 \eta_k^* + C s_{\max}  \kappa^{3/2} \mu^2  Lh_{\mathrm{S}}  \\
        \le & (C+1/2) s_{\max}  \kappa^{3/2} \mu^2 \eta_k^* + \frac{\lambda_k}{4C}  
        \le  \lambda_k.
    \end{split}
   \end{equation}
    In addition, we have
    \begin{equation}
    \begin{split}
	\gamma \leq& \sup_{\bm W \in \bm \Omega(2s_1, 2s_2,2s_3; N, K, p)} \sum_{k\in [K]} \left\langle w_k^{1/2} \nabla \mathcal{L}_k (\bm A_k^*), \bm w_k \right\rangle \\
    &+ \sup_{\bm W \in \bm \Omega(2s_1, 2s_2, 2s_3; N, K, p)} \sum_{k\in [K]} \left\langle w_k^{1/2} (\nabla \mathcal{L}_k (\bm B^* \bm d_k^*) - \nabla \mathcal{L}_k (\bm A_k^*)), \bm w_k \right\rangle \\
    \leq& \gamma^* + \left\| \{ w_k^{1/2} [ \nabla \mathcal{L}_k (\bm B^* \bm d_k^*) - \nabla \mathcal{L}_k (\bm A_k^*) ],k \in [K] \} \right\|_\F
    \leq  \gamma^* + L h_{\mathrm{S}}.
    \end{split}
    \end{equation}
    This together with Lemma \ref{lemma:deterministic1}, implies that
    \begin{align*}
\Big\| \left\{w_k^{1/2} (\widehat{\bm B} \widehat{\bm d}_k  - \bm B^{*} \bm d_k^{*}) ,k \in [K]\right\} \Big\|_\F \lesssim \rho^{-1}(\gamma^*  + L h_{\mathrm{S}})
    \end{align*}
    and 
    $$\max\{\|\sin \Theta (\widehat{\bm U}, \bm U^* )\|_\F, \|\sin \Theta (\widehat{\bm V}, \bm V^* )\|_\F, \|\sin \Theta (\widehat{\bm L}, \bm L^* )\|_\F\}
		\lesssim  \frac{ s_{\max}^{1/2}}{\rho \alpha_2}(\gamma^*  + L h_{\mathrm{S}}),$$ 
    %     and
    % $$\|\widehat{\bm B}\widehat{\bm B}^{\top} - \bm B^*\bm B^{*\top}\|_2  \lesssim  \frac{6\sqrt{s_{\max}}}{\rho \alpha_2}(\gamma^*  + L h),$$
    which completes the proof.
\end{proof}

\begin{proof}[\bf Proof of Proposition \ref{prop:probeventsrep}]

    By Assumptions \ref{asmp:1}--\ref{asmp:2} and the conditions of Proposition \ref{prop:probeventsrep}, the proof of Proposition \ref{prop:probeventsrep} is done by showing the following Claims \ref{claim:singledeviation}--\ref{claim:pooldeviation}.

    \begin{claim}[Event $\mathcal{E}_1$] \label{claim:singledeviation} 
         If $T_{\min}  \gtrsim [Np+\log (K) ]\max ( \widetilde{\kappa}^{2}\sigma^{4}, \widetilde{\kappa} \sigma^{2} )$ with $T_{\min} = \min_{k\in [K]}\{T_k\}$, then with probability at least
        $
        1 - \exp\left(-C_1 [Np+\log (NK) ]\right) - 2\exp(-C_2 T_{\min} \min\left( \widetilde{\kappa}^{-2} \sigma^{-4}, \widetilde{\kappa}^{-1} \sigma^{-2} \right)),
        $ for all $k \in [K]$, we have 
         \begin{equation}
           \rho \bm I_{N^2p} \preceq  \nabla^2 \mathcal{L}_k(\bbm \alpha_k^*)   \preceq L \bm I_{N^2p} 
         \quad \text{and} \quad
            \left\| \nabla \mathcal{L}_k(\bbm \alpha_k^*) \right\|_2 \lesssim\sigma M_{k}\sqrt{\frac{N^2p+N\log (N K)}{T_k}}.
        \end{equation}
    \end{claim}
    \begin{claim}[Event $\mathcal{E}_2$]\label{claim:deviationMTL} 
         If $T_{\min} \gtrsim [s_1s_2s_3 + \log(K)]\max ( \widetilde{\kappa}^{2}\sigma^{4}, \widetilde{\kappa} \sigma^{2} )$, then with probability at least
        $
        1  - 2\exp\left(-C T_{\min} \min\left( \widetilde{\kappa}^{-2} \sigma^{-4}, \widetilde{\kappa}^{-1} \sigma^{-2} \right)\right),
        $ for all $k \in [K]$, we have
         \begin{equation}
            \left\| \bm P^* \nabla \mathcal{L}_k(\bbm \alpha_k^*) \right\|_2 \lesssim\sigma M_{k}\sqrt{\frac{s_1s_2s_3+\log K}{T_k}}.
        \end{equation}

    \end{claim}

    \begin{claim}[Event $\mathcal{E}_3$]\label{claim:pooldeviation} 
    If 
    $
        T_{\min} \gtrsim [ N(s_1 \vee s_2) + \log K]\max\left( \widetilde{\kappa}^{2} \sigma^{4}, \widetilde{\kappa} \sigma^{2} \right),
    $
    then with probability at least
    $
        1 - \exp\left( -C_1 d_M \right) - 2\exp( -C_2 T_{\min} \min( \widetilde{\kappa}^{-2} \sigma^{-4}, \widetilde{\kappa}^{-1} \sigma^{-2} )),
    $
    we have
      
    \begin{equation}
        \sup_{\bm W \in \bm \Omega(2s_1, 2s_2, 2s_3; N, K, p)} \sum_{k \in [K]}w_k^{1/2} \left\langle  \nabla \mathcal{L}_k (\bbm \alpha_k^*), \bm w_k \right\rangle \lesssim \sigma M_{k}\sqrt{\frac{d_M}{\widetilde{T}}}.
    \end{equation}
    \end{claim}

    \begin{proof}[\bf Proof of Claim \ref{claim:singledeviation}]
            First, we will show that for all $k\in  [K]$, it follows that $$\rho  \bm I_{N^2p} \preceq  \nabla^2 \mathcal{L}_k(\bbm \alpha_k^*)   \preceq L \bm I_{N^2p}.$$
            Note that $T_k\nabla^2 \mathcal{L}_k(\bbm \alpha_k^*)  = \bm{Z}_k^\top \bm{Z}_k =  \bm{I}_N\otimes (\bm{X}_k \bm{X}_k^\top) $. It follows that 
        \begin{align*}
            \lambda_{\min}(\bm{X}_k \bm{X}_k^\top) = \lambda_{\min}(\bm{I}_N\otimes (\bm{X}_k \bm{X}_k^\top)) \le \lambda_{\max}(\bm{I}_N\otimes (\bm{X}_k \bm{X}_k^\top)) = \lambda_{\max}(\bm{X}_k \bm{X}_k^\top).
        \end{align*}
        Therefore, it suffices to show that, for all $k \in [K]$,
        \begin{align*}
            \rho \le T_k^{-1}\lambda_{\min}(\bm{X}_k \bm{X}_k^\top) \le T_k^{-1}\lambda_{\max}(\bm{X}_k \bm{X}_k^\top) \le L.
        \end{align*}
        For any $\bm{w}_k \in \mathbb{R}^{Np}$ with $\|\bm{w}_k\|_2 = 1$, we have
        \begin{align*}
            &\lambda_{\min}(\bm{X}_k \bm{X}_k^\top) 
            = \inf_{\bm{w} \in \mathbb{R}^{Np}, \|\bm{w}\|_2 = 1} \bm{w}^\top \bm{X}_k \bm{X}_k^\top \bm{w} 
            \le& \bm{w}_k^\top \bm{X}_k \bm{X}_k^\top \bm{w}_k 
           \le \sup_{\bm{w} \in \mathbb{R}^{Np}, \|\bm{w}\|_2 = 1} \bm{w}^\top \bm{X}_k \bm{X}_k^\top \bm{w} 
            = \lambda_{\max}(\bm{X}_k \bm{X}_k^\top).
        \end{align*}
        In addition, for any $\bm w_k$, there exists $\overline{\bm w}_k \in \overline{\bm \Theta}$ such that
            \begin{itemize}
                \item $\|\bm{w}_k - \overline{\bm{w}}_k\|_\F \leq \epsilon$, and the covering number of $\overline{\bm \Theta}$ is at most $(3/\epsilon)^{Np}$, by Lemma \ref{lemma:covering};
                \item $\bm{w}_k - \overline{\bm{w}}_k =  \bm{n}_k,$
                where $\bm{n}_k/\|\bm{n}_k\|_\F \in \bm \Theta(\epsilon, Np, 1)$ and
                $\|\bm{w}_k - \overline{\bm{w}}_k\|_\F^2 =  \|\bm{n}_k\|_\F^2$.
            \end{itemize}
        Thus, $\bm w_k\in \bm \Theta(\epsilon,Np,1)$. 
    %    Let $\overline{\mathcal{M}}$ be an $\epsilon$-covering of the unit sphere in $\mathbb{R}^{Np}$ with covering number at most $(3/\epsilon)^{Np}$ by Lemma \ref{lemma:covering}. Then for any $\bm{w}_k$, there exists $\overline{\bm{w}}_k \in \overline{\mathcal{M}}$ such that $\|\bm{w}_k - \overline{\bm{w}}_k\|_2 \le \epsilon$.
        Combining this with Lemma \ref{lemma:RSCRSS},  it follows that when $T_{\min} \gtrsim (Np+\log (K))\max ( \widetilde{\kappa}^{2}\sigma^{4}, \widetilde{\kappa} \sigma^{2} )$, then with probability at least $1-2\exp(-C T_{\min}
        \min ( \widetilde{\kappa}^{-2}\sigma^{-4},  \widetilde{\kappa}^{-1} \sigma^{-2}))$, for all $k\in [K]$, we have
        \begin{equation}\label{eq:singleRSCRSS-2}
           \frac{\rho}{2}  \le  T_k^{-1}\sum_{t=1}^{T_k}\|\bm w_{k}^\top  \bm{x}_{t,k}\|_2^2  \le \frac{L}{2}  .
        \end{equation}
        
        Second, let $S_k(\bm w_k) = T_k^{-1}\sum_{t=1}^{T_k}\langle\bbm \varepsilon_{t,k} \bm x_{t,k}^\top, \bm w_k \rangle $ and $R_k(\bm w_k) = T_k^{-1}\sum_{t=1}^{T_k}\|\bm w_k^\top \bm x_{t,k}\|_2^2$. Separate $ \mathcal{L}_k(\bbm \alpha_k^*)$ into $N$ parts such that $\mathcal{L}_k(\bbm \alpha_k^*) = \sum_{i=1}^N \mathcal{L}_{i,k}(\bbm \alpha_{i,k}^*) = T_k^{-1} \sum_{i=1}^N\sum_{t=1}^{T_k}  \varepsilon_{it,k}^2$, where $\mathcal{L}_{i,k}(\bbm \alpha_{i,k}^*) = T_k^{-1}\sum_{t=1}^{T_k}  \varepsilon_{it,k}^2 $ with $\varepsilon_{it,k} =  y_{it,k} - (\bm I\otimes \bm x_t^\top )\bbm \alpha_{i,k}$, and $\bbm \alpha_{i,k}$ is the $i$-th row of $\cm A_{k(1)}$. Then by Lemma \ref{lemma:deviation} and union bound, we have
        \begin{equation}
            \begin{split}
                &\mathbb{P}\left[ \bigcup_{k\in [K]} \left\{ \{\|  \nabla \mathcal{L}_k(\bbm \alpha_k^*)\|_\F \geq t_k\} \cap\{R_k(\bm w_k)\leq  L_k \}\right\}\right] \\
               \leq & \mathbb{P}\left[ \bigcup_{k\in [K]} \left\{ \{\max_{1\le i\le N}\|  \nabla \mathcal{L}_{i,k}(\bbm \alpha_{k,i}^*)\|_\F^2 \geq t_k^2/N \} \cap\{R_k(\bm w_k)\leq  L_k \} \right\}\right] \\
               \leq & \sum_{k\in [K]}\mathbb{P}\left[\{\sup_{\overline{\bm w}_k\in\overline{\bm \Theta}} \max_{1\le i\le N}\left\langle  \nabla \mathcal{L}_{i,k}(\bbm \alpha_{k,i}^*),\bm w_k\right\rangle\geq t_k/\sqrt{N}\} \cap\{R_k(\bm w_k)\leq  L_k \} \right]\\
                \leq & \sum_{k\in [K]}\mathbb{P}\left[\{\max_{\overline{\bm w}_k\in\overline{\bm \Theta}} \max_{1\le i\le N}\left\langle  \nabla \mathcal{L}_{i,k}(\bbm \alpha_{k,i}^*),\bm w_k\right\rangle\geq (1-\epsilon)t_k/\sqrt{N}\} \cap\{R_k(\bm w_k)\leq  L_k \} \right]\\
                    \leq & \sum_{k\in  [K]} |\overline{\bm \Theta}|\cdot \Bigg\{\mathbb{P}[\{ \max_{1\le i\le N}S_k(\bm w_k)\geq (1-\epsilon)t_k/\sqrt{N}\}\cap\{R_k(\bm w_k)\leq L_k \}]\Bigg\} \\
                    \leq & KN|\overline{\bm \Theta}|\cdot\exp\left[-\frac{(1-\epsilon)^2 t_k^2 T_k}{\sigma^2N \lambda_{\max}(\bm \Sigma_{\bbm{\varepsilon},k})L}\right].
            \end{split}
        \end{equation}
This together with \eqref{eq:singleRSCRSS-2} and the law of total probability, implies that
            \begin{equation}
                \begin{split}
                &\mathbb{P}\left[ \bigcup_{k\in [K]} \left\{ \|  \nabla \mathcal{L}_k(\bbm \alpha_k^*)\|_\F \geq t_k\right\}\right]\\
                 \le & \mathbb{P}\left[ \bigcup_{k\in [K]} \left\{ \{\max_{1\le i\le N}\|  \nabla \mathcal{L}_{i,k}(\bbm \alpha_{k,i}^*)\|_\F^2 \geq  t_k^2/N \} \cap\{R_k(\bm w_k)\leq  L_k \} \right\}\right]
                  +\mathbb{P}\Bigg[ \bigcup_{k\in [K]} \Bigg\{ R_k(\bm w_k) >  L_k\Bigg\} \Bigg]\\
                    \leq & \sum_{k\in  [K]}  |\overline{\bm \Theta}|\cdot\Bigg\{\sum_{i\in [N]}\exp\left[-\frac{(1-\epsilon)^2 T_k t_k^2}{\sigma^2N \lambda_{\max}(\bm \Sigma_{\bbm{\varepsilon},k})L}\right] \Bigg\}
                    + 2\exp\left(-C  T_{\min} \min (  \widetilde{\kappa}^{-2}\sigma^{-4},  \widetilde{\kappa}^{-1} \sigma^{-2})\right).
                \end{split}
            \end{equation}
       Take $\epsilon=2^{-1}$ and $t_k=C\sigma M_{2} \sqrt{[N^2p + N\log (NK) ]/T_k}$. When $T_{\min} \gtrsim [Np+\log (K)]\max ( \widetilde{\kappa}^{2}\sigma^{4}, \widetilde{\kappa} \sigma^{2} ) $, we have
            \begin{align*}
                &\mathbb{P}\left[ \bigcup_{k\in [K]} \left\{\| \nabla \mathcal{L}_k(\bbm \alpha_k^*)\|_\F \gtrsim \sigma M_{2}\sqrt{\frac{N^2p+N\log (NK)}{T_k}} \right\}\right]\\
                \leq & \exp(-C_1 [Np+\log (NK)])+2\exp(-C_2  T_{\min}  \min (  \widetilde{\kappa}^{-2}\sigma^{-4},  \widetilde{\kappa}^{-1} \sigma^{-2})).
            \end{align*}
           This completes the proof.

    \end{proof}
    
   \begin{proof}[\bf Proof of Claim \ref{claim:deviationMTL}]
        For any $k \in[K]$ and fixed $\bm P^* = \bm B^* \bm B^{*\top}$ with $\bm B^{*\top} \bm B^* = \bm I_{s_1s_2s_3},$ by dual property of $\ell_2$-norm, it follows that 
        \begin{align*}
            \left\| \bm P^* \nabla \mathcal{L}_k(\bbm \alpha_k^*) \right\|_2 =\left\| \bm B^{*\top} \nabla \mathcal{L}_k(\bbm \alpha_k^*) \right\|_2 = \sup_{\bm w \in \mathbb{R}^{s_1s_2 s_3}} \langle   \nabla \mathcal{L}_k(\bbm \alpha_k^*), \bm B^{*}\bm w\rangle.
        \end{align*}
        Together with $\bm w \in \bm \Theta(\epsilon, s_1s_2s_3, 1)$, the proof is complete by Lemma \ref{lemma:deviation}.
   \end{proof}
    
    \begin{proof}[\bf Proof of Claim \ref{claim:pooldeviation}]

    This claim follows directly from Lemma~\ref{lemma:uniondeviation}. Specifically, note that     	
     \begin{align*}
            \sup_{\bm W \in \bm \Omega(2s_1, 2s_2, 2s_3; N, K, p)}  \sum_{k \in [K]} w_k^{1/2} \left\langle \nabla \mathcal{L}_k (\bbm \alpha_k^*), \bm w_k  \right\rangle
            = \sup_{\cm W\in\bm \Theta(\epsilon, d_M, 8)}\sum_{k\in [K]} w_k^{1/2} \left\langle  \nabla \mathcal{L}_k(\cm{A}_k^*),\cm W_k\right\rangle,
        \end{align*}
        where $d_M = N(s_1+s_2)+ps_3 + Ks_1s_2s_3$, $\cm{W} = [ \cm{W}_1: \cdots : \cm{W}_K]$ with each $\cm{W}_k$ of Tucker ranks $(2s_1, 2s_2, 2s_3)$ and $\sum_{k \in [K]} \| \cm{W}_k \|_{\textup{F}}^2 = 1$. The desired result thus follows immediately from Lemma~\ref{lemma:uniondeviation}.    \end{proof}
\end{proof}

\begin{proof}[\bf Proof of Proposition 
    \ref{prop:determTL}]

By Claim \ref{claim-lowrank-crude} and \eqref{eq:representation-10}, it holds that $\mathcal{L}_0 (\widehat{\bm B} \widehat{\bm d}_0 ) =
\widetilde{L}_0 (\widehat{\bm B} \widehat{\bm d}_0 ) \leq \mathcal{L}_0 ( \widehat{\bm P} \bm B^* \bm d^*_0 ) $.
By the strong smoothness and convexity of $\mathcal{L}_0$ near $\bm B^* \bm d^*_0$, we have
\begin{align*}
& \mathcal{L}_0 ( \widehat{\bm P} \bm B^* \bm d^*_0 )  \leq \mathcal{L}_0 (  \bm B^* \bm d^*_0 ) 
+ \langle \nabla \mathcal{L}_0 (  \bm B^* \bm d^*_0 )  , (\widehat{\bm P} - \bm I) \bm B^* \bm d^*_0  \rangle + \frac{L}{2} \| (\widehat{\bm P} - \bm I) \bm B^* \bm d^*_0   \|_2^2 ,\\
& \mathcal{L}_0 (\widehat{\bm B} \widehat{\bm d}_0 ) \geq \mathcal{L}_0 (  \bm B^* \bm d^*_0 ) 
+ \langle \nabla \mathcal{L}_0 (  \bm B^* \bm d^*_0 )  ,\widehat{\bm B} \widehat{\bm d}_0 - \bm B^* \bm d^*_0  \rangle + \frac{\rho}{2} \| \widehat{\bm B} \widehat{\bm d}_0  - \bm B^* \bm d^*_0   \|_2^2.
\end{align*}
Combining the above inequalities, we get
\begin{equation}\label{eqn-lowrank-sharp-1}
	\begin{split}
&\frac{\rho}{2} \| \widehat{\bm B} \widehat{\bm d}_0  - \bm B^* \bm d^*_0   \|_2^2 +
\langle \nabla \mathcal{L}_0 (  \bm B^* \bm d^*_0 )  , \widehat{\bm B} \widehat{\bm d}_0 - \widehat{\bm P} \bm B^* \bm d^*_0  \rangle\\
\leq&  \frac{L}{2} \| (\widehat{\bm P} - \bm I) \bm B^* \bm d^*_0   \|_2^2 
=  \frac{L}{2} \| (\widehat{\bm P} - \bm P^*) \bm B^* \bm d^*_0   \|_2^2\\
\le &  \frac{L}{2} \Big( \| \widehat{\bm P} - \bm P^* \|_2  \| \bm B^* \|_2 \| \bm d^*_0   \|_2 \Big)^2  \\
\leq& \frac{L}{2} \bigg(
\frac{6\alpha_1 s_{\max}^{1/2}\gamma}{\rho \alpha_2}
\bigg)^2,
	\end{split}
\end{equation}
where the last inequality follows from Claim \ref{claim-lowrank-subspace-sharp}, $\|\bm B^*\|_2 \le 1$ and $\| \bm d^*_0 \|_2 \leq \alpha_1$.

Since $\widehat{\bm P} \widehat{\bm B} \widehat{\bm d}_0  = \widehat{\bm B} \widehat{\bm d}_0 $, it holds that
$$
\widehat{\bm B} \widehat{\bm d}_0 - \widehat{\bm P} \bm B^* \bm d^*_0
= \widehat{\bm P} ( \widehat{\bm B} \widehat{\bm d}_0 - \bm B^* \bm d^*_0 )
=  \bm P^* ( \widehat{\bm B} \widehat{\bm d}_0 - \bm B^* \bm d^*_0 ) + (\widehat{\bm P}  - \bm P^*) ( \widehat{\bm B} \widehat{\bm d}_0 - \bm B^* \bm d^*_0 ).
$$
Recall that $\| \bm P^*  \nabla \mathcal{L}_0 (  \bm B^* \bm d^*_0 )  \|_2 \le \tau_0$ and $\| \nabla \mathcal{L}_0 (  \bm B^* \bm d^*_0 )  \|_2 \le \eta_0$. Then we have
\begin{align*}
| \langle \nabla \mathcal{L}_0 (  \bm B^* \bm d^*_0 )  , \bm P^* ( \widehat{\bm B} \widehat{\bm d}_0 - \bm B^* \bm d^*_0 ) \rangle |
& \leq \| \bm P^*  \nabla \mathcal{L}_0 (  \bm B^* \bm d^*_0 )  \|_2 \|  \widehat{\bm B} \widehat{\bm d}_0 - \bm B^* \bm d^*_0 \|_2  \leq  \tau_0 \|  \widehat{\bm B} \widehat{\bm d}_0 - \bm B^* \bm d^*_0 \|_2,
\end{align*}
and by Claim \ref{claim-lowrank-subspace-sharp}, we have
\begin{align*}
& | \langle \nabla \mathcal{L}_0 (  \bm B^* \bm d^*_0 )  , (\bm P^* - \widehat{\bm P} ) ( \widehat{\bm B} \widehat{\bm d}_0 - \bm B^* \bm d^*_0 ) \rangle | \\
& \leq \| \nabla \mathcal{L}_0 (  \bm B^* \bm d^*_0 )\|_\F \|  \widehat{\bm P} - \bm P^* \|_2 \| \widehat{\bm B} \widehat{\bm d}_0 - \bm B^* \bm d^*_0 \|_2 \\ 
 &\leq \eta_0 \frac{6s_{\max}^{1/2}\gamma}{\rho \alpha_2}  \| \widehat{\bm B} \widehat{\bm d}_0 - \bm B^* \bm d^*_0 \|_2.
\end{align*}
Combining the above two inequalities, we have
$$
 | \langle \nabla \mathcal{L}_0 (  \bm B^* \bm d^*_0 )  ,   \widehat{\bm B} \widehat{\bm d}_0 - \bm B^* \bm d^*_0  \rangle | \leq \bigg(
\tau_0 +  \eta_0 \frac{6s_{\max}^{1/2}\gamma}{\rho \alpha_2} 
\bigg)  \| \widehat{\bm B} \widehat{\bm d}_0 - \bm B^* \bm d^*_0 \|_2, 
$$
and thus,
\begin{align}
& \frac{\rho}{2} \| \widehat{\bm B} \widehat{\bm d}_0  - \bm B^* \bm d^*_0   \|_2^2 +
\langle \nabla \mathcal{L}_0 (  \bm B^* \bm d^*_0 )  , \widehat{\bm B} \widehat{\bm d}_0 - \widehat{\bm P} \bm B^* \bm d^*_0  \rangle  \notag\\
& \geq \frac{\rho}{2} \| \widehat{\bm B} \widehat{\bm d}_0  - \bm B^* \bm d^*_0   \|_2 \bigg[
\| \widehat{\bm B} \widehat{\bm d}_0  - \bm B^* \bm d^*_0   \|_2 - \frac{2}{\rho} \bigg(
\tau_0 +  \eta_0 \frac{6 s_{\max}^{1/2}\gamma}{\rho \alpha_2} 
\bigg)
\bigg].
\label{eqn-lowrank-sharp-2}
\end{align}

From \eqref{eqn-lowrank-sharp-1} and \eqref{eqn-lowrank-sharp-2}, it holds that
\begin{align*}
 \| \widehat{\bm B} \widehat{\bm d}_0  - \bm B^* \bm d^*_0   \|_2 \bigg[
\| \widehat{\bm B} \widehat{\bm d}_0  - \bm B^* \bm d^*_0   \|_2 - \frac{2}{\rho}\bigg(
\tau_0 +  \eta_0 \frac{6 s_{\max}^{1/2}\gamma}{\rho \alpha_2} 
\bigg)
\bigg]
\leq \kappa \bigg(
\frac{6\alpha_1 s_{\max}^{1/2}\gamma}{\rho \alpha_2} 
\bigg)^2.
\end{align*}
Therefore,
\begin{align*}
\| \widehat{\bm B} \widehat{\bm d}_0  - \bm B^* \bm d^*_0   \|_2 &  \lesssim 
\frac{\tau_0}{\rho} + \frac{s_{\max}^{1/2}\gamma}{\rho \alpha_2} \left(\frac{\eta_0}{\rho}+ \alpha_1\kappa^{1/2}\right).
\end{align*}

    From the conditions of Lemma \ref{lemma:deterministic2}, we obtain
    \begin{align}
    2 C  s_{\max} \kappa^{3/2} \mu^2\eta_k^* < \lambda_k.
    \label{eqn-armul-deterministic-lowrank-lambda}
    \end{align}
    Note that $\nabla \mathcal{L}_k(\bbm \alpha_k) = T_k^{-1} \bm Z_k^{\top} (\bm y_k - \bm Z_k \bbm \alpha_k)$ and $\nabla^2 \mathcal{L}_k(\bbm \alpha_k) = T_k^{-1} \bm Z_k^{\top} \bm Z_k \preceq L \bm I_{N^2p} $. 
    For any $ k $, we have
    $$
     \| \nabla \mathcal{L}_k (\bbm \alpha_k^*) - \nabla \mathcal{L}_k (\bm B^* \bm d_k^*) \|_2 =  \| T_k^{-1} \bm Z_k^\top \bm Z_k (\bbm \alpha_k^* - \bm B^* \bm d_k^*) \|_2 
    \leq \| T_k^{-1} \bm Z_k^\top \bm Z_k \|_2  \| \bbm \alpha_k^* - \bm B^* \bm d_k^* \|_2 \leq L h.
    $$
    By the triangle inequality and the condition for $ h $ in Lemma \ref{lemma:deterministic2}, we obtain
    \begin{equation}\label{eq:turnLh}
    \begin{split}
    	 \| \nabla \mathcal{L}_k (\bm B^* \bm d_k^*) \|_2 
    &\leq  \| \nabla \mathcal{L}_k (\bbm \alpha_k^*) \|_2 +  \| \nabla \mathcal{L}_k (\bm B^* \bm d_k^*) - \nabla \mathcal{L}_k (\bbm \alpha_k^*) \|_2 \\
    &\leq \eta_k^* +  L h \\
    &< \frac{(2C + 1)\eta_k^*}{2 C} + \frac{\lambda_k}{4 C^2 s_{\max} \kappa^{3/2} \mu^2 }.
    \end{split}
    \end{equation}
    Consequently,
    \begin{align*}
    &C s_{\max} \kappa^{3/2} \mu^2   \| \nabla \mathcal{L}_k (\bm B^* \bm d_k^*) \|_2
    <(C + 1/2)s_{\max} \kappa^{3/2} \mu^2  \eta_k^* + \frac{\lambda_k}{4 C} 
    \overset{\mathrm{(i)}}{<} \lambda_k,
    \end{align*}
    where the inequality $ \mathrm{(i)} $ follows from \eqref{eqn-armul-deterministic-lowrank-lambda} as well as $ C \geq 1 $. This together with Lemma \ref{lemma:deterministic1} ensures that $ \widehat{\bbm \alpha}_0 = \widehat{\bm B} \widehat{\bm d}_0 $ and
    \begin{align*}
    \| \widehat{\bm B} \widehat{\bm d}_0 - \bm B^* \bm d_0^* \|_2 & \lesssim 
    \frac{\tau_0}{\rho} + \frac{\gamma s_{\max}^{1/2}}{\rho \alpha_2} \left( \frac{\eta_0}{\rho} + \alpha_1 \kappa^{1/2} \right),
    \end{align*}
    where $ \| \bm P^* \nabla \mathcal{L}_0 (\bm B^* \bm d_0^*) \|_2 \le \tau_0  $,
    $$
    \sup_{\bm w \in \bm \Omega(2s_1, 2s_2,2s_3; N, 1,p)} \langle  \mathcal{L}_0 (\bm B^* \bm d_k^*), \bm w \rangle\le \tau
    \quad \text{and} \quad 
     \sup_{\bm W \in \bm \Omega(2s_1, 2s_2,2s_3 ;N, K, p)} \sum_{k\in [K]} \left\langle w_k^{1/2} \nabla \mathcal{L}_k (\bm B^* \bm d_k^*), \bm w_k \right\rangle \le \gamma.
    $$
    By the triangle inequality, we have
    $$
    \| \widehat{\bbm \alpha}_0 - \bbm \alpha_0^* \|_2 \leq \| \widehat{\bbm \alpha}_0 - \bm B^* \bm d_0^* \|_2 + \| \bm B^* \bm d_0^* - \bbm \alpha_0^* \|_2 \leq \| \widehat{\bm B} \widehat{\bm d}_0 - \bm B^* \bm d_0^* \|_2 + h_{\mathrm{T}}.    $$
    Similar to \eqref{eq:turnLh}, we can show that $ \tau_0 \leq \tau_0^* + Lh_{\mathrm{T}} $, $ \eta_0 \leq \eta_0^* + Lh_{\mathrm{T}} $ and $\gamma \leq \gamma^* + L h_{\mathrm{S}}.$
    Based on the above inequalities, we conclude that
    \begin{align*}
    \| \widehat{\bbm \alpha}_0 - \bbm \alpha_0^* \|_2 & \lesssim 
    \frac{\tau_0}{\rho} + \frac{s_{\max}^{1/2} \gamma}{\rho \alpha_2} \left( \frac{\eta_0}{\rho} + \alpha_1 \kappa^{1/2} \right) + h_{\mathrm{T}} \\
    &\le \frac{\tau_0^* + Lh_{\mathrm{T}}}{\rho} + \frac{s_{\max}^{1/2} (\gamma^* + Lh_{\mathrm{S}})}{\rho \alpha_2} \left( \frac{\eta_0^* + Lh}{\rho} + \alpha_1 \kappa^{1/2} \right)+h_{\mathrm{T}} \\
    & \lesssim \frac{\tau_0^*}{\rho} + \kappa h_{\mathrm{T}} + \frac{s_{\max}^{1/2}}{\alpha_2} \left( \frac{\gamma^*}{\rho} + \kappa h_{\mathrm{S}} \right) \left( \frac{\eta_0^*}{\rho} + \kappa h_{\mathrm{T}} + \sqrt{\kappa} \alpha_1 \right),
    \end{align*}
   	which completes the proof.

\end{proof}

\begin{proof}[\bf Proof of Proposition 
    \ref{prop:probeventsTL}]
       By Assumptions \ref{asmp:1}--\ref{asmp:2} and the conditions of Proposition \ref{prop:probeventsTL}, the proof of Proposition \ref{prop:probeventsTL} is done by showing the following Claims \ref{claim:deviationRSSTL}--\ref{claim:deviationTL}.

       \begin{claim}[Event $\mathcal{E}_4$]
\label{claim:deviationRSSTL} 
         If $T_{0}  \gtrsim  Np\max ( \widetilde{\kappa}^{2}\sigma^{4}, \widetilde{\kappa} \sigma^{2} )$, then we have 
         \begin{equation}
           \rho \bm I_{N^2p} \preceq  \nabla^2 \mathcal{L}_0(\bbm \alpha_k^*)   \preceq L \bm I_{N^2p} 
         \quad \text{and} \quad
            \left\| \nabla \mathcal{L}_0(\bbm \alpha_0^*) \right\|_2 \lesssim\sigma M_{k}\sqrt{\frac{N^2p+N\log (N )}{T_0}}
        \end{equation}
        with probability at least
        $
        1 - \exp\left(-C_1 [Np+\log (N) ]\right) - 2\exp(-C_2 T_{0} \min\left( \kappa_0^{-2} \sigma^{-4}, \kappa_0^{-1} \sigma^{-2} \right))
        $.
    \end{claim}
    
     \begin{claim}[Event $\mathcal{E}_5$]\label{claim:deviationTL} 
         If $T_0 \gtrsim s_1s_2s_3\max ( \widetilde{\kappa}^{2}\sigma^{4}, \widetilde{\kappa} \sigma^{2} )$, then with probability at least
        $
        1  - 2\exp\left(-C T_0 \min\left( \widetilde{\kappa}^{-2} \sigma^{-4}, \widetilde{\kappa}^{-1} \sigma^{-2} \right)\right),
        $ we have
         \begin{equation}
            \left\| \bm P^* \nabla \mathcal{L}_0(\bbm \alpha_0^*) \right\|_2 \lesssim\sigma M_{0}\sqrt{\frac{s_1s_2s_3}{T_0}}.
        \end{equation}

    \end{claim}
     \begin{proof}[\bf Proof of Claim \ref{claim:deviationTL}]
        For any fixed $\bm P^* = \bm B^* \bm B^{*\top}$ and $\bm B^{*\top} \bm B^* = \bm I_{s_1s_2s_3},$ by dual property of $\ell_2$-norm, it follows that 
        \begin{align*}
            \left\| \bm P^* \nabla \mathcal{L}_0(\bbm \alpha_0^*) \right\|_2 =\left\| \bm B^{*\top} \nabla \mathcal{L}_0(\bbm \alpha_0^*) \right\|_2 = \sup_{\bm w \in \mathbb{R}^{s_1s_2 s_3}} \langle   \nabla \mathcal{L}_0(\bbm \alpha_0^*), \bm B^{*}\bm w\rangle.
        \end{align*}
        Together with $\bm w \in \bm \Theta(\epsilon, s_1s_2s_3, 1)$, the proof is complete by Lemma \ref{lemma:deviation}.
    \end{proof}
\end{proof}

\subsection{Proofs of Lemmas \ref{lemma:deterministic1}--\ref{lemma:uniondeviation}}\label{sec:proofoflemma}

This section includes Lemmas \ref{lemma:deterministic1}--\ref{lemma:uniondeviation} with proofs. 
Specifically, Lemmas~\ref{lemma:deterministic1}--\ref{lemma:deterministic2} are used to prove Propositions~\ref{prop:determrep} and \ref{prop:determTL},
and Lemmas~\ref{lemma:RSCRSS}--\ref{lemma:uniondeviation} are used to prove Propositions~\ref{prop:probeventsrep} and \ref{prop:probeventsTL}.

\begin{lemma}\label{lemma:deterministic1}
	Suppose that $\{\mathcal{L}_k,k\in [K]\}$ are $ h $-related with regularity parameters $ (\bm B^*,\{\bbm \alpha_k^*,\bm d_k^*,\\ \eta_k,k\in [K]\},  \rho, L) $
	and Assumption \ref{asmp:3} holds.
	Assume there exists positive constant $ C \geq 1 $ such that
	$
	\lambda_k > C s_{\max} \eta_k \kappa^{3/2} \mu^2,
	$
	then we have $ \widehat{\bbm \alpha}_k = \widehat{\bm B} \widehat{\bm d}_k $ for all $k\in [K]$, 
	$$\Big\| \left\{w_k^{1/2} (\widehat{\bm B} \widehat{\bm d}_k  - \bm B^{*} \bm d_k^{*}),k \in [K]\right\} \Big\|_\F \leq 2 \gamma \rho^{-1},$$
	and hence,
	$$\max\{\|\sin \Theta (\widehat{\bm U}, \bm U^* )\|_\F, \|\sin \Theta (\widehat{\bm V}, \bm V^* )\|_\F, \|\sin \Theta (\widehat{\bm L}, \bm L^* )\|_\F\}
	\lesssim  \frac{s_{\max}^{1/2}}{\rho \alpha_2}\gamma.$$
\end{lemma}

\begin{lemma}\label{lemma:deterministic2}
	Suppose that conditions in Proposition \ref{prop:determrep} hold. If for each $k \in [K]$,
	$
	2 C  L h_{\mathrm{S}} \leq \eta_k^* + \frac{\lambda_k}{ 2 C s_{\max} \kappa^{3/2} \mu^2 }
	$ with $ \eta_k^* =\|\nabla \mathcal{L}_k (\bbm \alpha_k^{*}) \|_2 $, 
	then we have $ \widehat{\bbm \alpha}_k = \widehat{\bm B} \widehat{\bm d}_k $ and 
	\begin{align*}
		\| \widehat{\bbm \alpha}_k - \bbm \alpha_k^* \|_2  
		\lesssim
		\frac{s_{\max}^{1/2}}{\alpha_2} \left( \frac{\gamma^*}{\rho} + \kappa h_{\mathrm{S}} \right) \left( \frac{\eta_0^*}{\rho} + \kappa h_{\mathrm{S}} + \kappa^{1/2} \alpha_1 \right)
		+\frac{ \tau_k^* }{\rho} + \kappa h_{\mathrm{S}}.
	\end{align*}
\end{lemma}

\begin{lemma}\label{lemma:RSCRSS}
	Suppose Assumptions \ref{asmp:1}--\ref{asmp:2} hold. 
	Let $\bm{W}_k \in \bm \Theta(\epsilon, d, q)$. If $T_k\gtrsim  d\max ( \widetilde{\kappa}_k^{2}\sigma^{4}, \widetilde{\kappa}_k \sigma^{2} )  $, then for any $k$, with probability at least $1-2\exp(-C  T_k \min (  \widetilde{\kappa}_k^{-2}\sigma^{-4},  \widetilde{\kappa}_k^{-1} \sigma^{-2}))$, it holds that 
	\begin{equation}\label{eq:singleRSCRSS-1}
		\frac{\rho_k}{2}  \le  T_k^{-1}\sum_{t=1}^{T_k}\|\bm W_{k}\bm{x}_{t,k}\|_2^2  \le \frac{L_k}{2} .
	\end{equation}
	Furthermore, if $T_{\min} \gtrsim (d+\log K)\max ( \widetilde{\kappa}^{2}\sigma^{4}, \widetilde{\kappa} \sigma^{2} )$, then for all $k\in [K]$, with probability at least $1-2\exp(-C T_{\min}
	\min (  \widetilde{\kappa}^{-2}\sigma^{-4},  \widetilde{\kappa}^{-1} \sigma^{-2}))$, we have
	\begin{equation}
		\frac{\rho_k}{2}  \le  T_k^{-1}\sum_{t=1}^{T_k}\|\bm W_{k}\bm{x}_{t,k}\|_2^2  \le \frac{L_k}{2} .
	\end{equation}
\end{lemma}

\begin{lemma}\label{lemma:deviation}
	Suppose Assumptions \ref{asmp:1}--\ref{asmp:2} hold. Let $\bm{W}_k \in \bm{\Theta}(\epsilon, d, q)$. If $T_k \gtrsim d \max\left( \widetilde{\kappa}_k^{2} \sigma^{4}, \widetilde{\kappa}_k \sigma^{2} \right)$, then for each $k \in [K]$, with probability at least 
	$
	1 - 2\exp\left(-C_2 T_k \min\left( \widetilde{\kappa}_k^{-2} \sigma^{-4}, \widetilde{\kappa}_k^{-1} \sigma^{-2} \right)\right),
	$
	we have
	\begin{equation}
		\sup_{\bm{W}_k \in \bm{\Theta}(\epsilon, d, q)}\left\langle \frac{1}{T_k} \sum_{t=1}^{T_k} \bbm{\varepsilon}_{t,k} \bm{x}_{t,k}^\top, \bm{W}_k \right\rangle
		\lesssim \sigma M_{k} \sqrt{ \frac{d }{T_k} }.
	\end{equation}
	Furthermore, if $T_{\min} \gtrsim (d + \log K) \max\left( \widetilde{\kappa}^{2} \sigma^{4}, \widetilde{\kappa} \sigma^{2} \right)$, then for all $k \in [K]$, with probability at least
	$
	1 - 2\exp\left(-C_2 T_{\min}\min\left( \widetilde{\kappa}^{-2} \sigma^{-4}, \widetilde{\kappa}^{-1} \sigma^{-2} \right)\right),
	$
	it holds that
	\begin{equation}
		\sup_{\bm{W}_k \in \bm{\Theta}(\epsilon, d, q)}\left\langle \frac{1}{T_k} \sum_{t=1}^{T_k} \bbm{\varepsilon}_{t,k} \bm{x}_{t,k}^\top, \bm{W}_k \right\rangle
		\lesssim \sigma M_{k} \sqrt{ \frac{d  + \log K}{T_k} }.
	\end{equation}
\end{lemma}

\begin{lemma}\label{lemma:uniondeviation}
	Suppose Assumptions \ref{asmp:1}--\ref{asmp:2} hold. 
	For any $\cm{W} = [ \cm{W}_1: \cdots :  \cm{W}_K]$, where each $\cm{W}_k$ has Tucker ranks $(2s_1, 2s_2, 2s_3)$ and satisfies $\sum_{k \in [K]} \| \cm{W}_k \|_{\textup{F}}^2 = 1$, if 
	$
	T_{\min} \gtrsim [ N(s_1 \vee s_2) + \log K]\max\left( \widetilde{\kappa}^{2} \sigma^{4}, \widetilde{\kappa} \sigma^{2} \right),
	$
	then with probability at least
	$
	1 - \exp\left( -C_1 d_M \right) - 2\exp( -C_2 T_{\min} \\\min( \widetilde{\kappa}^{-2} \sigma^{-4}, \widetilde{\kappa}^{-1} \sigma^{-2} )),
	$
	we have
	\begin{equation}
		\sup_{\cm W\in\bm \Theta(\epsilon, d_M, 8)}\sum_{k\in [K]} w_k^{1/2} \left\langle  \nabla \mathcal{L}_k(\cm{A}_k^*),\cm W_k\right\rangle \lesssim \sigma L^{1/2}\sqrt{\frac{d_M}{\widetilde{T}}}.
	\end{equation}
	
\end{lemma}

\begin{proof}[\bf Proof of Lemma \ref{lemma:deterministic1}]\label{sec-lemma:deterministic1-proof}
	For each $k$, when $\lambda_k$ is sufficiently large, $\widetilde{\mathcal{L}}_k = \mathcal{L}_k$ in a neighborhood of $\bm B^{*} \bm d^{*}_k$. The strong convexity of $\mathcal{L}_k$ ensures that $\widetilde{\mathcal{L}}_k$ also inherits this property in a neighborhood of $\bm B^{*} \bm d^{*}_k$. To prove Lemma \ref{lemma:deterministic1}, it suffices to show the following Claims \ref{claim-lowrank-subspace}--\ref{claim-lowrank-subspace-sharp}.

\begin{claim}\label{claim-lowrank-subspace}
When $ \lambda_k \ge \eta_k \kappa ( 3 + 2 s_{\max}^{1/2} \mu)$,
we have $\| (\bm I - \widehat{\bm P}) \bm B^{*} \|_2 \leq  \frac{ \sqrt{2s_{\max} }}{\alpha_2  } \frac{\eta}{\rho} ( 1 + 2 s_{\max}^{1/2} \mu )$, where  $\widehat{\bm P} = \widehat{\bm B} \widehat{\bm B}^\top$ is the projection onto $\mathcal{M}(\widehat{\bm B})$ and $\eta_k = w_k^{-1/2} \eta$.
\end{claim}

\begin{claim}\label{claim-lowrank-z}
    There exists a constant $C >0$ such that
    \begin{equation}\label{eq:claimLRzlam}
	 \eta_k \kappa \left( 3 +  C s_{\max} \kappa^{1/2} \mu^2\right)  \leq  \lambda_k,
    \end{equation}
    then we have
    $$
     \| \widehat{\bm B} \widehat{\bm d}_k - \widehat{\bm P} \bm B^{*} \bm d^{*}_k \|_2 \leq  C s_{\max}   \kappa^{1/2} \mu^2 \frac{\eta_k}{\rho}.
    $$
\end{claim}

\begin{claim}\label{claim-lowrank-crude}
    Suppose that there exist positive constants $C_1$ and $C_2$ such that $\lambda_k \ge C_1 L \xi_k$ holds with $\xi_k =  s_{\max}  \kappa^{1/2}\mu^2 \rho^{-1}\eta_k$, 
    then we have $\| \widehat{\bm B} \widehat{\bm d}_k -  \bm B^{*} \bm d^{*}_k \|_2 \leq C_2 \xi_k$ and $\widehat{\bbm \alpha}_k = \widehat{\bm B} \widehat{\bm d}_k $. In addition, $\widetilde{\mathcal{L}}_k = \mathcal{L}_k$ in $B( \bm B^{*} \bm d^{*}_k , C_2 \xi_k )$.
\end{claim}

\begin{claim}\label{claim-lowrank-subspace-sharp}
    Under Assumption \ref{asmp:3} and the conditions in Claim \ref{claim-lowrank-crude}, we have
    $$
    \Big\| \left\{w_k^{1/2} (\widehat{\bm B} \widehat{\bm d}_k  - \bm B^{*} \bm d_k^{*})  \right\}_{k\in [K]} \Big\|_\F  \leq \frac{2 \gamma}{\rho} \qquad\text{and}\qquad
        \|\widehat{\bm B}\widehat{\bm B}^{\top} - \bm B^*\bm B^{*\top}\|_2 \leq \frac{6s_{\max}^{1/2}\gamma}{\alpha_2 \rho }.
    $$
    \end{claim}
    
    By Claims \ref{claim-lowrank-subspace}--\ref{claim-lowrank-crude}, it holds that $\widehat{\bm B} \widehat{\bm d}_k $ belongs to the neighborhood of $\bm B^{*} \bm d^{*}_k$ such that $\widetilde{\mathcal{L}}_k = \mathcal{L}_k$. This together with the condition of Lemma \ref{lemma:deterministic1}, implies that $\widehat{\bbm \alpha}_k = \widehat{\bm B} \widehat{\bm d}_k$. 
    Then by the strong convexity of $\widetilde{\mathcal{L}}_k$, we can get the sharp bounds on $\| \{w_k^{1/2} (\widehat{\bm B} \widehat{\bm d}_k  - \bm B^{*} \bm d_k^{*})  ,k\in [K]\} \|_\F$ and $\| \widehat{\bm B} \widehat{\bm B}^\top - \bm B^{*} \bm B^{*\top}  \|_2$ as stated in Claim \ref{claim-lowrank-subspace-sharp}.
    The proof of Lemma \ref{lemma:deterministic1} is complete. 
\end{proof}

% \begin{claim}\label{claim-lowrank-dk-sharp}
% Under the conditions in Claim \ref{claim-lowrank-crude}, by Claim \ref{claim-lowrank-subspace-sharp}, for each $k \in [K]$, we have
%         $$
%             \| \widehat{\bm B} \widehat{\bm d}_k  - \bm B^* \bm d^*_k  \|_2   \lesssim
%         \frac{\tau_0}{\rho} + \frac{s_{\max}^{1/2}\gamma}{\rho \alpha_2} \left(\frac{\eta_k}{\rho}+ \alpha_1\kappa^{1/2}\right).
%         $$
% \end{claim}

\begin{proof}[\bf Proof of Claim \ref{claim-lowrank-subspace}]
Define $r_k =2  s_{\max}^{1/2} \mu \rho^{-1}\eta_k  $. Since $ \lambda_k \ge 3 \eta_k \kappa + Lr_k $, combining \eqref{eq:lowrank-mtl-2} and Lemma \ref{lem:repre-lower-bound}, we have
\begin{equation}
\begin{split}
0 &\geq \sum_{k=1}^{K} w_k \widetilde{\mathcal{L}}_k ( \widehat{\bm B} \widehat{\bm d}_{ k}  ) - \sum_{k=1}^{K} w_k \widetilde{\mathcal{L}}_k ( \bm B^{*} \bm d^{*}_{k} )\\ 
&\geq \sum_{k =1}^K  w_k H_k \Big(
(\| \widehat{\bm B} \widehat{\bm d}_{ k} - \bm B^{*} \bm d^{*}_{k} \|_2 -  \eta_k/\rho
)_{+} 
\Big) 
- \frac{ \sum_{k =1}^K  w_k\eta_k^2}{\rho},
\end{split}
\end{equation}
where $\widetilde{\mathcal{L}}_k(\cdot)$ is defined in \eqref{eq:infcov}, $H_k(t) = t^2 / 2$ if $0 \leq t \leq r_k$ and $H_k(t) = r_k(t-r_k/2)$ if $t > r_k$. It can be shown that $
\|  \widehat{\bm B} \widehat{\bm d}_k - \bm B^{*} \bm d^{*}_k \|_2 \geq 
\| (\bm I - \widehat{\bm P}) ( \widehat{\bm B} \widehat{\bm d}_k - \bm B^{*} \bm d^{*}_k ) \|_2 
= \| (\bm I - \widehat{\bm P}) \bm B^{*} \bm d^{*}_k \|_2. $
This together with the monotonicity property of $H_k(\cdot)$, implies that
\begin{align}
\sum_{k =1}^K w_k H_k \Big(
( \| (\bm I - \widehat{\bm P}) \bm B^{*} \bm d^{*}_k \|_2 -  \eta_k/\rho
)_{+} 
\Big) 
\leq  \frac{ \sum_{k =1}^K  w_k \eta_k^2}{\rho^2} = \frac{K\eta^2}{\rho^2}.
\label{eqn-lowrank-lower-1}
\end{align}

Next, we prove Claim \ref{claim-lowrank-subspace} by contradiction. Suppose that 
\begin{align}
\| (\bm I - \widehat{\bm P}) \bm B^{*} \|_2 >  \frac{ \sqrt{2s_{\max} } }{ \alpha_2  } \frac{\eta}{\rho} \bigg( 1 + 2 s_{\max}^{1/2} \mu \bigg).
\label{eqn-lowrank-cond}
\end{align}
% since $\max_{k \in [K]} \|\bm d_k\|_2\le \alpha_1$, $\sum_{k\in [K]}w_k\bm D_k \bm D_k^\top \succeq \alpha_1^2/s_1 \bm I_{s_1}$, $\sum_{k\in [K]}w_k\bm D_k^\top \bm D_k \succeq \alpha_2^2/s_2 \bm I_{s_2}$ for $0< \max\{\alpha_1,\alpha_2\} \le \alpha_0$, Let $\alpha_3^2 = \max\{\alpha_1^2/s_1, \alpha_2^2/s_2\} $ and 
Let 
$
\mathcal{T} = \{ k \in [K]:~ \| (\bm I - \widehat{\bm P}) \bm B^{*} \bm d^{*}_k \|_2 > w_k^{-1/2}\alpha_2 \| (\bm I - \widehat{\bm P}) \bm B^{*} \|_2 / \sqrt{2s_{\max}} \}.
$
It follows that
\begin{align*}
& \| (\bm I - \widehat{\bm P}) \bm B^{*} \bm d^{*}_k \|_2 >\frac{ \eta_k}{ \rho } \bigg( 1 + 2 s_{\max}^{1/2} \mu  \bigg) , \qquad \forall k \in \mathcal{T} , \\
& \sum_{k =1}^K w_k H_k \Big(
( \| (\bm I - \widehat{\bm P}) \bm B^{*} \bm d^{*}_k \|_2 -  \eta_k/\rho
)_{+} 
\Big) 
> \sum_{k \in \mathcal{T}} w_k H_k\bigg(2 s_{\max}^{1/2} \mu \frac{\eta_k}{\rho} \bigg).
\end{align*}
Recall that $r_k =  2 s_{\max}^{1/2} \mu \frac{\eta_k}{\rho}$. Hence,
\begin{align*}
H_k\bigg(2 s_{\max}^{1/2} \mu \frac{\eta_k}{\rho}\bigg)
=  \frac{1}{2} \bigg(2 s_{\max}^{1/2} \mu \frac{\eta_k}{\rho} \bigg)^2.
\end{align*}
Meanwhile, by Lemma \ref{lem-frac}, we have
$$
|\mathcal{T}|/K  \geq \frac{\alpha_2^2 / s_{\max} - (\alpha_2/\sqrt{2 s_{\max} })^2}{\alpha_1^2 - (\alpha_2/\sqrt{2 s_{\max} })^2} = \frac{1}{ 2 (\sqrt{ s_{\max}  } \alpha_1 / \alpha_2 )^2 - 1} \geq  \frac{\alpha_2^2}{ 2 s_{\max}\alpha_1 ^2 }.
$$
As a result,
\begin{align}
 \sum_{k =1}^K H_k \Big(
[ w_k^{1/2}\| (\bm I - \widehat{\bm P}) \bm B^{*} \bm d^{*}_k \|_2 -  \eta/\rho
]_{+} 
\Big) >\frac{K\alpha_2^2}{ 2 s_{\max}\alpha_1 ^2 } \frac{1}{2} \bigg(
2 s_{\max}^{1/2} \mu \frac{\eta}{\rho}  \bigg)^2 = \frac{ K\eta^2}{\rho^2}.
\end{align}
This strict lower bound contradicts \eqref{eqn-lowrank-lower-1}. Therefore, the condition \eqref{eqn-lowrank-cond} does not hold. The proof is completed.
\end{proof}

\begin{proof}[\bf Proof of Claim \ref{claim-lowrank-z}]
Since $\widehat{\bm P}$ is the projection onto $\mathcal{M}(\widehat{\bm B})$, there exists $\bm u_k \in \mathbb{R}^{s_1s_2s_3}$ such that $\widehat{\bm B} \bm u_k = \widehat{\bm P} \bm B^{*} \bm d^{*}_k $. By \eqref{eq:representation-1},
\begin{align}
\widetilde{\mathcal{L}}_k (\widehat{\bm B} \widehat{\bm d}_k ) \leq \widetilde{\mathcal{L}}_k (\widehat{\bm B} \bm u_k ) = \widetilde{\mathcal{L}}_k ( \widehat{\bm P} \bm B^{*} \bm d^{*}_k ) \leq \mathcal{L}_k ( \widehat{\bm P} \bm B^{*} \bm d^{*}_k ) .
\label{eq:representation-10}
\end{align}
%Note that
%$$
%\frac{\eta_k L }{\rho} \left[3  +  C s_{\max}   \sqrt{ \frac{ L}{ \rho} } \bigg( \frac{ \alpha_1 }{\alpha_2} \bigg)^2 \right] \leq  \lambda_k.
%$$
When $C$ in \eqref{eq:claimLRzlam} is large enough, the conditions in Claim \ref{claim-lowrank-subspace} are satisfied. Then by Claim \ref{claim-lowrank-subspace}, $s_{\max} \ge 1 $ and $\mu \ge 1$, we have
\begin{equation}\label{eq:claim-zdistance}
\begin{split}
	&\| \widehat{\bm P} \bm B^{*} \bm d^{*}_k -  \bm B^{*} \bm d^{*}_k \|_2 
\leq \| (\bm I - \widehat{\bm P} ) \bm B^{*} \|_2 \| \bm d^{*}_k \|_2\\
\leq &\sqrt{2s_{\max} } ~\mu \frac{ \eta_k  }{  \rho }   \bigg( 1 + 2 s_{\max}^{1/2}\mu \bigg) 
\leq 5s_{\max} \mu^2  \frac{ \eta_k}{\rho}.
\end{split}
\end{equation}
 When $C \geq 5$, we have $\widehat{\bm P} \bm B^{*} \bm d^{*}_k \in B(\bm B^{*} \bm d^{*}_k, 5s_{\max} \mu^2 \rho^{-1} \eta_k)$. By Lemma \ref{lem-inf-conv}, we have $\widetilde{\mathcal{L}}_k(\widehat{\bm P} \bm B^{*} \bm d^{*}_k) = \mathcal{L}_k(\widehat{\bm P} \bm B^{*} \bm d^{*}_k) $. This together with smoothness condition and \eqref{eq:claim-zdistance}, implies that
\begin{align*}
\mathcal{L}_k (\widehat{\bm P} \bm B^{*} \bm d^{*}_k) - \mathcal{L}_k ( \bm B^{*} \bm d^{*}_k )  
&\leq \langle \nabla \mathcal{L}_k (  \bm B^{*} \bm d^{*}_k ) , (\widehat{\bm P} - \bm I)  \bm B^{*} \bm d^{*}_k \rangle + \frac{L}{2}  \| (\widehat{\bm P} - \bm I)  \bm B^{*} \bm d^{*}_k \|_2^2 \\
& \leq \bigg( \eta_k + \frac{L}{2}  \| ( \widehat{\bm P} - \bm I)  \bm B^{*} \bm d^{*}_k \|_2 \bigg)   \| (\widehat{\bm P} - \bm I)  \bm B^{*} \bm d^{*}_k \|_2 \\
& \leq \bigg( \eta_k + \frac{5}{2}L  s_{\max}   \mu^2\frac{\eta_k}{\rho} \bigg)  5 s_{\max}  \mu^2  \frac{\eta_k}{\rho}\\
&\leq 20s_{\max}^2\kappa \mu^4  \frac{\eta_k^2}{\rho}  .
\end{align*}
This inequality and \eqref{eq:representation-10} lead to
\begin{align}
\widetilde{\mathcal{L}}_k (\widehat{\bm B} \widehat{\bm d}_k) - \mathcal{L}_k ( \bm B^{*} \bm d^{*}_k ) 
\leq \mathcal{L}_k (\widehat{\bm P} \bm B^{*} \bm d^{*}_k) - \mathcal{L}_k ( \bm B^{*} \bm d^{*}_k ) 
\leq 20s_{\max}^2 \kappa \mu^4  \frac{\eta_k^2}{\rho} .
\label{eqn-lowrank-lower-10}
\end{align}

On the other hand, let
$$
r_k = C s_{\max} \kappa^{1/2} \mu^2\frac{  \eta_k}{\rho}  .
$$
Then $3 \kappa \eta_k  + Lr_k \leq  \lambda_k$. Applying Lemma \ref{lem:repre-lower-bound} to $\mathcal{L}_k(\cdot)$ yields
\begin{align}\label{eq:claimslowerboundf_k}
\widetilde{\mathcal{L}}_k (\widehat{\bm B} \widehat{\bm d}_k) - \mathcal{L}_k ( \bm B^{*} \bm d^{*}_k ) \geq \rho \cdot H_k \Big(
( \| \widehat{\bm B} \widehat{\bm d}_k - \bm B^{*} \bm d^{*}_k \|_2 - \eta_k / \rho )_+
\Big) - \frac{\eta_k^2}{\rho},
\end{align}
where $H_k(t) = t^2 / 2$ if $0 \leq t \leq r_k$ and $H_k(t) = r_k(t-r_k/2)$ if $t > r_k$. 

Finally, we complete the proof by contradiction. Suppose that $\| \widehat{\bm B} \widehat{\bm d} - \widehat{\bm P} \bm B^{*} \bm d^{*}_k \|_2 > r_k$. It is easily seen that
\begin{align*}
\| \widehat{\bm B} \widehat{\bm d}_k - \bm B^{*} \bm d^{*}_k \|_2	
\geq \| \widehat{\bm P} ( \widehat{\bm B} \widehat{\bm d}_k - \bm B^{*} \bm d^{*}_k ) \|_2	
= \| \widehat{\bm B} \widehat{\bm d}_k - \widehat{\bm P} \bm B^{*} \bm d^{*}_k \|_2.
\end{align*}
By monotonicity of $H_k(\cdot)$ and \eqref{eq:claimslowerboundf_k}, it holds that
\begin{align*}
\widetilde{\mathcal{L}}_k (\widehat{\bm B} \widehat{\bm d}_k) - \mathcal{L}_k ( \bm B^{*} \bm d^{*}_k ) > \rho H_k(r_k - \eta_k / \rho) - \frac{\eta_k^2}{\rho} = \frac{\rho (r_k - \eta_k / \rho)^2}{2} - \frac{\eta_k^2}{\rho}.
\end{align*}
When $C$ is sufficiently large, this lower bound will contradict \eqref{eqn-lowrank-lower-10}. As a result, we must have $\| \widehat{\bm B} \widehat{\bm d}_k - \widehat{\bm P} \bm B^{*} \bm d^{*}_k \|_2 \leq r_k$. This completes the proof.
\end{proof}

\begin{proof}[\bf Proof of Claim \ref{claim-lowrank-crude}]
Since $\lambda_k \ge C_1 L \xi_k$, when $C_1$ is sufficiently large, the conditions in Claims \ref{claim-lowrank-subspace} and \ref{claim-lowrank-z} are satisfied. Then by Claims \ref{claim-lowrank-subspace} and \ref{claim-lowrank-z}, we have
\begin{align*}
& \| (\bm I - \widehat{\bm P}) \bm B^{*} \|_2 \leq  \frac{ \sqrt{2s_{\max} } }{ \alpha_2 }\frac{\eta_k}{\rho } \bigg( 1 + 2  s_{\max}^{1/2}\mu \bigg) , \\
& \| \widehat{\bm P} ( \widehat{\bm B} \widehat{\bm d}_k -  \bm B^{*} \bm d^{*}_k ) \|_2 =
 \| \widehat{\bm B} \widehat{\bm d}_k - \widehat{\bm P} \bm B^{*} \bm d^{*}_k \|_2 \leq C s_{\max} \kappa^{1/2} \mu^2 \frac{   \eta_k}{\rho} ,
\end{align*}
where $C$ is the constant in Claim \ref{claim-lowrank-z}. Note that
$$
\| (\bm I - \widehat{\bm P}) ( \widehat{\bm B} \widehat{\bm d}_k -  \bm B^{*} \bm d^{*}_k ) \|_2 
= \| (\bm I - \widehat{\bm P})  \bm B^{*} \bm d^{*}_k \|_2 
\leq \| (\bm I - \widehat{\bm P})  \bm B^{*} \|_2   \| \bm d^{*}_k \|_2 \leq \alpha_1 \| (\bm I - \widehat{\bm P}) \bm B^{*} \|_2 .
$$
Then, the claimed bound on $  \| \widehat{\bm B} \widehat{\bm d}_k -  \bm B^{*} \bm d^{*}_k \|_2$ follows from
$$
 \| \widehat{\bm B} \widehat{\bm d}_k -  \bm B^{*} \bm d^{*}_k \|_2^2 = 
\| \widehat{\bm P} ( \widehat{\bm B} \widehat{\bm d}_k -  \bm B^{*} \bm d^{*}_k ) \|_2^2 +
\| (\bm I - \widehat{\bm P}) ( \widehat{\bm B} \widehat{\bm d}_k -  \bm B^{*} \bm d^{*}_k ) \|_2^2 .
$$

The other claims are implied by \eqref{eqn-lowrank-mtl-theta}, Lemma \ref{lem-inf-conv} and Lemma \ref{lem:repre-lower-bound}.
\end{proof}

\begin{proof}[\bf Proof of Claim \ref{claim-lowrank-subspace-sharp}]
By Claim \ref{claim-lowrank-crude} and the strong convexity of $\mathcal{L}_k$ near $\bm B^{*} \bm d^{*}_k$, we have
\begin{align*}
\widetilde{\mathcal{L}}_k(\widehat{\bm B} \widehat{\bm d}_k) - \widetilde{\mathcal{L}}_k (\bm B^{*} \bm d^{*}_k)
& = \mathcal{L}_k (\widehat{\bm B} \widehat{\bm d}_k) - \mathcal{L}_k (\bm B^{*} \bm d^{*}_k)  \\
& \geq 
\langle \nabla \mathcal{L}_k (\bm B^{*} \bm d^{*}_k ) , \widehat{\bm B} \widehat{\bm d}_k - \bm B^{*} \bm d^{*}_k \rangle + \frac{\rho}{2} \| \widehat{\bm B} \widehat{\bm d}_k - \bm B^{*} \bm d^{*}_k \|_2^2, \qquad\forall k \in [K].
\end{align*}
Therefore, it holds that
\begin{align*}
0 &\geq \sum_{k=1}^{K}
w_k\widetilde{L}(\widehat{\bm B} \widehat{\bm d}_k) - \sum_{k=1}^{K}w_k\widetilde{L} (\bm B^{*} \bm d^{*}_k) \\
& \geq 
 \Big \langle\Big\{ w_k^{1/2}\nabla \mathcal{L}_k (\bm B^{*} \bm d^{*}_k) \Big \}_{k\in [K]} , \left\{w_k^{1/2}(\widehat{\bm B} \widehat{\bm d}_k  - \bm B^{*} \bm d_k^{*} ) \right\}_{k\in [K]} \Big \rangle + \frac{\rho}{2}\sum_{k\in [K]}w_k \| \widehat{\bm B} \widehat{\bm d}_k  - \bm B^{*} \bm d_k^{*} \|_\F^2  \\
& \geq - \sup_{\bm W \in \bbm \Omega(2s_1,2s_2,2s_3;N,K,p)}\sum_{k\in [K]} \Big\langle  w_k^{1/2}\nabla \mathcal{L}_k (\bm B^{*} \bm d^{*}_k), \bm w_k \Big \rangle \Big\| \left\{w_k^{1/2} (\widehat{\bm B} \widehat{\bm d}_k  - \bm B^{*} \bm d_k^{*})  \right\}_{k\in [K]} \Big\|_\F\\
&~~~+\frac{\rho}{2}\sum_{k\in [K]}w_k \| \widehat{\bm B} \widehat{\bm d}_k  - \bm B^{*} \bm d_k^{*} \|_\F^2 \\
& =   \frac{\rho}{2}\Big\| \left\{w_k^{1/2} (\widehat{\bm B} \widehat{\bm d}_k  - \bm B^{*} \bm d_k^{*})  \right\}_{k\in [K]} \Big\|_\F \Big( \Big\| \left\{w_k^{1/2} (\widehat{\bm B} \widehat{\bm d}_k  - \bm B^{*} \bm d_k^{*})  \right\}_{k\in [K]} \Big\|_\F - \frac{2\gamma}{\rho}\Big) .
\end{align*}
Thus, we have
\begin{equation}\label{eq:Claimsharpbound}
	\Big\| \left\{w_k^{1/2} (\widehat{\bm B} \widehat{\bm d}_k  - \bm B^{*} \bm d_k^{*})  \right\}_{k\in [K]} \Big\|_\F  \leq \frac{2 \gamma}{\rho}.
\end{equation}

Recall that $\bm B^{*} = \bm U^* \otimes \bm L^* \otimes \bm V^*  $. Then, by Assumption \ref{asmp:3}($ii$), we have 
\begin{align*}
&\lambda_{s_1}\Big(\sum_{k\in [K]} w_k\bm U^{*} \cm D_{k(1)}^{*} (\bm L^* \otimes \bm V^*)^\top(\bm L^* \otimes \bm V^*) \cm D_{k(1)}^{*\top} \bm U^{*\top}\Big)  \\
=& \lambda_{s_1}\Big(\bm U^{*} \sum_{k\in [K]} w_k( \cm D_{k(1)}^{*}  \cm D_{k(1)}^{*\top}  )\bm U^{*\top}\Big) = \lambda_{s_1}\Big(\sum_{k\in [K]} w_k \cm D_{k(1)}^{*}  \cm D_{k(1)}^{*\top}\Big) \ge \frac{ \alpha_2^2 }{s_1}.
\end{align*}
Then by Wedin's theorem \citep{wedin1972perturbation} and \eqref{eq:Claimsharpbound}, it follows that 
\begin{equation}
\begin{split}
    \|\widehat{\bm U}\widehat{\bm U}^{\top} - \bm U^*\bm U^{*\top}\|_2 
    &\le \Big\| \left\{w_k^{1/2} (\widehat{\bm B} \widehat{\bm d}_k  - \bm B^{*} \bm d_k^{*})  \right\}_{k\in [K]} \Big\|_2/\sigma_{s_1}\Big(\sum_{k\in [K]} w_k \cm D_{k(1)}^{*}  \cm D_{k(1)}^{*\top}\Big)\\
    &\le \Big\| \left\{w_k^{1/2} (\widehat{\bm B} \widehat{\bm d}_k  - \bm B^{*} \bm d_k^{*})  \right\}_{k\in [K]} \Big\|_\F/\sigma_{s_1}\Big(\sum_{k\in [K]} w_k \cm D_{k(1)}^{*}  \cm D_{k(1)}^{*\top}\Big)\\
    &\le \frac{2\sqrt{s_1}\gamma}{\rho \alpha_2}.
\end{split}
\end{equation}
Similarly, we have $\|\widehat{\bm V}\widehat{\bm V}^{\top} - \bm V^*\bm V^{*\top}\|_2 \le \frac{2\sqrt{s_2}\gamma}{\rho \alpha_2}$ and $\|\widehat{\bm L}\widehat{\bm L}^{\top} - \bm L^*\bm L^{*\top}\|_2 \le \frac{2\sqrt{s_3}\gamma}{\rho \alpha_2}$. 
Thus, 
\begin{equation}
\begin{split}
    \|\widehat{\bm B}\widehat{\bm B}^{\top} - \bm B^*\bm B^{*\top}\|_2 
    &=\|\widehat{\bm U}\widehat{\bm U}^{\top}\otimes \widehat{\bm V}\widehat{\bm V}^{\top} \otimes \widehat{\bm L}\widehat{\bm L}^{\top} - \bm U^*\bm U^{*\top}\otimes \widehat{\bm V}\widehat{\bm V}^{\top}\otimes \widehat{\bm L}\widehat{\bm L}^{\top} \\
    &\quad + \bm U^*\bm U^{*\top}\otimes \widehat{\bm V}\widehat{\bm V}^{\top}\otimes \widehat{\bm L}\widehat{\bm L}^{\top} -\bm U^*\bm U^{*\top} \otimes \bm V^*\bm V^{*\top}\otimes \widehat{\bm L}\widehat{\bm L}^{\top} \\
    &\quad + \bm U^*\bm U^{*\top}\otimes \bm V^*\bm V^{*\top}\otimes \widehat{\bm L}\widehat{\bm L}^{\top} -\bm U^*\bm U^{*\top} \otimes \bm V^*\bm V^{*\top}\otimes \bm L^*\bm L^{*\top} \|_2\\
    &\le \|\widehat{\bm U}\widehat{\bm U}^{\top} -  \bm U^*\bm U^{*\top}\|_2 \|\widehat{\bm V}\widehat{\bm V}^{\top}\|_2 \|\widehat{\bm L}\widehat{\bm L}^{\top}\|_2 \\
    & \quad + \|\bm U^*\bm U^{*\top}\|_2 \|\widehat{\bm V}\widehat{\bm V}^{\top}- \bm V^*\bm V^{*\top}\|_2 \|\widehat{\bm L}\widehat{\bm L}^{\top}\|_2\\
    & \quad +\|\bm U^*\bm U^{*\top}\|_2  \|\bm V^*\bm V^{*\top}\|_2\|\widehat{\bm L}\widehat{\bm L}^{\top}- \bm L^*\bm L^{*\top}\|_2 \\
    &\le \frac{6s_{\max}^{1/2}\gamma}{\rho \alpha_2}.
\end{split}
\end{equation}
This completes the proof.
\end{proof}

\begin{proof}[\textbf{Proof of Lemma \ref{lemma:deterministic2}}]

    From the conditions of Lemma \ref{lemma:deterministic2}, we obtain
    \begin{align}
    2 C  s_{\max} \kappa^{3/2} \mu^2\eta_k^* < \lambda_k.
    \label{eqn-armul-deterministic-lowrank-lambda2}
    \end{align}
    Note that $\nabla \mathcal{L}_k(\bbm \alpha_k) = T_k^{-1} \bm Z_k^{\top} (\bm y_k - \bm Z_k \bbm \alpha_k)$ and $\nabla^2 \mathcal{L}_k(\bbm \alpha_k) = T_k^{-1} \bm Z_k^{\top} \bm Z_k \preceq L \bm I_{N^2p} $. 
    For any $ k\in [K] $, we have
    $$
     \| \nabla \mathcal{L}_k (\bbm \alpha_k^*) - \nabla \mathcal{L}_k (\bm B^* \bm d_k^*) \|_2 =  \| T_k^{-1} \bm Z_k^\top \bm Z_k (\bbm \alpha_k^* - \bm B^* \bm d_k^*) \|_2 
    \leq \| T_k^{-1} \bm Z_k^\top \bm Z_k \|_2  \| \bbm \alpha_k^* - \bm B^* \bm d_k^* \|_2 \leq L h_{\mathrm{S}}.
    $$
    By the triangle inequality and the condition for $ h $ in Lemma \ref{lemma:deterministic2}, we have
    \begin{equation}
    \begin{split}
    	 \| \nabla \mathcal{L}_k (\bm B^* \bm d_k^*) \|_2 
    &\leq  \| \nabla \mathcal{L}_k (\bbm \alpha_k^*) \|_2 +  \| \nabla \mathcal{L}_k (\bm B^* \bm d_k^*) - \nabla \mathcal{L}_k (\bbm \alpha_k^*) \|_2 \\
    &\leq \eta_k^* +  L h_{\mathrm{S}} \\
    &< \frac{(2C + 1)\eta_k^*}{2 C} + \frac{\lambda_k}{4 C^2 s_{\max} \kappa^{3/2} \mu^2 }.
    \end{split}
    \end{equation}
    Consequently,
    \begin{align*}
    &C s_{\max} \kappa^{3/2} \mu^2   \| \nabla \mathcal{L}_k (\bm B^* \bm d_k^*) \|_2
    <(C + 1/2)s_{\max} \kappa^{3/2} \mu^2  \eta_k^* + \frac{\lambda_k}{4 C} 
    \overset{\mathrm{(i)}}{<} \lambda_k,
    \end{align*}
    where the inequality $ \mathrm{(i)} $ follows from \eqref{eqn-armul-deterministic-lowrank-lambda2} and $ C \geq 1 $. This together with Lemma \ref{lemma:deterministic1} ensures that $ \widehat{\bbm \alpha}_k = \widehat{\bm B} \widehat{\bm d}_k $, and
    \begin{align*}
    \| \widehat{\bm B} \widehat{\bm d}_k - \bm B^* \bm d_k^* \|_2 & \lesssim 
    \frac{\gamma s_{\max}^{1/2}}{\rho \alpha_2} \left( \frac{\eta_k}{\rho} + \alpha_1 \kappa^{1/2} \right) 
    +\frac{\tau_k}{\rho} ,
    \end{align*}
    where $ \| \bm P^* \nabla \mathcal{L}_k (\bm B^* \bm d_k^*) \|_2 \le \tau_k  $,
    $$
    \| \nabla \mathcal{L}_k (\bm B^* \bm d_k^*)\|_2 \le \eta_k
    \quad \text{and} \quad 
     \sup_{\bm W \in \bm \Omega(2s_1, 2s_2,2s_3 ;N, K, p)} \sum_{k\in [K]} \left\langle w_k^{1/2} \nabla \mathcal{L}_k (\bm B^* \bm d_k^*), \bm w_k \right\rangle \le \gamma.
    $$
    By the triangle inequality,
    $$
    \| \widehat{\bbm \alpha}_k - \bbm \alpha_k^* \|_2 \leq \| \widehat{\bbm \alpha}_k - \bm B^* \bm d_k^* \|_2 + \| \bm B^* \bm d_k^* - \bbm \alpha_k^* \|_2 \leq \| \widehat{\bm B} \widehat{\bm d}_k - \bm B^* \bm d_k^* \|_2 + h_{\mathrm{S}}.    $$
    Similar to \eqref{eq:turnLh}, we can show that $ \tau_k \leq \tau_k^* + Lh_{\mathrm{S}} $, $ \eta_k \leq \eta_k^* + Lh_{\mathrm{S}} $, and
    \begin{equation}
    \begin{split}
	\gamma \leq& \sup_{\bm W \in \bm \Omega(2s_1, 2s_2,2s_3; N, K, p)} \sum_{k\in [K]} \left\langle w_k^{1/2} \nabla \mathcal{L}_k (\bm A_k^*), \bm w_k \right\rangle \\
    &+ \sup_{\bm W \in \bm \Omega(2s_1, 2s_2, 2s_3; N, K, p)} \sum_{k\in [K]} \left\langle w_k^{1/2} (\nabla \mathcal{L}_k (\bm B^* \bm d_k^*) - \nabla \mathcal{L}_k (\bm A_k^*)), \bm w_k \right\rangle \\
    \leq& \gamma^* + \left\| \{ w_k^{1/2} [ \nabla \mathcal{L}_k (\bm B^* \bm d_k^*) - \nabla \mathcal{L}_k (\bm A_k^*) ] \}_{k \in [K]} \right\|_\F\\
    \leq&  \gamma^* + L h_{\mathrm{S}}.
    \end{split}
    \end{equation}
    
    Based on the above inequalities, we conclude that
    \begin{align*}
    \| \widehat{\bbm \alpha}_k - \bbm \alpha_k^* \|_2 & \lesssim 
    \frac{s_{\max}^{1/2} \gamma}{\rho \alpha_2} \left( \frac{\eta_k}{\rho} 
    + \alpha_1 \kappa^{1/2} \right) + \frac{\tau_k}{\rho} + h_{\mathrm{S}} \\
    &\le \frac{s_{\max}^{1/2} (\gamma^* + Lh_{\mathrm{S}})}{\rho \alpha_2} \left( \frac{\eta_k^* + Lh_{\mathrm{S}}}{\rho} 
    + \alpha_1 \kappa^{1/2} \right) + \frac{\tau_k^* + Lh_{\mathrm{S}}}{\rho}  +h_{\mathrm{S}} \\
    & \lesssim  \frac{s_{\max}^{1/2}}{\alpha_2} \left( \frac{\gamma^*}{\rho} + \kappa h_{\mathrm{S}} \right) \left( \frac{\eta_k^*}{\rho} + \kappa h_{\mathrm{S}} + \sqrt{\kappa} \alpha_1 \right) + \frac{\tau_k^*}{\rho} + \kappa h_{\mathrm{S}},
    \end{align*}
   	which completes the proof.
\end{proof}

\begin{proof}[\bf Proof of Lemma \ref{lemma:RSCRSS}]
        Let $R_{k}(\bm W_k) =  T_k^{-1}\sum_{t=1}^{T_k}\|\bm W_{k}\bm{x}_{t,k}\|_2^2 $. It suffices to show that 
        $$\mathbb{E}[R_{k}(\bm W_k)] - \sup_{\bm W_k\in \bm \Theta}|R_k(\bm W_k)-\mathbb{E}[R_k(\bm W_k)] \le R_{k}(\bm W_k) \le \mathbb{E}[R_{k}(\bm W_k)] + \sup_{\bm W_k\in \bm \Theta}|R_k(\bm W_k)-\mathbb{E}[R_k(\bm W_k)].$$
        First, by the spectral measure of ARMA process in \cite{basu2015regularized}, we have $\lambda_{\min}\left\{\mathbb{E}(\bm x_{t,k} \bm x_{t,k}^\top)\right\} \ge \lambda_{\min}(\bm{\Sigma}_{\bbm{\varepsilon},k})/ \mu_{\max}(\bm \Xi_k) = \rho_k$ and $\lambda_{\max}\left\{\mathbb{E}(\bm x_{t,k} \bm x_{t,k}^\top)\right\} \le \lambda_{\max}(\bm{\Sigma}_{\bbm{\varepsilon},k})/ \mu_{\min}(\bm \Xi_k) = L_k$. Since $\|\bm W_k\|_\F = 1$, then for $\mathbb{E}[R_{k}(\bm W_k)]$, it holds that 
        \begin{equation}\label{eq:EexpHess}
        	\rho_k\le \mathbb{E}[R_{k}(\bm W_k)] =\mathbb{E}\left(T_k^{-1}\sum_{t=1}^{T_k}\tr(\bm W_k \bm x_{t,k} \bm x_{t,k}^\top \bm W_k^\top) \right) \le L_k.
        \end{equation}
        
        Next, we will derive an upper bound for $\sup_{\bm W_k\in \bm \Theta}|R_k(\bm W_k)-\mathbb{E}[R_k(\bm W_k)]$.
        Recall the equivalent VMA($\infty$) form in \eqref{eq:VMAinf}, it follows that 
        \begin{equation}\label{eq:R_k}
        	 R_{k}(\bm W_k) =T_k^{-1}\widetilde{\bm x}_k^\top (\bm I_{T_k}\otimes \bm W_k^\top \bm W_k) \widetilde{\bm x}_k 
             =T_k^{-1} \widetilde{\bbm \zeta}_{k}^\top \widetilde{\bm \Sigma}_{\bbm\varepsilon,k}^{1/2} \widetilde{\bm \Phi}_k^\top(\bm I_{T_k}\otimes \bm W_k^\top \bm W_k)\widetilde{\bm \Phi}_k \widetilde{\bm \Sigma}_{\bbm\varepsilon,k}^{1/2} \widetilde{\bbm \zeta}_{k} 
        	=\widetilde{\bbm \zeta}_{k}^\top \bm H_{W,k} \widetilde{\bbm \zeta}_{k},
        \end{equation}
        where $\bm H_{W,k} = T_k^{-1}\widetilde{\bm \Sigma}_{\bbm\varepsilon,k}^{1/2} \widetilde{\bm\Phi}_k^\top(\bm I_{T_k}\otimes \bm W_k^\top \bm W_k)\widetilde{\bm \Phi}_k \widetilde{\bm \Sigma}_{\bbm\varepsilon,k}^{1/2} $.
        By the sub-multiplicative property of the Frobenius norm and operator norm, we have
        \begin{equation}
            \|\bm H_{W,k}\|_\F^2\leq T_k^{-1} \lambda_{\max}^2(\bm{\Sigma}_{\bbm{\varepsilon},k}) \lambda_{\max}^2(\widetilde{\bm{\Phi}}_k\widetilde{\bm{\Phi}}_k^\top) = T_k^{-1} \widetilde{L}_k^2
        \end{equation}
        and
        \begin{equation}
            \|\bm H_{W,k}\|_\op\leq T_k^{-1}\lambda_{\max}(\bm{\Sigma}_{\bbm{\varepsilon},k}) \lambda_{\max}(\widetilde{\bm{\Phi}}_k\widetilde{\bm{\Phi}}_k^\top) = T_k^{-1} \widetilde{L}_k,
        \end{equation}
        % where $\overline{\lambda}(\bm{\Sigma}_{\bbm{\varepsilon}}) = \max_{k\in\{0,\cdots,K\}} \{\lambda_{\max}(\bm{\Sigma}_{\bbm{\varepsilon},k})\}$ and $\overline{\lambda}(\widetilde{\bm{\Phi}}\widetilde{\bm{\Phi}}^\top)= \max_{k\in\{0,\cdots,K\}}\{\lambda_{\max}(\widetilde{\bm{\Phi}}_k\widetilde{\bm{\Phi}}_k^\top)\}$.
        where the equalities hold by the properties of VMA($\infty$) process. Note that $\widetilde{\bm{\Phi}}_k$ is related to the VMA($\infty$) process, by the spectral measure of ARMA process discussed in \citet{basu2015regularized}, we can replace $\lambda_{\max}(\widetilde{\bm{\Phi}}_k\widetilde{\bm{\Phi}}_k^{\top})$ with $1/\mu_{\min}(\widetilde{\bm \Xi}_k)$, hence
        $\lambda_{\max}(\bm{\Sigma}_{\bbm{\varepsilon},k}) \lambda_{\max}(\widetilde{\bm{\Phi}}_k\widetilde{\bm{\Phi}}_k^\top) = 
        \lambda_{\max}(\bm{\Sigma}_{\bbm{\varepsilon},k})/\mu_{\min}(\widetilde{\bm \Xi}_k) = \widetilde{L}_k$.
        For any $\bm W_k\in \bm \Theta(\epsilon, d,q)$ and $t>0$, by Hanson-Wright inequality,
        \begin{equation}\label{eq:boundR-ER1}
            \begin{split}
                 &\mathbb{P}[|R_k(\bm W_k)-\mathbb{E}[R_k(\bm W_k)]|\geq t]
                \leq 2\exp\Bigg(-T_k\min\Bigg(\frac{t^2}{\sigma^4 \widetilde{L}_k^2},
                \frac{t}{ \sigma^2  \widetilde{L}_k }\Bigg)\Bigg).
            \end{split}
        \end{equation}
        Denote $\bm Q_k =T_k^{-1} \sum_{t=1}^{T_k}[\bm{x}_{t,k}\bm{x}_{t,k}^\top - \mathbb{E}(\bm{x}_{t,k}\bm{x}_{t,k}^\top)]$ for $k \in [K]$. Then we can rewrite $R_k(\bm W_k)-\mathbb{E}[R_k(\bm W_k)]$ as follows,
        \begin{equation}
            \begin{split}
                R_k(\bm W_k)-\mathbb{E}[R_k(\bm W_k)] 
                &= T_k^{-1}\sum_{t=1}^{T_k}\tr\Big( \bm W_k\left[\bm{x}_{t,k}\bm{x}_{t,k}^\top - \mathbb{E}(\bm{x}_{t,k}\bm{x}_{t,k}^\top)\right] \bm W_k^\top \Big)\\
                &= T_k^{-1}\tr\Big( \bm W_k\sum_{t=1}^{T_k}\left[\bm{x}_{t,k}\bm{x}_{t,k}^\top - \mathbb{E}(\bm{x}_{t,k}\bm{x}_{t,k}^\top)\right] \bm W_k^\top \Big)\\
                &= \langle \bm W_k \bm{Q}_k, \bm W_k \rangle.
            \end{split}
        \end{equation}
%        By assumption that $\overline{\bm \Theta}$ is an $\epsilon$-covering net of $\mathcal{M}$ with $|\overline{\bm \Theta}| \leq \exp[d\log(C/\epsilon)]$.
        Given any $\bm W_k \in \bm \Theta(\epsilon, d, q)$, by Definition \ref{definition:matrix_set}, there exists $\overline{\bm W}_k \in \overline{\bm \Theta}$ such that
        \begin{itemize}
			\item $\|\bm{W}_k - \overline{\bm{W}}_k\|_\F \leq \epsilon$, and the covering number of $\overline{\bm \Theta}$ is at most $(C/\epsilon)^d$;
			\item $\bm{W}_k - \overline{\bm{W}}_k = \sum_{i=1}^q \bm{N}_{k,i},$
			where $\bm{N}_{k,i}/\|\bm{N}_{k,i}\|_\F \in \bm \Theta(\epsilon, d, q)$ and
			$\|\bm{W}_k - \overline{\bm{W}}_k\|_\F^2 = \sum_{i=1}^q \|\bm{N}_{k,i}\|_\F^2.$
		\end{itemize}
        Thus, we have $\sum_{i=1}^q\| \bm N_{k.i}\|_\F = \sqrt{q}\|\bm W_k-\overline{\bm W}_k\|_\F \le \sqrt{q} \epsilon$. It follows that 
        \begin{equation}
            \begin{split}
                \langle \bm W_k \bm{Q}_k, \bm W_k \rangle
                =& \langle \overline{\bm W}_k \bm{Q}_k, \overline{\bm W}_k \rangle + 2\langle (\bm W_k-\overline{\bm W}_k) \bm{Q}_k, \bm W_k \rangle
                +\langle (\bm W_k-\overline{\bm W}_k) \bm{Q}_k, (\bm W_k -\overline{\bm W}_k) \rangle\\
                \le & 2\langle \overline{\bm W}_k \bm{Q}_k, \overline{\bm W}_k \rangle %+ \sum_{i=1}^4 \|\bm{N}_i\|_\F \sup_{\bm W\in\bm \Theta(2s_1,2s_2,N;1)} \langle \bm W_k \bm{Q}_k, \bm W_k \rangle \\
                +2\sum_{i=1}^q \sum_{j=1}^q  \|\bm N_{k,i}\|_\F \|\bm N_{k,j}\|_\F\Big\langle \frac{\bm N_{k,i}}{\|\bm N_{k,i}\|_\F} \bm{Q}_k, \frac{\bm N_{k,j}}{\|\bm N_{k,j}\|_\F} \Big\rangle\\
                \le&2\langle \overline{\bm W}_k \bm{Q}_k, \overline{\bm W}_k \rangle
                +2(\sum_{i=1}^q  \|\bm N_{k,i}\|_\F)^2 \sup_{\bm W_k\in \bm \Theta(\epsilon,d,q)} \langle \bm W_k \bm{Q}_k, \bm W_k \rangle\\
                \le & 2\langle \overline{\bm W}_k \bm{Q}_k, \overline{\bm W}_k \rangle
                +2q\epsilon^2 \sup_{\bm W_k\in \bm \Theta(\epsilon,d,q)} \langle \bm W_k \bm{Q}_k, \bm W_k \rangle .
            \end{split}
        \end{equation}
        Then we have
        \begin{equation}\label{eq:supR-ER1}
            \sup_{\bm W_k\in\bm \Theta(\epsilon,d,q)} \langle \bm W_k \bm{Q}_k, \bm W_k \rangle \le [1-2q \epsilon^2]^{-1}2 \max_{\overline{\bm W}_k\in \overline{\bm \Theta}} \langle \overline{\bm W}_k \bm{Q}_k, \overline{\bm W}_k \rangle.
        \end{equation}        
        Therefore, by \eqref{eq:boundR-ER1}, \eqref{eq:supR-ER1} and $|\bm \Theta|\le (C/\epsilon)^d$ with $\epsilon=(2\sqrt{q})^{-1}$, we can easily construct the union bound,
        \begin{equation}
            \begin{split}
                & \mathbb{P}\left(\sup_{\bm W_k\in \bm \Theta}|R_k(\bm W_k)-\mathbb{E}[R_k(\bm W_k)] |\geq t\right)\\
                \le &\mathbb{P}\left(\max_{\overline{\bm W}_k \in  \overline{\bm \Theta}}|R_k(\bm W_k)-\mathbb{E}[R_k(\bm W_k)] | \geq (1-2q \epsilon^2)2^{-1} t\right)\\
                \leq & 2\exp\Bigg\{C_1d -C_2T_k\min\Bigg(\frac{t^2}{\sigma^4 \widetilde{L}_k^2},
                \frac{t}{ \sigma^2  \widetilde{L}_k }\Bigg)\Bigg\}.
            \end{split}
        \end{equation}
        Note that $\widetilde{\kappa}_k = \widetilde{L}_k/\rho_k$. Let $t=\rho_k/2$, then we have
        \begin{equation}
            \begin{split}
                \min\Bigg(\frac{t^2}{\sigma^4 \widetilde{L}_k^2},
                \frac{t}{ \sigma^2  \widetilde{L}_k }\Bigg)
                \asymp \min (  \widetilde{\kappa}_k^{-2}\sigma^{-4},  \widetilde{\kappa}_k^{-1} \sigma^{-2} ).
            \end{split}
        \end{equation}
        For $T_k\gtrsim  d\max ( \widetilde{\kappa}_k^{2}\sigma^{4}, \widetilde{\kappa}_k \sigma^{2} )  $, we can show that
        \begin{equation}\label{eq:deviation1}
            \begin{split}
                 &\mathbb{P}\left(\sup_{\bm W_k\in\bm \Theta}|R_k(\bm W_k)-\mathbb{E}[R_k(\bm W_k)] |\geq  \rho_k/2\right)
                \le 2\exp\left(-CT_k \min (  \widetilde{\kappa}_k^{-2}\sigma^{-4},  \widetilde{\kappa}_k^{-1} \sigma^{-2} )  \right).
            \end{split}
        \end{equation}
        This together with \eqref{eq:EexpHess} and \eqref{eq:deviation1}, implies that, for any $k$, 
        \begin{equation}
           \frac{1}{2} \rho_k  \le \frac{1}{2}\lambda_{\min}(\bm{\Sigma}_{\bbm{\varepsilon},k})/ \mu_{\max}(\bm \Xi_k)  \le R_k( \bm W_k)\le \frac{3}{2}\lambda_{\max}(\bm{\Sigma}_{\bbm{\varepsilon},k})/ \mu_{\min}(\bm \Xi_k)   = \frac{1}{2}L_k
        \end{equation}
        with probability at least $1-2\exp(-C  T_k \min (  \widetilde{\kappa}_k^{-2}\sigma^{-4},  \widetilde{\kappa}_k^{-1} \sigma^{-2}))$.

        For all $k\in [K]$, by union bound, it can be easily shown that
        \begin{equation}
            \begin{split}
                 &\mathbb{P}\left(\bigcup_{k\in[K]}\left\{\sup_{\bm W_k\in\bm \Theta}|R_k(\bm W_k)-\mathbb{E}[R_k(\bm W_k)]| \ge t\right\} \right)\\
                 \le &\sum_{k\in [K]}\mathbb{P}\left(\sup_{\bm W_k\in\bm \Theta}|R_k(\bm W_k)-\mathbb{E}[R_k(\bm W_k)]|\geq t   \right)\\
                \le & 2\sum_{k\in [K]}\exp\Bigg\{C_1d -C_2T_k\min\Bigg(\frac{t^2}{\sigma^4 \widetilde{L}_k^2},
                \frac{t}{ \sigma^2  \widetilde{L}_k }\Bigg)\Bigg\}\\
                \le & 2 K \exp\Bigg\{C_1d -C_2 T_{\min} \min\Bigg(\frac{t^2}{\sigma^4 \widetilde{L}^2},
                \frac{t}{ \sigma^2  \widetilde{L} }\Bigg)\Bigg\}.
            \end{split}
        \end{equation}
        We take $t = \rho/2$. For $T_{\min} \gtrsim (d+\log K)\max ( \widetilde{\kappa}^{2}\sigma^{4}, \widetilde{\kappa} \sigma^{2} ) $, it follows that 
        \begin{equation}
            \begin{split}
                 &\mathbb{P}\left(\bigcup_{k\in[K]}\left\{\sup_{\bm W_k\in\bm \Theta}|R_k(\bm W_k)-\mathbb{E}[R_k(\bm W_k)]| \ge \rho/2\right\}    \right)
                \le  2\exp\left(-C T_{\min}\min ( \widetilde{\kappa}^{-2}\sigma^{-4}, \widetilde{\kappa}^{-1} \sigma^{-2} )\right).
            \end{split}
        \end{equation}
        Thus, combining the above inequality with \eqref{eq:EexpHess}, for all $k\in [K]$, with probability at least $1-2\exp\left(-C T_{\min}\min (  \widetilde{\kappa}^{-2}\sigma^{-4},  \widetilde{\kappa}^{-1} \sigma^{-2})\right)$, we have 
        \begin{align}
            \frac{1}{2} \rho  \leq &\mathbb{E}[R_{k}(\bm W_k)] - \sup_{\bm W_k\in \bm \Theta}|R_k(\bm W_k)-\mathbb{E}[R_k(\bm W_k)] \\ \le & R_{k}(\bm W_k) \le \mathbb{E}[R_{k}(\bm W_k)] + \sup_{\bm W_k\in \bm \Theta}|R_k(\bm W_k)-\mathbb{E}[R_k(\bm W_k)]\leq\frac{1}{2} L .
        \end{align}   
    This completes the proof.
        
\end{proof}

\begin{proof}[\bf Proof of Lemma \ref{lemma:deviation}]

        Let $\bm \Theta(\epsilon,d,q)$ satisfy the properties in Definition \ref{definition:matrix_set}. 
        For fixed $\bm W_k \in \bm \Theta(\epsilon,d,q)$ for each $k$, then there exists $\overline{\bm W}_k \in \overline{\bm \Theta}$ such that 
        \begin{itemize}
        	\item $\|\bm{W}_k - \overline{\bm{W}}_k\|_\F \leq \epsilon$, and the covering number of $\overline{\bm \Theta}$ is at most $(C/\epsilon)^d$;
			\item $\bm{W}_k - \overline{\bm{W}}_k = \sum_{i=1}^q \bm{N}_{k,i},$
			where $\bm{N}_{k,i}/\|\bm{N}_{k,i}\|_\F \in \bm \Theta(\epsilon, d, q)$ and
			$\|\bm{W}_k - \overline{\bm{W}}_k\|_\F^2 = \sum_{i=1}^q \|\bm{N}_{k,i}\|_\F^2.$       
		\end{itemize}
		Thus, $\sum_{i=1}^q \|\bm N_{k.i}\|_\F \le  \sqrt{q}\|\bm T_k\|_\F = \sqrt{q} \epsilon$. Combining these properties, we have 
        \begin{equation}
            \begin{split}
                \left\langle  \nabla \mathcal{L}_k(\bbm \alpha_k^*),\bm W_k\right\rangle 
                \leq& \max_{\overline{\bm W}_k\in\overline{\bm \Theta}} \left\langle  \nabla \mathcal{L}_k(\bbm \alpha_k^*),\overline{\bm W}_k\right\rangle + \sum_{i=1}^q \left\langle  \nabla \mathcal{L}_k(\bbm \alpha_k^*),\bm{N}_{k,i}/\|\bm{N}_{k,i}\|_\F\right\rangle\|\bm{N}_{k,i}\|_\F\\
                \leq& \max_{\overline{\bm W}_k\in\overline{\bm \Theta}}
                 \left\langle \nabla \mathcal{L}_k(\bbm \alpha_k^*),\overline{\bm W}_k\right\rangle + \sqrt{q}\epsilon\sup_{\bm W_k\in \bm \Theta(\epsilon,d,q)} \left\langle  \nabla \mathcal{L}_k(\bbm \alpha_k^*),\bm W_k\right\rangle,
             \end{split}
        \end{equation}
        which implies that
        \begin{equation}
            \sup_{\bm W_k\in\bm \Theta(\epsilon,d,q)} \left\langle  \nabla \mathcal{L}_k(\bbm \alpha_k^*),\bm W_k\right\rangle \leq (1-\sqrt{q}\epsilon)^{-1}\max_{\overline{\bm W}_k\in\overline{\bm \Theta}} \left\langle  \nabla \mathcal{L}_k(\bbm \alpha_k^*),\overline{\bm W}_k\right\rangle.
        \end{equation}
    
        For any fixed $\bm W_k$ such that $  \|\bm W_k\|_\F^2=1$, denote $S_k(\bm W_k)= T_k^{-1}\sum_{t=1}^{T_k}\langle\bbm{\varepsilon}_{t,k},\bm W_k\bm{x}_{t,k}\rangle$ and $R_k(\bm W_k)= T_k^{-1}\sum_{t=1}^{T_k}\|\bm W_k\bm{x}_{t,k}\|_2^2$. Similar to Lemma S5 of \citet{wang2024high}, by the standard Chernoff bound, for any $z_1>0$ and $z_2>0$,
        \begin{equation}
        \begin{split}
            &\mathbb{P}[\{S_k(\bm W_k)\geq z_1\}\cap\{R_k(\bm W_k)\leq z_2\}]\\
             = & \inf_{m>0} \mathbb{P}[\{\exp(m S_k(\bm W_k))\geq \exp(m  z_1)\}\cap\{R_k(\bm W_k)\leq z_2\}]\\
            \leq& \inf_{m>0} \exp(-m  z_1) \mathbb{E}\left(\exp[m S_k(\bm W_k)]\mathbb{I}\{R_k(\bm W_k)\leq z_2\}\right)\\
            \leq& \inf_{m>0} \exp(-m z_1 + c m^2 z_2)
            \cdot \mathbb{E}\left(\exp[m S_k(\bm W_k) - cm^2 z_2]\mathbb{I}\{R_k(\bm W_k)\leq z_2\}\right)\\
            \leq& \inf_{m>0} \exp(-m z_1 +  c m^2z_2 )
            \cdot \mathbb{E}\left(\exp[m S_k(\bm W_k) - cm^2 R_k(\bm W_k)]\right).
        \end{split}
        \end{equation}
        By the tower rule, we have
        \begin{equation}\label{eq:deviationMartingale}
            \begin{split}
                &\mathbb{E}\left(\exp[m S_k(\bm W_k) - cm^2 R_k(\bm W_k)]\right) \\
                =&  \mathbb{E}\left(\exp\left[m T_k^{-1} \sum_{t=1}^{T_k}\langle\bbm{\varepsilon}_{t,k}, \bm W_k\bm{x}_{t,k}\rangle - m^2 c T_k^{-1}\sum_{t=1}^{T_k}\|\bm W_k \bm x_{t,k}\|_2^2\right]\right)\\
                =& \mathbb{E}\left\{\mathbb{E}\left(\exp\left[m T_k^{-1} \sum_{t=1}^{T_k}\langle\bbm{\varepsilon}_{t,k},\bm W_k\bm{x}_{t,k}\rangle - m^2 c T_k^{-1}\sum_{t=1}^{T_k}\|\bm W_k \bm x_{t,k}\|_2^2\right] \right)\mid \mathcal{F}_{T_k-1,k}\right\} \\
                =&  \mathbb{E}\Bigg\{\exp\left[m T_k^{-1} \sum_{t=1}^{T_k-1}\langle\bbm{\varepsilon}_{t,k},\bm W_k\bm{x}_{t,k}\rangle - m^2 c T_k^{-1} \sum_{t=1}^{T_k-1}\|\bm W_k \bm x_{t,k}\|_2^2\right]\\
                &\quad \quad  \cdot  \mathbb{E}\Big(\exp\Big[m T_k^{-1} \langle\bbm{\varepsilon}_{T_k,k},\bm W_k\bm{x}_{T_k,k}\rangle - m^2 c T_k^{-1} \|\bm W_k \bm x_{T_k,k}\|_2^2 \Big]\mid \mathcal{F}_{T_k-1,k} \Big) \Bigg\}.
            \end{split}
        \end{equation}
        Recall that $T_k^{-1}\langle\bbm{\varepsilon}_{T_k,k},\bm W_k\bm{x}_{T_k-1,k}\rangle = T_k^{-1}\langle\bbm{\zeta}_{T_k,k},\bm \Sigma_{\bbm{\varepsilon},k}^{1/2}\bm W_k\bm{x}_{T_k-1,k}\rangle$. Given $\mathcal{F}_{T_k-1,k}$, by Assumption \ref{asmp:1}, it can be easily verified that $T_k^{-1}\langle\bbm{\varepsilon}_{T_k,k},\bm W_k\bm{x}_{T_k-1,k}\rangle$ is a $T_k^{-2}\sigma^2 \lambda_{\max}(\bm \Sigma_{\bbm{\varepsilon},k})\|\bm W_k \bm x_{T_k-1,k}\|_2^2$-sub-Gaussian random variable, i.e., 
        \begin{equation}
            \mathbb{E}\Big(\exp\Big[m T_k^{-1} \langle\bbm{\varepsilon}_{T_k,k},\bm W_k\bm{x}_{T_k,k}\rangle\Big] \mid \mathcal{F}_{T_k-1,k} \Big) \le \exp(m^2 T_k^{-2}\sigma^2\lambda_{\max}(\bm \Sigma_{\bbm{\varepsilon},k})\|\bm W_k \bm x_{T_k,k}\|_2^2/2 ).
        \end{equation}
        Letting $c = T_k^{-1}\sigma^2 \lambda_{\max}(\bm \Sigma_{\bbm{\varepsilon},k})/2$, we have $\mathbb{E}\Big(\exp\Big[m  T_k^{-1} \langle\bbm{\varepsilon}_{T_k,k},\bm W_k\bm{x}_{T_k,k}\rangle - m^2 c  T_k^{-1} \|\bm W_k \bm x_{T_k,k}\|_2^2 \Big]\mid \mathcal{F}_{T_k-1,k} \Big) = 0$. Then it follows that 
        \begin{equation}
            \begin{split}
                &\mathbb{E}\left(\exp[m T_k^{-1} S_k(\bm W_k) - m^2 T_k^{-2}\sigma^2 \lambda_{\max}(\bm \Sigma_{\bbm{\varepsilon},k}) R_k(\bm W_k)/2]\right) \\
                \le &  \mathbb{E}\left(\exp\Bigg[m T_k^{-1}  \sum_{t=1}^{T_k-1}\langle\bbm{\varepsilon}_{t,k},\bm W_k\bm{x}_{t,k}\rangle - m^2 T_k^{-2}\sigma^2 \lambda_{\max}(\bm \Sigma_{\bbm{\varepsilon},k})  \sum_{t=1}^{T_k-1}\|\bm W_k \bm x_{t,k}\|_2^2/2\Bigg]\right)\\
                \le& \cdots\le 
                 \mathbb{E}\left(\exp[m T_k^{-1}  \langle\bbm{\varepsilon}_{1,k},\bm W_k\bm{x}_{1,k}\rangle - m^2 T_k^{-2}\sigma^2 \lambda_{\max}(\bm \Sigma_{\bbm{\varepsilon},k}) \|\bm W_k \bm x_{1,k}\|_2^2/2]\right) \le 1.
            \end{split}
        \end{equation}
        Hence, for any $z_1>0$ and $z_2>0$, we have 
        \begin{equation}
        \begin{split}
            \mathbb{P}[\{S_k(\bm W_k)\geq z_1\}\cap\{R_k(\bm W_k)\leq z_2\}]
            \leq &\inf_{m>0} \exp(-m  z_1 + T_k^{-1} \sigma^2 \lambda_{\max}(\bm \Sigma_{\bbm{\varepsilon},k}) m^2z_2/2 )\\
            =& \exp\left(-\frac{z_1^2 T_k }{2\sigma^2 \lambda_{\max}(\bm \Sigma_{\bbm{\varepsilon},k}) z_2}\right).
        \end{split}
        \end{equation}

        By Lemma \ref{lemma:RSCRSS}, when $T_k\gtrsim  d\max ( \widehat{\kappa}_k^{2}\sigma^{4}, \widetilde{\kappa}_k \sigma^{2} )  $, then for any $k$, with probability at least $1-2\exp(-C  T_k \min (  \widetilde{\kappa}_k^{-2}\sigma^{-4},  \widetilde{\kappa}_k^{-1} \sigma^{-2}))$, we have $R_k(\bm W_k)\leq\frac{1}{2}L_k.$
        Note that $ \left\langle  \nabla \mathcal{L}_k(\bbm \alpha_k^*),\bm W_k\right\rangle  = S_k(\bm W_k)$. Therefore, for any $t>0$, 
        \begin{equation}
            \begin{split}
                & \mathbb{P}\left[\sup_{\bm W_k\in\bm \Theta(\epsilon,d,q)} \left\langle  \nabla \mathcal{L}_k(\bbm \alpha_k^*),\bm W_k\right\rangle \geq t\right]\\
                \leq & \mathbb{P}\left[\max_{\overline{\bm W}_k\in\overline{\bm \Theta}} \left\langle  \nabla \mathcal{L}_k(\bbm \alpha_k^*),\bm W_k\right\rangle\geq (1-\sqrt{q}\epsilon)t\right]\\
                \leq & \mathbb{P}\left[\{\max_{\overline{\bm W}_k\in\overline{\bm \Theta}} \left\langle  \nabla \mathcal{L}_k(\bbm \alpha_k^*),\bm W_k\right\rangle\geq (1-\sqrt{q}\epsilon)t\} \cup \{R_k(\overline{\bm W}_k) \le  L_k\}\right]
                + \mathbb{P}\left[ \sup_{\bm W_k\in\bm \Theta(\epsilon,d,q)}R_k(\bm W_k) >  L_k \right]\\
                \leq & |\overline{\bm \Theta}|\cdot\Bigg\{\mathbb{P}[\{ S_k(\bm W_k)\geq (1-\sqrt{q}\epsilon)t\}\cap\{R_k(\bm W_k)\leq  L_k \}]\Bigg\}  + \mathbb{P}\left[ \sup_{\bm W_k\in\bm \Theta(\epsilon,d,q)}R_k(\bm W_k) >  L_k \right]\\
                \leq & |\overline{\bm \Theta}|\cdot\Bigg\{\exp\left[-\frac{(1-\sqrt{q}\epsilon)^2T_k t^2}{\sigma^2 \lambda_{\max}(\bm \Sigma_{\bbm{\varepsilon},k})L_k}\right] \Bigg\}
                + 2\exp\left(-C  T_k \min (  \widetilde{\kappa}_k^{-2}\sigma^{-4},  \widetilde{\kappa}_k^{-1} \sigma^{-2})\right).
            \end{split}
        \end{equation}
        In addition, $|\overline{\bm \Theta}|\leq (C/\epsilon)^{d}$. 
        %By Lemma \ref{lemma:tensor_covering}, $|\overline{\bm \Theta}|\leq (24/\epsilon)^{d }$. 
        Thus, if we take $\epsilon=(2\sqrt{q})^{-1}$ and $t=C\sigma M_{k} \sqrt{d/T_k}$ with 
        $M_{k} = \lambda_{\max}(\bm \Sigma_{\bbm{\varepsilon},k})/\mu_{\min}^{1/2}(\bm \Xi_k)$, when $T_k\gtrsim d\max ( \widetilde{\kappa}_k^{2}\sigma^{4}, \widetilde{\kappa}_k \sigma^{2} ) $, then we have
        \begin{equation}
            \mathbb{P}\left[\sup_{\bm W_k\in\bm \Theta} S_k(\bm W_k)\gtrsim\sigma M_{k}\sqrt{\frac{d}{T_k}}\right]\leq \exp(-C_1d)+2\exp(-C_2  T_k \min (  \widetilde{\kappa}_k^{-2}\sigma^{-4},  \widetilde{\kappa}_k^{-1} \sigma^{-2})).
         \end{equation}

         The results for all $k\in [K]$ can be similarly proved by union bound as in the proof of Lemma \ref{lemma:RSCRSS}.
         This completes the proof.
\end{proof}

\begin{proof}[\bf Proof of Lemma \ref{lemma:uniondeviation}]
    
    This proof is similar to that of Lemma \ref{lemma:deviation}. 
    Recall that $d_M = N(s_1+s_2)+s_3p+Ks_1s_2s_3$,    
    and $\cm W = [\cm W_1: \cdots : \cm W_K]$ is of Tucker ranks $(2s_1,2s_2,2s_3,K)$ with each $\cm W_k\in\mathbb{R}^{N\times N\times p}$ such that $\sum_{k\in [K]}\|\cm W_k\|_\F^2=1$. In order to borrow the RSS condition bulit in Lemma \ref{lemma:RSCRSS}, we first show that $\cm W \in \bm \Theta(\epsilon,d_M,8)$, where $\Theta(\epsilon,d_M,8)$ is defined in Definition \ref{definition:matrix_set}.
            
    Consider an $\epsilon$-net $\overline{\bm \Theta}$ for $\bm \Theta(\epsilon,d_M,8)$. For any $\cm W\in \bm\Theta(\epsilon,d_M,8)$, there exists a matrix $\overline{\cm W} \in \overline{\bm \Theta}$ such that 
        \begin{itemize}
        	\item $\|\cm W-\overline{\cm W}\|_\F \leq \epsilon$, and by Lemma \ref{lemma:tensor_covering}, $|\overline{\bm \Theta} |\leq (12/\epsilon)^{d_M}$;
        	\item $\cm N=\cm W-\overline{\cm W}$ is a tensor of Tucker ranks at most $(4s_1, 4s_2,4s_3, K)$. Based on the HOSVD of $\cm N$, we can split $\cm N$ into 8 tensors, i.e., $\cm N=\sum_{i=1}^8 \cm N_i$, where each $\cm N_i$ is a tensor of Tucker ranks $(2s_1, 2s_2, 2s_3, K)$ and $\langle\cm N_i,\cm N_j\rangle=0$ for all $i \neq j$. Thus, $\|\cm N\|_\F^2=\sum_{i=1}^8\|\cm N_i\|_\F^2$.
        \end{itemize} 
        
        By Cauchy-Schwarz inequality, we have $\sum_{i=1}^8\|\cm N_i\|_\F \leq 2\sqrt{2}\epsilon$. Moreover, since $\cm N_i/\|\cm N_i\|_\F\in\bm \Theta(\epsilon,d_M,8)$,
        \begin{equation}
            \begin{split}
                &\sum_{k\in [K]} w_k^{1/2}\left\langle  \nabla \mathcal{L}_k(\cm{A}_k^*),\cm W_k \right\rangle \\
                \leq& \sum_{k\in [K]} w_k^{1/2}\left\langle  \nabla \mathcal{L}_k(\cm{A}_k^*),\overline{\cm W}_k\right\rangle + \sum_{i=1}^8\sum_{k\in [K]} w_k^{1/2}\left\langle  \nabla \mathcal{L}_k(\cm{A}_k^*),\cm N_{i,k}/\|\cm N_i\|_\F\right\rangle\|\cm N_i\|_\F\\
                \leq& \sum_{k\in [K]} w_k^{1/2}\left\langle  \nabla \mathcal{L}_k(\cm{A}_k^*),\overline{\cm W}_k\right\rangle + 2\sqrt{2}\epsilon\sup_{\cm W\in\bm \Theta(\epsilon,d_M, 8)}\sum_{k\in [K]} w_k^{1/2}\left\langle  \nabla \mathcal{L}_k(\cm{A}_k^*),\cm W_k\right\rangle,
             \end{split}
        \end{equation}
        which implies that
        \begin{equation}
            \sup_{\cm W\in\bm \Theta(\epsilon,d_M, 8)}\sum_{k\in [K]} w_k^{1/2}\left\langle  \nabla \mathcal{L}_k(\cm{A}_k^*),\cm W_k\right\rangle \leq (1-2\sqrt{2}\epsilon)^{-1}\max_{\overline{\cm W}\in\overline{\bm \Theta}}\sum_{k\in [K]} w_k^{1/2}\left\langle  \nabla \mathcal{L}_k(\cm{A}_k^*),\cm W_k\right\rangle.
        \end{equation}
    
        Next, for any fixed $\cm W$ such that $\sum_{k\in [K]}  \|\cm W_k\|_\F^2=1$, we denote 
        $$R(\cm W)=\sum_{k\in [K]} w_k R_k(\cm W_k) \quad \text{and} \quad S(\cm W)=\sum_{k\in [K]} w_k^{1/2} S_k(\cm W_k),$$
        with $R_k(\cm W_k)= T_k^{-1}\sum_{t=1}^{T_k}\|\cm W_{k(1)}\bm{x}_{t,k}\|_2^2$ and $S_k(\cm W_k)=T_k^{-1}\sum_{t=1}^{T_k}\langle\bbm{\varepsilon}_{t,k},\cm W_{k(1)}\bm{x}_{t,k}\rangle$, for $1\leq k\leq K$. Similar to Lemma S5 of \citet{wang2022high}, by the standard Chernoff bound and independence of $\bbm \varepsilon_{t,k}$ across $k\in [K]$ in Assumption \ref{asmp:2}, for any $z_1>0$ and $z_2>0$,
        \begin{equation}
        \begin{split}\label{eq:uniondeviation1}
            &\mathbb{P}[\{S(\cm W)\geq z_1\}\cap\{\max_{k\in[K]}\{\|\cm W_k\|_\F^{-2}R_k(\cm W_k)\}\leq z_2\}]\\
             = & \inf_{m>0} \mathbb{P}[\{\exp(m S(\cm W))\geq \exp(m  z_1)\}\cap\{\max_{k\in[K]}\{\|\cm W_k\|_\F^{-2}R_k(\cm W_k)\}\leq z_2\}]\\
            \leq& \inf_{m>0} \exp(-m  z_1) \mathbb{E}\left(\exp[m S(\cm W)]\mathbb{I}\{\max_{k\in[K]}\{\|\cm W_k\|_\F^{-2}R_k(\cm W_k)\}\leq z_2\}\right)\\
            \leq & \inf_{m>0} \exp(-m  z_1)  \prod_{k=1}^K \mathbb{E}[\exp(m w_k^{1/2}S_k(\cm W_k))\mathbb{I}\{R_k(\cm W_k)\leq z_2\|\cm W_k\|_\F^2\}]\\
            \leq& \inf_{m>0} \exp(-m z_1 + m^2 z_2\sum_{k\in [K]}c_k\|\cm W_k\|_\F^2)\\
            &\cdot \prod_{k=1}^K \mathbb{E}\left(\exp[m w_k^{1/2} S_k(\cm W_k) -  m^2 z_2 c_k\|\cm W_k\|_\F^2]\mathbb{I}\{R_k(\cm W_k)\leq z_2\|\cm W_k\|_\F^2\}\right)\\
            \leq& \inf_{m>0} \exp(-m z_1 +  m^2  z_2\sum_{k\in [K]} c_k\|\cm W_k\|_\F^2 ) \prod_{k=1}^K  \mathbb{E}\left(\exp[m w_k^{1/2} S_k(\cm W_k) - m^2 c_k R_k(\cm W_k)]\right).
        \end{split}
        \end{equation}
        By the tower rule and \eqref{eq:deviationMartingale} in Lemma \ref{lemma:deviation}, we have
        \begin{equation}\label{eq:uniondeviation2}
            \begin{split}
                &\mathbb{E}\left(\exp[m w_k^{1/2} S_k(\cm W_k) - c_km^2 R_k(\cm W_k)]\right) \\
                \le & \mathbb{E}\left(\exp\left[w_k^{1/2}T_k^{-1} m \sum_{t=1}^{T_k}\langle\bbm{\varepsilon}_{t,k}, \cm W_k\bm{x}_{t,k}\rangle - m^2 c_k  T_k^{-1} \sum_{t=1}^{T_k}\|\cm W_k \bm x_{t,k}\|_2^2 \right]\right)\\
                =&\mathbb{E}\left[\mathbb{E}\left(\exp\Big[w_k^{1/2}T_k^{-1} m  \sum_{t=1}^{T_k}\langle\bbm{\varepsilon}_{t,k},\cm W_k\bm{x}_{t,k}\rangle - m^2 c_k T_k^{-1} \sum_{t=1}^{T_k}\|\cm W_k \bm x_{t,k}\|_2^2  \Big]\right)\mid \mathcal{F}_{T_k-1,k}\right] \\
                =& \mathbb{E}\Bigg\{\exp\Big[ w_k^{1/2}T_k^{-1} m  \sum_{t=1}^{T_k-1}\langle\bbm{\varepsilon}_{t,k},\cm W_k\bm{x}_{t,k}\rangle - m^2 c_k T_k^{-1} \sum_{t=1}^{T_k-1}\|\cm W_k \bm x_{t,k}\|_2^2 \Big]\\
                & \quad \quad \quad \cdot \mathbb{E}\left(\exp\Big[w_k^{1/2}T_k^{-1} m \langle\bbm{\varepsilon}_{T_k,k},\cm W_k\bm{x}_{T_k,k}\rangle - m^2 c_k T_k^{-1} \|\cm W_k \bm x_{T_k,k}\|_2^2 \Big] \mid \mathcal{F}_{T_k-1,k}\right) \Bigg\}.
            \end{split}
        \end{equation}
        Note that $w_k^{1/2}T_k^{-1}   \langle\bbm{\varepsilon}_{T_k,k}, \cm W_k\bm{x}_{T_k,k}\rangle = w_k^{1/2}T_k^{-1}   \langle\bbm{\zeta}_{T_k,k}, \bm \Sigma_{\bbm{\varepsilon},k}^{1/2}\cm W_k\bm{x}_{T_k,k}\rangle$. Given $\mathcal{F}_{T_k,k}$, by Assumption \ref{asmp:2}, it can be easily verified that $ w_k^{1/2}T_k^{-1} \langle\bbm{\varepsilon}_{T_k,k}, \cm W_k\bm{x}_{T_k,k}\rangle$ is a $\sigma^2  w_k T_k^{-2}  \lambda_{\max}(\bm \Sigma_{\bbm{\varepsilon},k})\|\cm W_{k(1)} \bm x_{T_k,k}\|_2^2$-sub-Gaussian random variable, i.e., 
        \begin{align*}         
            \mathbb{E}\left(\exp\Big[w_k^{1/2}T_k^{-1}    \langle\bbm{\varepsilon}_{T_k,k},\cm W_k\bm{x}_{T_k,k}\rangle \Big] \mid \mathcal{F}_{T_k-1,k}\right) 
            &\le \exp(w_kT_k^{-2} \sigma^2m^2 \lambda_{\max}(\bm \Sigma_{\bbm{\varepsilon},k}) \|\cm W_{k(1)} \bm x_{T_k,k}\|_2^2/2 ).
        \end{align*}
        Thus, letting $c_k =\sigma^2 w_k T_k^{-1}\lambda_{\max}(\bm \Sigma_{\bbm{\varepsilon},k})/2$ for $k\in [K]$,
        %  with $\overline{\lambda}(\bm \Sigma_{\bbm{\varepsilon}}) = \max_{k\in [K]}\{\lambda_{\max}(\bm \Sigma_{\bbm{\varepsilon},k})\}$,
        we have
        \begin{equation}
            \begin{split}
                &\mathbb{E}\left(\exp\Big[m w_k^{1/2}S_k(\cm W_k) - m^2  \sigma^2  w_k T_k^{-2} \lambda_{\max}(\bm \Sigma_{\bbm{\varepsilon},k}) R_k(\cm W_k)/2\Big]\right) \\
                \le & \mathbb{E}\left(\exp\Big[m  w_k^{1/2}T_k^{-1} \sum_{t=1}^{T_k-1}\langle\bbm{\varepsilon}_{t,k},\cm W_{k(1)}\bm{x}_{t,k}\rangle - m^2 \sigma^2 w_k T_k^{-2} \lambda_{\max}(\bm \Sigma_{\bbm{\varepsilon},k})\sum_{t=1}^{T_k-1}\|\cm W_{k(1)} \bm x_{t,k}\|_2^2/2\Big]\right)\\
                \le& \cdots\leq\mathbb{E}\left(\exp\Big[m \langle\bbm{\varepsilon}_{1,k},\cm W_{k(1)}\bm{x}_{1,k}\rangle - m^2 \sigma^2 w_k T_k^{-2} \lambda_{\max}(\bm \Sigma_{\bbm{\varepsilon},k}) \|\cm W_{k(1)} \bm x_{1,k}\|_2^2/2\Big]\right) \le 1.
            \end{split}
        \end{equation}
        Hence, this together with \eqref{eq:uniondeviation1} and \eqref{eq:uniondeviation2}, for any $z_1>0$ and $z_2>0$, implies that
         \begin{equation}
            \begin{split}
             &\mathbb{P}[\{S(\cm W)\geq z_1\}\cap\{\max_{k\in[K]}\{\|\cm W_k\|_\F^{-2}R_k(\cm W_k)\}\leq z_2\}]\\
            \leq& \inf_{m>0} \exp\left(-m z_1 + 2^{-1}m^2\sigma^2  z_2  \sum_{k\in [K]} w_k T_k^{-1} \lambda_{\max}(\bm \Sigma_{\bbm{\varepsilon},k})\right)\\
            \leq & \exp\left(-\frac{z_1^2 }{2\sigma^2\sum_{k\in [K]} \|\cm W_k\|_\F^2 w_k T_k^{-1} \lambda_{\max}(\bm \Sigma_{\bbm{\varepsilon},k})}\right)\\
            \leq & \exp\left(-\frac{z_1^2 \widetilde{T}}{2\sigma^2  z_2}\right),
            \end{split}
         \end{equation}
         where $\widetilde{T} = \min_{k\in [K]}\{w_k^{-1} \lambda_{\max}^{-1}(\bm \Sigma_{\bbm{\varepsilon},k})T_k\}$.
        %  $\bar{\lambda}(\bm \Sigma_{\bbm{\varepsilon}}) = \max_{k\in [K]}\{\lambda_{\max}(\bm \Sigma_{\bbm{\varepsilon},k})\}$ and 
        % Hence,  we have
        % \begin{equation}
        % \begin{split}
        %     \mathbb{P}[\{S(\cm W)\geq z_1\}\cap\{R(\cm W)\leq z_2\}]
        %     \leq& \inf_{m>0} \exp(-m  z_1 +  \sigma^2 \bar{\lambda}(\bm \Sigma_{\bbm{\varepsilon}}) m^2z_2/2 )\\
        %     \leq& \exp\left(-\frac{z_1^2 }{2\sigma^2 \bar{\lambda}(\bm \Sigma_{\bbm{\varepsilon}}) z_2}\right).
        % \end{split}
        % \end{equation}
        
        Let $d = N(s_1+s_2) + p s_3 +s_1s_2s_3 $. By \eqref{eq:singleRSCRSS-2} in Lemma \ref{lemma:RSCRSS}, when $ T_{\min} \gtrsim [N(s_1 \vee s_2) +\log K] \max(\widetilde{\kappa}^{2}\sigma^{4} , \widetilde{\kappa}\sigma^{2} ) $, we have 
\begin{align}
       \mathbb{P}\Bigg[ \bigcup_{k\in [K]}\Bigg\{\sup_{\cm W\in\bm \Theta(\epsilon, d, 8)} R_k(\cm W_k) > L\Bigg\} \Bigg] 
       \leq 2\exp(-C T_{\min} \min (\widetilde{\kappa}^{-2}\sigma^{-4},  \widetilde{\kappa}^{-1} \sigma^{-2}))
\end{align}

        % By Lemma \ref{lemma:unionRSCRSS}, with probability at least $1-2\exp\left(-CT_{\min} \min(\widetilde{\kappa}^{-2}\sigma^{-4} ,\widetilde{\kappa}^{-1} \sigma^{-2} )\right)$,
        % \begin{equation}
        %     R(\cm{W})\leq\frac{1}{2}T L .
        % \end{equation}
        Therefore, for any $t>0$, 
        \begin{equation}
            \begin{split}
                & \mathbb{P}\Bigg[\sup_{\cm W\in\bm \Theta(\epsilon, d_M, 8)}\sum_{k\in [K]} w_k^{1/2}  \left\langle  \nabla \mathcal{L}_k(\cm{A}_k^*),\cm W_k\right\rangle \geq t\Bigg]\\
                \leq & \mathbb{P}\Bigg[\Big\{\max_{\overline{\cm W} \in\overline{\bm \Theta}} \sum_{k\in [K]} w_k^{1/2}\left\langle  \nabla \mathcal{L}_k(\cm{A}_k^*),\overline{\cm W}_k \right\rangle\geq (1-2\sqrt{2}\epsilon)t\Big\}\cap\{\max_{k\in [K]}\{\|\overline{\cm W}_k\|_\F^{-2}R(\overline{\cm W}_k)\}\leq L \}\Bigg] \\
                & + \mathbb{P}\Bigg[ \bigcup_{k\in [K]}\Bigg\{\sup_{\cm W\in\bm \Theta(\epsilon, d, 8)} R_k(\cm W_k) > L\Bigg\} \Bigg]\\
                \leq & |\overline{\bm \Theta}|\cdot\mathbb{P}\Bigg[\Big\{\sum_{k\in [K]}w_k^{1/2} \left\langle  \nabla \mathcal{L}_k(\cm{A}_k^*),\cm W_k\right\rangle\geq (1-2\sqrt{2}\epsilon)t\Big\} \cap \Big\{\max_{k\in [K]}\{\|\overline{\cm W}_k\|_\F^{-2}R(\overline{\cm W}_k)\}\leq L \Big\}\Bigg] \\
                 & + \mathbb{P}\Bigg[ \bigcup_{k\in [K]}\Bigg\{\sup_{\cm W\in\bm \Theta(\epsilon, d, 8)} R_k(\cm W_k) > L\Bigg\} \Bigg]\\
                % \leq & |\overline{\bm \Theta}|\cdot\Bigg\{\mathbb{P}[\{ S(\cm W)\geq (1-2\sqrt{2}\epsilon)Tt\}\cap\{R(\cm W)\leq  0.5 TL \}] \Bigg\}
                %  +\mathbb{P}\Bigg[ \sup_{\cm W\in\bm \Theta(\epsilon, d_M, 8)} R(\cm W) >  0.5 T L \Bigg]\\
                \leq & |\overline{\bm \Theta}|\cdot\Bigg\{
                \exp\left(-\frac{(1-2\sqrt{2}\epsilon)^2  \widetilde{T} t^2 }{2\sigma^2  L}\right)\Bigg\}
                + 2\exp\left(-CT_{\min} \min(\widetilde{\kappa}^{-2}\sigma^{-4} ,\widetilde{\kappa}^{-1} \sigma^{-2} )\right).
            \end{split}
        \end{equation}
        By Lemma \ref{lemma:tensor_covering}, $|\overline{\bm \Theta}|\leq (12/\epsilon)^{d_M}$. Thus, if we take $\epsilon=0.1$ and $t=C\sigma L^{1/2} \sqrt{d_M /\widetilde{T}}$ with $d_M = N(s_1+s_2)+s_3p+ K s_1s_2s_3$,
        when $ T_{\min} \gtrsim [N(s_1 \vee s_2)+\log K] \max(\widetilde{\kappa}^{2}\sigma^{4} ,\widetilde{\kappa}\sigma^{2} )$, then we have
        \begin{equation}\label{eq:stat_error_xi_1}
        \begin{split}
        	&\mathbb{P}\left[\sup_{\cm W\in\bm \Theta(\epsilon, d_M, 8)}\sum_{k\in [K]} w_k^{1/2} \left\langle  \nabla \mathcal{L}_k(\cm{A}_k^*),\cm W_k\right\rangle \gtrsim \sigma L^{1/2}\sqrt{\frac{d_M}{\widetilde{T}}}\right]\\
        	&\leq \exp(-C d_M)+2\exp\left(-CT_{\min} \min(\widetilde{\kappa}^{-2}\sigma^{-4} ,\widetilde{\kappa}^{-1} \sigma^{-2} )\right).
        \end{split}
        \end{equation}
       This completes the proof.
    \end{proof}

\section{Lemmas \ref{lem-inf-conv}--\ref{lemma:tensor_covering}}\label{sec:techlemmas}

Lemmas \ref{lem-inf-conv}--\ref{lemma:tensor_covering} present some basic results in the literature that are instrumental for the proofs in Section \ref{sec:mainproof}. %They are included for the sake of completeness and to provide a coherent and self-contained presentation.
Specifically, Lemmas~\ref{lem-inf-conv}--\ref{lem-frac} serve as key auxiliary results required for establishing Lemma~\ref{lemma:deterministic1}, which are adapted from the results developed in \citet{duan2023adaptive}.
Lemmas \ref{lemma:personalization}--\ref{lem-frac} follow directly from Theorem A.1, Lemma F.6 and F.7 of \cite{duan2023adaptive}, respectively, and thus we omit their proofs. 
Lemmas \ref{lemma:covering} and \ref{lemma:tensor_covering} include basic properties of covering numbers \citep{vershynin2018high,candes2011tight}, and they are used to establish the covering numbers of different sets in Claims \ref{claim:singledeviation}--\ref{claim:pooldeviation}.
Since Lemma \ref{lemma:covering} follows directly from Theorem 4.2.13 of \citet{vershynin2018high}, we omit its proof.
%As a result, we only provide proofs for Lemmas~\ref{lem-inf-conv} and \ref{lemma:tensor_covering}.  
Recall that 
\begin{align*}
	\bm \Omega(s_1, s_2, s_3; N, K, p) = \Big\{\bm W = [\bm w_1 \; \cdots \; \bm w_K]: &~\bm w_k = (\bm U\otimes \bm L \otimes \bm V) \bm d_k,\\
	&~\bm U \in \mathbb{O}^{N \times s_1}, \bm V \in \mathbb{O}^{N \times s_2}, \bm L \in \mathbb{O}^{p \times s_3},\\
	&~\bm d_k \in \mathbb{R}^{s_1 s_2}, \quad \|\bm W_k\|_\F = 1 \Big\}.
\end{align*}

\begin{lemma}\label{lem-inf-conv}
    Let $L:\mathbb{R}^{N^2p} \to \mathbb{R}$ be a convex function with Lipschitz continuous gradient, i.e., there exists $L>0$ such that for all $\bbm \alpha_1,\bbm \alpha_2 \in \mathbb{R}^{N^2p}$,
    $$\| \nabla \mathcal{L}(\bbm \alpha_1)- \nabla \mathcal{L}(\bbm \alpha_2)\| _2\le L\|\bbm \alpha_1 - \bbm \alpha_2\|_2.$$
    Let $\bm B^*\bm d^*$ be a fixed point and $\lambda> \| \nabla \mathcal{L}(\bm B^*\bm d^*)\|_2$. 
    Then, for any $\bm B \bm d$ satisfying $\|\bm B\bm d-\bm B^* \bm d^* \| \le L^{-1} (\lambda -\|\nabla \mathcal{L}(\bm B^* \bm d^*)\|_2)$, we have $\widetilde{\mathcal{L}} (\bm B \bm d) = \mathcal{L} (\bm B \bm d)$, and the unique minimizer is 
   
\begin{align*}
    \argmin_{\bbm \alpha \in \mathbb{R}^{N^2p}} \{ \mathcal{L}(\bbm \alpha) + \lambda \|\bm B \bm d - \bbm \alpha \|_2 \} = \bm B \bm d.
\end{align*}
\end{lemma}

\begin{proof}[\bf Proof of Lemma \ref{lem-inf-conv}]
For any $\bm B \bm d$ in the neighborhood $\|\bm B \bm d - \bm B^* \bm d^* \|_2 \leq L^{-1}( \lambda - \| \nabla \mathcal{L} (\bm B^* \bm d^*) \|_{2}  ) $, it follows that 
\begin{align}
\| \nabla \mathcal{L}(\bm B \bm d) \|_2 
& \leq \| \nabla \mathcal{L}(\bm B^* \bm d^*) \|_2 + \| \nabla \mathcal{L}(\bm B \bm d) - \nabla \mathcal{L} (\bm B^* \bm d^*) \|_2\\
&\leq \| \nabla \mathcal{L} (\bm B^* \bm d^*) \|_{2} + \mathcal{L} \| \bm B \bm d- \bm B^* \bm d^* \|_2 <\lambda.
\end{align}
By the first-order optimality condition of the $\widetilde{\mathcal{L}}$, $\bm B \bm d$ is the minimizer if 
$$0\in \nabla \mathcal{L}(\bm B\bm d) + \lambda \partial (\|\cdot\|_2)(\bm B \bm d - \bm B^* \bm d^*).$$
Since $\partial (\| \cdot \|_2) (0) = \{ \bm g \in \mathbb{R}^{N^2}: \| \bm g \|_{2} \leq 1 \}$ and 
$\| \nabla \mathcal{L}(\bm B \bm d) \|\le \lambda$, we have $- \nabla \mathcal{L}(\bm B \bm d) \in \partial (\| \cdot \|_2) (0)$, which implies that $\widetilde{\mathcal{L}} (\bm B \bm d) = \mathcal{L} (\bm B \bm d)$. 

To show uniqueness, for any $ \bbm \alpha \neq \bm B \bm d $,
$$
L(\bbm \alpha) + \lambda \|\bm B \bm d - \bbm \alpha \|_2 \geq L(\bm B \bm d) + \langle \nabla L(\bm B \bm d), \bbm \alpha - \bm B \bm d \rangle + \lambda \|\bm B \bm d - \bbm \alpha \|_2.
$$
By the Cauchy-Schwarz inequality,
$$
\langle \nabla L(\bm B \bm d), \bbm \alpha - \bm B \bm d \rangle \geq -\| \nabla L(\bm B \bm d) \|_2 \|\bm B \bm d - \bbm \alpha \|_2,
$$
which implies that
$$
L(\bbm \alpha) + \lambda \|\bm B \bm d - \bbm \alpha \|_2 \geq L(\bm B \bm d) + (\lambda - \| \nabla L(\bm B \bm d) \|_2) \|\bm B \bm d - \bbm \alpha \|_2 > L(\bm B \bm d).
$$
Hence, $ \bm B \bm d $ is the unique minimizer, which completes the proof.

\end{proof}

\begin{lemma}\label{lemma:personalization}
Let $\widetilde{\bbm \alpha}_k = \argmin_{\bbm \alpha} \mathcal{L}_k(\bbm \alpha)$ and $\widehat{\bbm \alpha}_k = \argmin_{\bbm \alpha} \mathcal{L}_k(\bbm \alpha) + \lambda_k \|\bbm \alpha - \widehat{\bm B} \widehat{\bm d}_k\|_2$. If $\mathcal{L}_k$ is $(\bbm \alpha_k^*, \rho,L,\|\nabla \mathcal{L}_k(\bbm \alpha_k^*)\|_2)$-regular and $0 \le \lambda_k < \rho M / 2$, then
$$
\|\widetilde{\bbm \alpha}_k - \bbm \alpha_k^*\|_2 \le \frac{\|\nabla \mathcal{L}_k(\bbm \alpha_k^*)\|_2}{\rho}
\quad \text{and} \quad
\|\widehat{\bbm \alpha}_k - \widetilde{\bbm \alpha}_k\|_2 \le \frac{\lambda_k}{\rho}.
$$
\end{lemma}

\begin{lemma}\label{lem:repre-lower-bound}

   Suppose there are $0 < \rho, L ,h < +\infty$, $0 < M \leq +\infty$ and $\{ \bm B^{*} \bm d^{*}_k \}_{k=1}^K \subseteq \mathbb{R}^{N^2}$ such that, for $\forall \bm B \bm d_k \in B(\bm B^{*} \bm d^{*}_k, M)$, 
   \begin{align*}
   	 \rho \bm I_{N^2p} \preceq  \nabla^2 \mathcal{L}_k(\bbm \alpha_k^*)   \preceq L\bm I_{N^2p}  
   	 \quad \text{and} \quad 
   	 \|\nabla \mathcal{L}_{k}(\bm B^* \bm d_k^*)\|_2 \le \eta_k
   \end{align*}
   hold for all $k \in [K] $. 
%   Let $\lambda_k = w_k^{1/2}\lambda$. 
   When 
   $3 \eta_k L / \rho + Lr_k <  \lambda_k < LM$
   for $k \in [K]$ and some $r_k > 0$, we have
   \begin{align}
   \widetilde{\mathcal{L}}_k(\bm B \bm d_k ) & =  \mathcal{L}_k (\bm B \bm d_k ) , \qquad \forall \bm B \bm d_k  \in B ( \bm B^{*} \bm d^{*}_k, 2 \eta_k / \rho + r_r) ; \label{eq:lem:repre-lower-bound-1}\\
   \| \widetilde{\bm B}\widetilde{\bm d}_k - \bm B^* \bm d_k^* \|_2 &\leq  \eta_k / \rho ;\label{eq:lem:repre-lower-bound-2} \\
   \sum_{k\in [K]}w_k \widetilde{\mathcal{L}}_k ( \bm B \bm d_k ) - \sum_{k\in [K]} w_k\widetilde{\mathcal{L}}_k ( \bm B^* \bm d_k^* ) & \geq \rho  \sum_{k\in [K]} w_k H_k \Big(
   (\| \bm B \bm d_k - \bm B \bm d_k^* \|_2 -  \eta_k/\rho
   )_{+} 
   \Big) 
   - \frac{ \sum_{k\in [K]}w_k \eta_k^2}{\rho},
   \end{align}
   where $\widetilde{\bm B}\widetilde{\bm d}_k = \argmin_{\bm B \bm d_k  \in \mathbb{R}^{N^2}} \widetilde{\mathcal{L}}_k(\bm B \bm d_k )$, $H_k(t) = t^2/2$ if $0\le t \le r_k$ and $H_k(t) = r_k(t-r_k/2)/2$ if $t > r_k$.
\end{lemma}

\begin{lemma}\label{lem-frac}
	Let $\bm A \in \mathbb{R}^{m\times n} \backslash \{ \bm{0} \}$, $\{ \bm z_{k} \}_{k=1}^K \subseteq \mathbb{R}^K$ and $\{ w_k \}_{k=1}^K \subseteq [0,+\infty)$. Suppose that $\max_{k \in [K]} \| \bm z_k \|_2 \leq \alpha_1$, $\sum_{k\in [K]} w_k \bm z_k \bm z_k^{\top} \succeq ( \alpha_2^2 \sum_{k\in [K]} w_k / n ) \bm I_K$ for some $0 < \alpha_2 \leq \alpha_1$ and $\sum_{k\in [K]} w_k > 0$. We have
	$$
	\sum_{ k \in [K]:~\| \bm A \bm z_k \|_2 > t \| \bm A \|_2} w_k \geq \frac{\alpha_2^2 / n - t^2}{\alpha_1^2 - t^2} \sum_{j = 1}^K w_j , \qquad \forall t \in [0, \alpha_2 / \sqrt{n}] .
	$$
\end{lemma}

\begin{lemma}\label{lemma:covering}
    
    Let $\mathcal{N}$ be an $\epsilon$-net of the unit sphere $\mathbb{S}^{p-1}$, where $\epsilon\in(0,1]$. Then,
    \begin{equation}
        |\mathcal{N}| \leq \left(\frac{3}{\epsilon}\right)^p.
    \end{equation}
    
\end{lemma}

\begin{lemma}\label{lemma:tensor_covering}
    Let $\overline{\bm \Omega} $ be an $\epsilon$-net of $\bm \Omega(s_1, s_2, s_3; N, K, p)$, where $\epsilon\in(0,1]$. Then
    \begin{equation}
        |\overline{\bm \Omega} | \leq \left(\frac{24}{\epsilon}\right)^{N(s_1+s_2)+s_1s_2K}.
    \end{equation}
\end{lemma}

\begin{proof}[\bf Proof of Lemma \ref{lemma:tensor_covering}]
    The proof of this lemma follows that of Lemma 3.1 in \citet{candes2011tight}. 
    For any $\bm W = (\bm U\otimes \bm L \otimes \bm V)\bm D$, where $\bm{U}\in\mathbb{O}^{N\times s_1},~\bm{V}\in\mathbb{O}^{N\times s_2}, ~\bm{L}\in\mathbb{O}^{p\times s_3}$ and $\bm D \in\mathbb{R}^{s_1 s_2 s_3 \times K}$, we construct an $\epsilon$-net for $\bm \Omega$ by covering the set of $\bm{U}$, $\bm{V}, \bm L$, and $\bm D$. 
    
    By Lemma \ref{lemma:covering}, we take $\overline{\bm D} \in \overline{\mathbb{D}}$ to be an $\epsilon/4$-net for $\bm D$ with covering number $|\overline{\mathbb{D}}|\leq (12/\epsilon)^{s_1s_2s_3K}$.
    Next, to cover $\mathbb{O}^{N\times s}$, we consider the $\|\cdot\|_{2,\infty}$ norm, defined as
    \begin{equation}
        \|\bm{X}\|_{2,\infty}=\max_{i}\|\bm{X}_i\|_2,
    \end{equation}
    where $\bm{X}_i$ is the $i$-th column of $\bm{X}$. Let $\mathbb{Q}^{N\times s}=\{\bm{X}\in\mathbb{R}^{N\times s}:\|\bm{X}\|_{2,\infty}\leq 1\}$. It can be easily shown that $\mathbb{O}^{N\times s}\subset\mathbb{Q}^{N\times s}$, and thus an $\epsilon/4$-net $\overline{\mathbb{O}}^{N\times s}$ for $\mathbb{O}^{N\times s}$ satisfies $|\overline{\mathbb{O}}^{N\times s}|\leq (12/\epsilon)^{Ns}$.
    
    Denote $\overline{\bm \Omega}=\{\overline{\bm{D}}\in\overline{\mathbb{D}},\overline{\bm{U}}\in\overline{\mathbb{O}}^{N\times s_1},\overline{\bm{V}}\in\overline{\mathbb{O}}^{N\times s_2},\overline{\bm{L}}\in\overline{\mathbb{O}}^{p\times s_3}\}$. We have
    \begin{equation}
        |\overline{\bm \Omega}| \leq |\overline{\mathbb{D}}|\times |\overline{\mathbb{O}}^{N\times s_1}| \times |\overline{\mathbb{O}}^{N\times s_2}|\times |\overline{\mathbb{O}}^{p\times s_3}|= \left(\frac{12}{\epsilon}\right)^{N(s_1+s_2)+s_3p+s_1s_2s_3K}.
    \end{equation}
    It suffices to show that for any $\bm W\in\bm \Omega(s_1,s_2,s_3;N,K,p)$, there exists a $\overline{\bm W}\in\overline{\bm \Omega}$ such that $\|\bm W - \overline{\bm W}\|_\F\leq \epsilon$.    
    For any fixed $\bm W\in\bm \Omega(s_1,s_2,s_3;N,K,p)$, decompose it as $\bm W=\bm D\times_1 \bm{U} \times_2 \bm{V}\times_3 \bm{L}$. Then, there exists $\overline{\bm W}=\overline{\bm D} \times_1\overline{\bm{U}}\times_2\overline{\bm{V}}\times_3\overline{\bm{L}}$ satisfying that $\|\overline{\bm{U}}-\bm{U}\|_{2,\infty}\leq \epsilon/4$, $\|\overline{\bm{V}}-\bm{V}\|_{2,\infty}\leq \epsilon/4$, $\|\overline{\bm{L}}-\bm{L}\|_{2,\infty}\leq \epsilon/4$ and $\|\overline{\bm D}-\bm D\|_\F\leq \epsilon/4$.
    Then we have
    \begin{equation}
        \begin{split}
            &\|\bm W-\overline{\bm W}\|_\F\\
            \leq & \|(\bm D -\overline{\bm D}) \times_1 \bm{U} \times_2 \bm{V}\times_3 \bm{L}\|_\F 
            + \|\overline{\bm D} \times_1 (\bm{U} - \overline{\bm{U}}) \times_2 \bm{V}\times_3 \bm{L}\|_\F\\
            &+ \|\overline{\bm D} \times_1  \overline{\bm{U}} \times_2 (\bm{U} - \overline{\bm{V}})\times_3 \bm{L}\|_\F
            + \|\overline{\bm D} \times_1  \overline{\bm{U}} \times_2  \overline{\bm{V}}\times_3 (\bm{L} - \overline{\bm L})\|_\F\\
            \leq & \|\bm D -\overline{\bm D}\|_\F\cdot\|\overline{\bm{U}}\|_\op\cdot\|\bm{V}\|_\op \cdot\|\bm{L}\|_\op 
            +\|\overline{\bm D}\|_\F\cdot\|\bm{U}-\overline{\bm{U}}\|_{2,\infty} \cdot\|\bm{V}\|_\op \cdot\|\bm{L}\|_\op 
             \\
             &+ \|\overline{\bm D}\|_\F \|\overline{\bm{U}}\|_\op \cdot \|\bm{V}-\overline{\bm{V}}\|_{2,\infty}\cdot\|\bm{L}\|_\op 
             + \|\overline{\bm D}\|_\F \|\overline{\bm{U}}\|_\op \cdot\|\overline{\bm{V}}\|_\op \cdot \|\bm{L}-\overline{\bm{L}}\|_{2,\infty} \\
            \leq & \frac{\epsilon}{4}+\frac{\epsilon}{4}+\frac{\epsilon}{4}+\frac{\epsilon}{4}=\epsilon.
        \end{split}
    \end{equation}
    The proof is complete. 
\end{proof} 

\section{Additional Results for Empirical Analysis}
\label{sec:additional-empirical}

%This section provides additional results for the empirical analysis in Section \ref{sec:RealData}. 

\subsection{Data preprocessing}

For the dataset in Section \ref{sec:RealData} of the manuscript,  
we first truncate the data such that the 20 variables have uniformly aligned time period for each country, and a preprocessing procedure is then conducted: a possible seasonal adjustment, a possible transformation to remove non-stationarity, and finally standardizing each resulting sequence with mean zero and variance one. 
Table~\ref{tab:realdata} provides economic category for each variable and details for data preprocessing. 

After data preprocessing, the target series (Japan) has 87 quarterly observations, whereas the source series have lengths ranging from 91 (Sweden) to 216 (United States). The starting and ending quarters of the aligned data across ten countries are illustrated in Figure~\ref{fig:country_indicator_length2}.

%We first truncate the $20$ quarterly macroeconomic series so that all variables share a uniformly aligned time span within each country. Each variable is then preprocessed to ensure comparability and stationarity. In particular, we (i) remove seasonal components when necessary and (ii) apply the appropriate transformation (logarithmic, first difference, or second difference) to achieve stationarity; see Table~\ref{tab:realdata} for the variable-level preprocessing codes. Finally, each transformed series is standardized to have zero mean and unit variance.
%Table~\ref{tab:realdata} also reports economic category for each variable.

\subsection{Additional results for projection}

In the manuscript, we only report the projection matrices of $\widetilde{\bm U}_k$'s and their residual counterparts.
This section illustrates the residual projection matrices of $\widetilde{\bm V}_k$'s and $\widetilde{\bm L}_k$'s in Figure~\ref{fig:proj-VL}, which indicates that their residual projection matrices are nearly zero for all countries. This together with Figure \ref{fig:country_indicator_length} in the manuscript, suggests that these task-specific representations $\widetilde{\bm U}_k$'s, $\widetilde{\bm V}_k$'s and $\widetilde{\bm L}_k$'s share common spaces with small rank, respectively.

\newpage

\begin{table}[htbp]
	\centering
	\caption{Categories of twenty quarterly macroeconomic variables and data preprocessing. The category code (C): 1 = GDP and production, 2 = labor market, 3 = interest rates and money supply, 4 = prices and inflation, 5 = housing and real estate, and 6 = others. The seasonal adjustment code (S) indicates whether a variable is seasonally adjusted: 0 = not adjusted, and 1 = seasonally adjusted. All variables are transformed to be stationary using the transformation code (T): 1 = first difference, 2 = second difference, 3 = first difference of log-transformed series, and 4 = second difference of log-transformed series.}
	\label{tab:realdata}
	\resizebox{\textwidth}{!}{
		\begin{tabular}{lrlrl}
			\toprule
			Abbreviation & \multicolumn{1}{l}{C} & \multicolumn{1}{c}{S} & \multicolumn{1}{l}{T} & Description \\
			\midrule
			GDP   & 1 & 1 & 3 & Real Gross Domestic Product \\
			EGS   & 1 & 1 & 3 & Exports of Goods and Services/International Trade: Total Exports of Goods \\
			IGS   & 1 & 1 & 3 & Imports of Goods and Services/International Trade: Total Imports of Goods \\
			RPT   & 1 & 1 & 3 & Total Production (All Sectors) \\
			RPM   & 1 & 1 & 3 & Manufacturing Production \\
			UNE   & 2 & 1 & 1 & Unemployment Rate \\
			EMP   & 2 & 1 & 3 & Employment Level \\
			HE    & 2 & 1 & 4 & Hourly Earnings (Index = 2015) \\
			ULC   & 2 & 1 & 4 & Unit Labour Cost (Index = 2015) \\
			LTIR  & 3 & 0 & 1 & Interest Rates: 10-Year Government Bond Yields \\
			STIR  & 3 & 0 & 1 & Interest Rates: 3-Month or 90-Day Interbank Rates and Yields \\
			M1    & 3 & 1 & 4 & Moneystock: M1 \\
			M3    & 3 & 1 & 4 & Moneystock: M3 \\
			CPI:T & 4 & 0 & 4 & Consumer Price Index: Total (Index = 2015) \\
			CPI:E & 4 & 0 & 4 & Consumer Price Index: Energy (Index = 2015) \\
			CPI:NFNE & 4 & 0 & 4 & Consumer Price Index: All items no food no energy (Index = 2015) \\
			HPI   & 5 & 0 & 4 & Real House Price Indices (Index = 2015) \\
			PIR   & 5 & 0 & 2 & Price to Income Ratio (Index = 2015) \\
			EXCH  & 6 & 0 & 4 & Nominal Effective Exchange Rate \\
			SP    & 6 & 0 & 3 & Share Price Index \\
			\bottomrule
	\end{tabular}}
\end{table}

\begin{figure}
	\centering
	\includegraphics[width=0.8\linewidth]{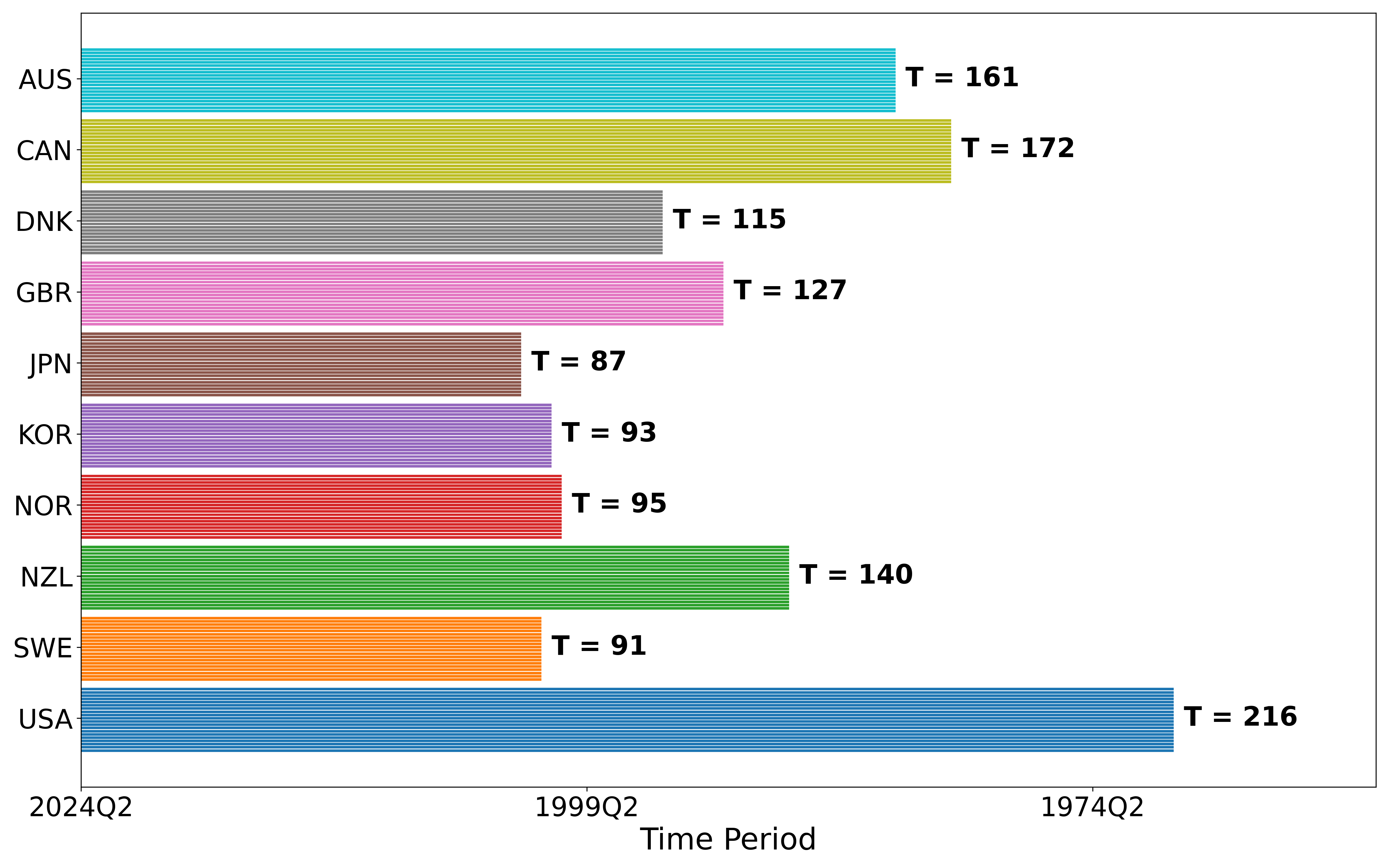}
	\caption{Ending and starting points for quarterly data of 20 macroeconomic variables across 10 countries after truncation.}
	\label{fig:country_indicator_length2}
\end{figure}

\begin{figure}[htbp]
	\centering
	\includegraphics[width=0.8\textwidth]{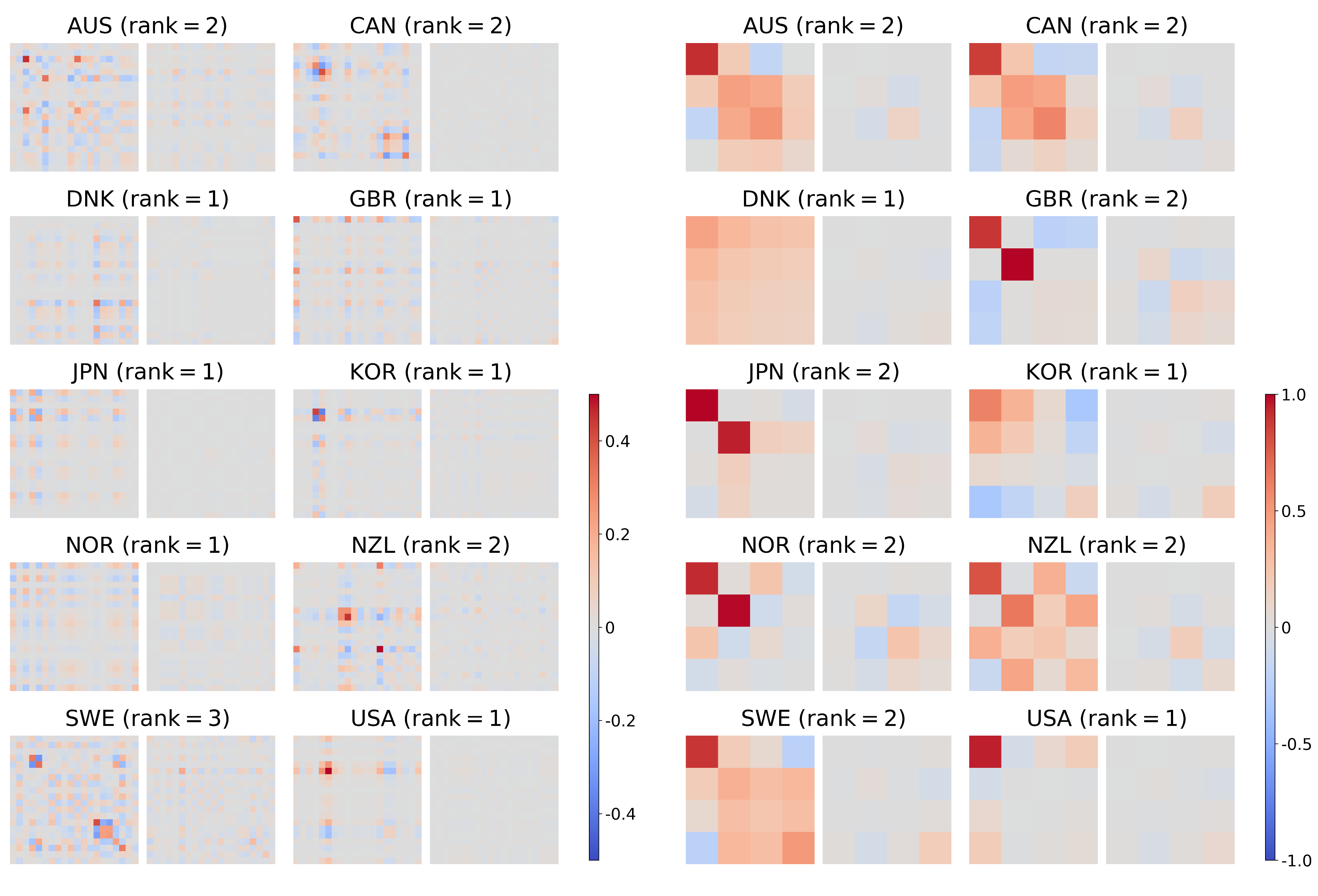}
	\caption{Heatmaps for the projection matrices of predictor factor spaces (left) and temporal factor spaces (right), before and after removing the estimated common subspace. The common subspace dimensions are selected via the elbow method, with rank $s_2 = 5$ for the predictor space and rank $s_3 = 2$ for the temporal space.}
	\label{fig:proj-VL}
\end{figure}

\newpage

\linespread{1.54}
\selectfont{}

\setlength{\bibsep}{1pt}
\bibliography{mybib}

\begin{thebibliography}{}

\bibitem[Bahadori et~al., 2014]{bahadori2014fast}
Bahadori, M.~T., Yu, Q.~R., and Liu, Y. (2014).
\newblock Fast multivariate spatio-temporal analysis via low rank tensor learning.
\newblock {\em Advances in neural information processing systems}, 27.

\bibitem[Bai and Wang, 2016]{BW16}
Bai, J. and Wang, P. (2016).
\newblock Econometric analysis of large factor models.
\newblock {\em Annual Review of Economics}, 8:53--80.

\bibitem[Basu et~al., 2019]{basu2019low}
Basu, S., Li, X., and Michailidis, G. (2019).
\newblock Low rank and structured modeling of high-dimensional vector autoregressions.
\newblock {\em IEEE Transactions on Signal Processing}, 67:1207--1222.

\bibitem[Basu and Michailidis, 2015]{basu2015regularized}
Basu, S. and Michailidis, G. (2015).
\newblock Regularized estimation in sparse high-dimensional time series models.
\newblock {\em Annals of Statistics}, 43:1535--1567.

\bibitem[Beck and Teboulle, 2009]{beck2009fast}
Beck, A. and Teboulle, M. (2009).
\newblock A fast iterative shrinkage-thresholding algorithm for linear inverse problems.
\newblock {\em SIAM journal on imaging sciences}, 2:183--202.

\bibitem[Cand\`{e}s and Plan, 2011]{candes2011tight}
Cand\`{e}s, E.~J. and Plan, Y. (2011).
\newblock Tight oracle inequalities for low-rank matrix recovery from a minimal number of noisy random measurements.
\newblock {\em IEEE Transactions on Information Theory}, 57:2342--2359.

\bibitem[Chen et~al., 2025]{chen2025distributed}
Chen, E., Chen, X., Jing, W., and Zhang, Y. (2025).
\newblock Distributed tensor principal component analysis with data heterogeneity.
\newblock {\em Journal of the American Statistical Association}, pages 1--23.

\bibitem[Chua et~al., 2021]{chua2021similarRep}
Chua, K., Lei, Q., and Lee, J.~D. (2021).
\newblock How fine-tuning allows for effective meta-learning.
\newblock In {\em NeurIPS 2021}, volume~34, pages 8871--8884. Curran Associates, Inc.

\bibitem[Davis et~al., 2016]{davis2016sparse}
Davis, R.~A., Zang, P., and Zheng, T. (2016).
\newblock Sparse vector autoregressive modeling.
\newblock {\em Journal of Computational and Graphical Statistics}, 25:1077--1096.

\bibitem[De~Lathauwer et~al., 2000]{de2000multilinear}
De~Lathauwer, L., De~Moor, B., and Vandewalle, J. (2000).
\newblock A multilinear singular value decomposition.
\newblock {\em SIAM Journal on Matrix Analysis and Applications}, 21:1253--1278.

\bibitem[Dowell and Pinson, 2016]{DP16}
Dowell, J. and Pinson, P. (2016).
\newblock Very-short-term probabilistic wind power forecasts by sparse vector autoregression.
\newblock {\em IEEE Transactions on Smart Grid}, 7:763--770.

\bibitem[Du et~al., 2020]{du2020few}
Du, S.~S., Hu, W., Kakade, S.~M., Lee, J.~D., and Lei, Q. (2020).
\newblock Few-shot learning via learning the representation, provably.
\newblock {\em arXiv preprint arXiv:2002.09434}.

\bibitem[Duan and Wang, 2023]{duan2023adaptive}
Duan, Y. and Wang, K. (2023).
\newblock Adaptive and robust multi-task learning.
\newblock {\em Annals of Statistics}, 51:2015--2039.

\bibitem[Gao and Tsay, 2022]{Gao_Tsay2022}
Gao, Z. and Tsay, R.~S. (2022).
\newblock Modeling high-dimensional time series: a factor model with dynamically dependent factors and diverging eigenvalues.
\newblock {\em Journal of the American Statistical Association}, 117:1398–1414.

\bibitem[Gu et~al., 2024]{Gu2022TLangleLM}
Gu, T., Han, Y., and Duan, R. (2024).
\newblock Robust angle-based transfer learning in high dimensions.
\newblock {\em Journal of the Royal Statistical Society, Series B}, page qkae111.

\bibitem[Han et~al., 2015]{han2015direct}
Han, F., Lu, H., and Liu, H. (2015).
\newblock A direct estimation of high dimensional stationary vector autoregressions.
\newblock {\em Journal of Machine Learning Research}, 16:3115--3150.

\bibitem[Han et~al., 2021]{han2021optimal}
Han, R., Willett, R., and Zhang, A. (2021).
\newblock An optimal statistical and computational framework for generalized tensor estimation.
\newblock {\em Annals of Statistics}, 50:1--29.

\bibitem[Huang et~al., 2025]{HUANG2025}
Huang, F., Lu, K., Zheng, Y., and Li, G. (2025).
\newblock Supervised factor modeling for high-dimensional linear time series.
\newblock {\em Journal of Econometrics}, 249:105995.

\bibitem[Kock and Callot, 2015]{kock2015oracle}
Kock, A.~B. and Callot, L. (2015).
\newblock Oracle inequalities for high dimensional vector autoregressions.
\newblock {\em Journal of Econometrics}, 186:325--344.

\bibitem[Kolda and Bader, 2009]{kolda2009tensor}
Kolda, T.~G. and Bader, B.~W. (2009).
\newblock Tensor decompositions and applications.
\newblock {\em SIAM Review}, 51:455--500.

\bibitem[Koop, 2013]{koop2013forecasting}
Koop, G.~M. (2013).
\newblock Forecasting with medium and large {Bayesian VARs}.
\newblock {\em Journal of Applied Econometrics}, 28:177--203.

\bibitem[Lam and Yao, 2012]{lam2012factor}
Lam, C. and Yao, Q. (2012).
\newblock Factor modeling for high-dimensional time series: Inference for the number of factors.
\newblock {\em Annals of Statistics}, 40:694--726.

\bibitem[Li et~al., 2020]{Li2020TLLM}
Li, S., Cai, T.~T., and Li, H. (2020).
\newblock Transfer learning for high‐dimensional linear regression: Prediction, estimation and minimax optimality.
\newblock {\em Journal of the Royal Statistical Society, Series B}, 84:149 -- 173.

\bibitem[Li et~al., 2023]{Li2023TLGLM}
Li, S., Zhang, L., Cai, T.~T., and Li, H. (2023).
\newblock Estimation and inference for high-dimensional generalized linear models with knowledge transfer.
\newblock {\em Journal of the American Statistical Association}, 119:1274 -- 1285.

\bibitem[Lozano et~al., 2009]{LZLR09}
Lozano, A.~C., Abe, N., Liu, Y., and Rosset, S. (2009).
\newblock Grouped graphical \textsc{G}ranger modeling for gene expression regulatory networks discovery.
\newblock {\em Bioinformatics}, 25:i110--i118.

\bibitem[Ma and Safikhani, 2025]{ma2025transfer}
Ma, M. and Safikhani, A. (2025).
\newblock Transfer learning for high-dimensional reduced rank time series models.
\newblock In {\em The 28th International Conference on Artificial Intelligence and Statistics}.

\bibitem[Negahban and Wainwright, 2011]{negahban2011estimation}
Negahban, S. and Wainwright, M.~J. (2011).
\newblock Estimation of (near) low-rank matrices with noise and high-dimensional scaling.
\newblock {\em Annals of Statistics}, 39:1069--1097.

\bibitem[Nicholson et~al., 2020]{nicholson2020high}
Nicholson, W.~B., Wilms, I., Bien, J., and Matteson, D.~S. (2020).
\newblock High dimensional forecasting via interpretable vector autoregression.
\newblock {\em Journal of Machine Learning Research}, 21:1--52.

\bibitem[Pan and Yang, 2010]{Pan2010ASO}
Pan, S.~J. and Yang, Q. (2010).
\newblock A survey on transfer learning.
\newblock {\em IEEE Transactions on Knowledge and Data Engineering}, 22:1345--1359.

\bibitem[Parikh and Boyd, 2014]{Parikh2014prox}
Parikh, N. and Boyd, S. (2014).
\newblock {Proximal Algorithms}.
\newblock {\em Foundations and Trends® in Optimization}, 1(3):127--239.

\bibitem[Park et~al., 2025]{park2025transfer}
Park, S., Lee, E.~R., Kim, H., and Zhao, H. (2025).
\newblock Transfer learning under large-scale low-rank regression models.
\newblock {\em Journal of the American Statistical Association}, pages 1--25.

\bibitem[Stewart, 1998]{stewart1998perturbation}
Stewart, G.~W. (1998).
\newblock Perturbation theory for the singular value decomposition.

\bibitem[Stock and Watson, 2009]{stock2009forecasting}
Stock, J.~H. and Watson, M. (2009).
\newblock Forecasting in dynamic factor models subject to structural instability.
\newblock {\em The Methodology and Practice of Econometrics. A Festschrift in Honour of David F. Hendry}, 173:205.

\bibitem[Tian et~al., 2025]{tian2024RepLLM}
Tian, Y., Gu, Y., and Feng, Y. (2025).
\newblock {Learning from similar linear representations: Adaptivity, minimaxity, and robustness}.
\newblock {\em Journal of Machine Learning Research}, 26:1--125.

\bibitem[Tripuraneni et~al., 2021]{tripuraneni21linear}
Tripuraneni, N., Jin, C., and Jordan, M. (2021).
\newblock Provable meta-learning of linear representations.
\newblock In {\em Proceedings of the 38th International Conference on Machine Learning}, volume 139, pages 10434--10443. PMLR.

\bibitem[Tripuraneni et~al., 2020]{Tripuraneni2020diversity}
Tripuraneni, N., Jordan, M., and Jin, C. (2020).
\newblock On the theory of transfer learning: The importance of task diversity.
\newblock {\em Advances in neural information processing systems}, 33:7852--7862.

\bibitem[Tu et~al., 2016]{tu2016low}
Tu, S., Boczar, R., Simchowitz, M., Soltanolkotabi, M., and Recht, B. (2016).
\newblock Low-rank solutions of linear matrix equations via procrustes flow.
\newblock In {\em Proceedings of the International Conference on Machine Learning (ICML)}, pages 964--973.

\bibitem[Velu and Reinsel, 2013]{velu2013multivariate}
Velu, R. and Reinsel, G.~C. (2013).
\newblock {\em Multivariate reduced-rank regression: theory and applications}.
\newblock Springer Science \& Business Media.

\bibitem[Vershynin, 2018]{vershynin2018high}
Vershynin, R. (2018).
\newblock {\em High-dimensional probability: An introduction with applications in data science}, volume~47.
\newblock Cambridge university press.

\bibitem[Wang and Tsay, 2023]{wang2023robust}
Wang, D. and Tsay, R.~S. (2023).
\newblock {Rate-optimal robust estimation of high-dimensional vector autoregressive models}.
\newblock {\em Annals of Statistics}, 51:846 -- 877.

\bibitem[Wang et~al., 2024a]{wang2024common}
Wang, D., Zhang, X., Li, G., and Tsay, R. (2024a).
\newblock High-dimensional vector autoregression with common response and predictor factors.
\newblock \emph{arXiv preprint arXiv:2203.15170}.

\bibitem[Wang et~al., 2024b]{wang2024high}
Wang, D., Zheng, Y., and Li, G. (2024b).
\newblock High-dimensional low-rank tensor autoregressive time series modeling.
\newblock {\em Journal of Econometrics}, 238:105544.

\bibitem[Wang et~al., 2022]{wang2022high}
Wang, D., Zheng, Y., Lian, H., and Li, G. (2022).
\newblock High-dimensional vector autoregressive time series modeling via tensor decomposition.
\newblock {\em Journal of the American Statistical Association}, 117:1338--1356.

\bibitem[Wang et~al., 2017]{wang2017unified}
Wang, L., Zhang, X., and Gu, Q. (2017).
\newblock A unified computational and statistical framework for nonconvex low-rank matrix estimation.
\newblock In {\em Artificial Intelligence and Statistics}, pages 981--990.

\bibitem[Wedin, 1972]{wedin1972perturbation}
Wedin, P.-{\AA}. (1972).
\newblock Perturbation bounds in connection with singular value decomposition.
\newblock {\em BIT Numerical Mathematics}, 12:99--111.

\bibitem[Wilms et~al., 2023]{Wilms2023}
Wilms, I., Basu, S., Bien, J., and Matteson, D.~S. (2023).
\newblock Sparse identification and estimation of large-scale vector autoregressive moving averages.
\newblock {\em Journal of the American Statistical Association}, 118:571--582.

\bibitem[Xia et~al., 2015]{xia2015consistently}
Xia, Q., Xu, W., and Zhu, L. (2015).
\newblock Consistently determining the number of factors in multivariate volatility modelling.
\newblock {\em Statistica Sinica}, 25:1025--1044.

\bibitem[Xu et~al., 2018]{xu2018muscat}
Xu, J., Liu, X., Wilson, T., Tan, P.-N., Hatami, P., and Luo, L. (2018).
\newblock Muscat: Multi-scale spatio-temporal learning with application to climate modeling.
\newblock In {\em IJCAI}, pages 2912--2918.

\bibitem[Zeng et~al., 2024]{zeng2024TLsptialAR}
Zeng, H., Zhong, W., and Xu, X. (2024).
\newblock {Transfer learning for spatial autoregressive models with application to U.S. presidential election prediction}.
\newblock \emph{arXiv preprint arXiv:2405.15600}.

\bibitem[Zheng and Cheng, 2021]{zheng20}
Zheng, Y. and Cheng, G. (2021).
\newblock Finite time analysis of vector autoregressive models under linear restrictions.
\newblock {\em Biometrika}, 108:469--489.

\bibitem[Zhu et~al., 2025]{Zhu_Li_Zhang_Li2025}
Zhu, Q., Li, W., Zhang, W., and Li, G. (2025).
\newblock {Panel Quantile GARCH Models under Homogeneity}.
\newblock {\em Journal of Business \& Economic Statistics}, To appear.

\end{thebibliography}

\end{document}